\documentclass[11pt,draftclsnofoot,twoside,onecolumn,letter]{IEEEtran}
\usepackage{graphicx}
\usepackage{epstopdf}
\DeclareGraphicsExtensions{.eps,.pdf} 
\graphicspath{ {figures/} }
\usepackage{cite}
\usepackage{amssymb,amsmath,array}
\usepackage{subfigure}
\usepackage{amssymb,amsmath}
\usepackage{algorithm}
\usepackage{algorithmic}
\usepackage{color}
\usepackage{multirow}
\usepackage{multicol}

\usepackage{fancyhdr}
\fancypagestyle{mahmood}{%
	\fancyhf{} 
	
	\fancyhead[C]{\footnotesize This work has been accepted for publication in IEEE Transactions on Communications. Copyright may be transferred without notice, after which this version may no longer be accessible.}
}%

\makeatletter
\let\ps@IEEEtitlepagestyle\ps@mahmood
\makeatother

\usepackage{setspace}
\usepackage{caption}

\makeatletter
\newcommand\restartchapters{\par
  \setcounter{chapter}{0}%
  \setcounter{section}{0}%
  \gdef\@chapapp{\chaptername}%
  \gdef\thechapter{\@arabic\c@chapter}}
\makeatother

\newtheorem{theorem}{\bf {Theorem}}

\newtheorem{remark}{{\bf{Remark}}}
\newtheorem{definition}{\bf {Definition}}
\newtheorem{lemma}{\bf {Lemma}}

\newcommand{\tr}{{\mathrm{tr}}}

\newcommand{\tensor}[1]{\textsf{\bfseries{#1}}}

\newcommand{\st}{{\mathrm{s.t.}}}

\newcommand{\ULU}{\mathtt{U}_{\ell}^{\mathtt{u}}}
\newcommand{\ul}{\mathtt{u}}
\newcommand{\SI}{\mathtt{SI}}
\newcommand{\dl}{\mathtt{d}}
\newcommand{\DLU}{\mathtt{U}_{ik}^{\mathtt{d}}}
\newcommand{\DLUi}[1]{\mathtt{U}_{#1}^{\mathtt{d}}}
\newcommand{\bs}{\mathtt{bs}}

\newcommand{\ds}{\displaystyle}

\newcommand*{\hili}{\color{black}}
\newcommand*{\hilidr}{\color{black}}
\newcommand*{\hilidra}{\color{black}}

\usepackage{nccmath}

\usepackage{capt-of}
\usepackage{dblfloatfix}
\usepackage{varwidth}
\makeatletter
\g@addto@macro\normalsize{%
	\setlength\abovedisplayskip{0.2pt}
	\setlength\belowdisplayskip{0.2pt}
	\setlength\abovedisplayshortskip{0.2pt}
	\setlength\belowdisplayshortskip{0.2pt}
}
\makeatother

\usepackage{mathtools}

\newcommand{\subparagraph}{}

\usepackage{titlesec}
\titlespacing{\section}{0pt}{0.5pt}{0pt}

\makeatletter

\renewcommand{\IEEEbibitemsep}{0pt plus 0.5pt}
\makeatletter
\IEEEtriggercmd{\reset@font\normalfont\fontsize{8pt}{8.5pt}\selectfont}
\makeatother
\IEEEtriggeratref{1}

\begin{document}

\title{{\huge Joint Power Control and User Association for NOMA-Based Full-Duplex Systems}}
\author{
	\IEEEauthorblockN{ Hieu V. Nguyen, Van-Dinh Nguyen, Octavia A. Dobre, Diep N. Nguyen, \\ Eryk Dutkiewicz, and Oh-Soon Shin \vspace{-2pt}}\\
	\thanks{H. V. Nguyen, V.-D. Nguyen,  and O.-S. Shin are with the School of Electronic Engineering \& Department of ICMC Convergence Technology, Soongsil University, Seoul 06978, Korea (e-mail: \{hieuvnguyen, nguyenvandinh, osshin\}@ssu.ac.kr).}
	\thanks{O.~A.~Dobre is with the Faculty of Engineering and Applied Science, Memorial University,  Canada (e-mail: odobre@mun.ca).}
	\thanks{D. N. Nguyen and E. Dutkiewicz are with the Faculty of Engineering and Information Technology, University of Technology Sydney, Sydney, NSW
2007, Australia (e-mail: \{diep.nguyen, Eryk.Dutkiewicz\}@uts.edu.au).}
\thanks{Part of this paper was submitted to IEEE Global Communications Conference in Dec 2019.}
	}

\maketitle
\begin{abstract}\vspace{-5pt}
{{\hili This paper investigates the coexistence of non-orthogonal multiple access (NOMA) and full-duplex (FD) to improve both spectral efficiency (SE) and user fairness. In such a scenario, NOMA based on the successive interference cancellation technique is simultaneously applied to both uplink (UL) and downlink (DL) transmissions in an FD system.} We consider the problem of jointly optimizing user association (UA) and power control to maximize the overall SE, subject to user-specific quality-of-service and total transmit power constraints. To be spectrally-efficient, we introduce the tensor model to optimize UL users' decoding order and DL users' clustering, which results in a mixed-integer non-convex problem. {\hilidra For practically appealing applications, we first relax the binary variables and then propose two low-complexity designs. In the first design, the continuous relaxation problem is solved using the inner convex approximation framework. Next, we additionally introduce the penalty method to further accelerate the performance of the former design. For a benchmark, we develop an optimal solution based on brute-force search (BFS) over all possible cases of UAs.} It is demonstrated in numerical results that the proposed algorithms outperform the conventional FD-based schemes and its half-duplex counterpart, as well as yield data rates close to those obtained by BFS-based algorithm.}
\end{abstract}

\newpage
\section{Introduction} \label{Introduction} 
Multiple access techniques are crucial for next generation of mobile communications to meet the exponential demand of mobile data  and new services over limited radio spectrum\cite{Andrews-14-A, Yadav:IEEEWirelessComm:Aug2018}. Among them, by enabling multiple concurrent transmissions, non-orthogonal multiple access (NOMA) has recently been recognized as a  promising solution, due to its superior spectral efficiency (SE) and user fairness feature \cite{Ding:CommunMag:Feb2017,LiuPIEEE17,Islam:COMSurTutor:2017,Islam:IEEEWirelessComm:Apr2018}. The key idea of NOMA\footnote{Henceforth, power domain-based NOMA is simply referred to as NOMA.} is to concurrently allocate different portions of the total power for multiple users over the same spectrum. NOMA is particularly efficient while user equipment (UEs) simultaneously experience significantly different channel conditions \cite{DingJSAC17,DingTVT2016,ChenJSAC2017,LiuPIEEE17}. For example, in a typical scenario of two-user NOMA, the user with poorer channel condition is allocated much higher transmit power than that of the one with more favorable channel condition. Then, using the successive interference cancellation (SIC) technique, the latter is able to remove  the signal of the former before decoding its own. Thus, while the user with better channel condition benefits from removing the strong interference, the throughput of the user with worse channel condition is clearly improved, leading to a higher total throughput.

Also to improve the SE, in-band full-duplex (FD) radios that enable the downlink (DL) transmission and uplink (UL) reception at the same time-frequency resource, has received paramount interest. Theoretically, FD radios can double the SE of a wireless link over its half-duplex (HD) counterparts \cite{Wong5Gbook17, Yadav:IEEETVT:June2018}. To achieve such a potential gain, the self-interference (SI) due to concurrent transmission and reception at the FD device, has to be canceled/suppressed to the noise floor level. Although recent advances in  active and passive SI suppression (SiS) techniques have led to implementable FD systems  \cite{Bharadia13,Sabharwal:JSAC:Feb2014,Bharadia14},  there always exists a small residual SI, but not negligible, due to the imperfect SiS. FD systems with imperfect SiS have been widely studied in small-cell cellular setups \cite{Dan:TWC:14,Aquilina:TCOMM:2017,Dinh:Access,Yadav:Access, Dinh:JSAC:18,Tam:TCOM:16, Hieu:IEEETWC:June2019}. Besides SI, such FD systems also suffer from multiuser interference (MUI) and co-channel interference (CCI, caused by a UL user to DL users). Recently, the coexistence of NOMA and FD has been analytically and numerically investigated to effectively handle the network interference, which helps boost the  performance of FD systems \cite{Sun:TCOMM:Mar2017,  Ding:LWC:2018,Sun:TCOMM:2018,Dinh:JSAC:18}.

Although the SIC technique has been widely adopted for the UL reception in FD systems \cite{Dan:TWC:14,Dinh:Access,Yadav:Access, Dinh:JSAC:18,Tam:TCOM:16}, only random decoding order with respect to (w.r.t.) UL users' indices is considered.\footnote{It is worth mentioning that the UL users' decoding order has a strong impact on the SE, especially when taking into account the DL interference and quality-of-service constraints (see Fig. \ref{fig: Gain SumRate vs Pbs}).} Further, the FD-NOMA system in \cite{Sun:TCOMM:Mar2017} with a single-antenna BS can serve at most two DL and two UL users on a frequency resource. Ding \textit{et al.} \cite{Ding:LWC:2018}  mainly focused on the
comparison of the UL sum rate between FD and HD systems. Moreover, both \cite{Dinh:JSAC:18} and \cite{Hieu:ICTC:Oct2017}  considered the beamforming design only, and thus the maximum improvement in SE provided by the \textit{optimal} FD-NOMA systems compared to FD-NOMA ones is still unknown. {\hili This motivates us to devise a general model for DL user clustering and UL users' decoding order under power control and quality-of-service (QoS) requirements so that both total SE and user fairness are remarkably enhanced.}

\vspace{-18pt}
\subsection{Related Work}\vspace{-10pt}
On one hand, FD-based systems have been developed primarily for small-cell setups. The authors in  \cite{Dan:TWC:14} and \cite{Tam:TCOM:16} first studied the SE maximization problem for a small-cell FD system under the assumption of perfect channel state information (CSI), and the worst-case  robust design for FD  multi-cell system  was  considered in \cite{Aquilina:TCOMM:2017}. The application of FD radio to other designs is presently an emerging subject; e.g., the FD-based energy harvesting design \cite{Yadav:Access,ChaliseTCOM17} and the FD-based physical layer security (PLS)  design \cite{Dinh:JSAC:18}. However, the performance of these FD systems is very limited due to  severe network interference in small cell scenarios, which creates a fundamental bottleneck on the network interference management. Consequently, the conventional  techniques (i.e.,  interference-limited ones) are no longer applicable to attain the optimal performance, and thus, new design techniques for FD systems are required.

On the other hand, beamforming design for NOMA has been studied for different optimization targets.   Choi \textit{et al}. \cite{Choi15} considered a  power minimization problem for a two-user multiple-input single-output NOMA (MISO-NOMA) system, where a heuristic method was proposed for its solution. For a general problem of $2K$-UE MISO-NOMA,  zero-forcing  beamformer was adopted at the base station (BS) to cancel the inter-pair interference, resulting in $K$ independent subproblems. A closed-form solution for the power minimization problem of two-user MISO-NOMA to meet given QoS was obtained in \cite{ChenTSP16} and \cite{ChenAC16}.  In  \cite{HDRK16}, SIC was performed at users based on their channel gain differences to maximize the sum rate of a MISO-NOMA DL system. In these works, a joint  power and user clustering (i.e.,  users with distinct channel conditions are grouped to perform NOMA jointly) has not been reported yet. Although the works in \cite{DSP16} and \cite{Dinh:JSAC:Dec2017} proposed a two-zone pairing  that randomly pairs two UEs from different zones, the use of random user pairing scheme may cause significant performance loss, compared to the optimal one.  The authors in \cite{ChenJSAC2017} studied a multi-zone based clustering, where the BS randomly selects one UE from each zone to form a cluster, leading to a suboptimal solution. 

Regarding FD-NOMA, a joint power and subcarrier allocation scheme to enhance the throughput of users was investigated in \cite{Sun:TCOMM:Mar2017}. Tackling the channel uncertainty in the PLS was studied in \cite{Sun:TCOMM:2018}, showing that FD-NOMA is able to secure both DL and UL transmissions simultaneously and obtain a significant system secrecy rate improvement  compared with the traditional FD scheme. Further, the authors in \cite{Ding:LWC:2018} showed that FD-NOMA can improve the achievable rate, through both analysis and numerical simulations. Our approach of joint NOMA beamforming and user scheduling  in FD systems  was also reported in \cite{Hieu:ICTC:Oct2017}, that aims  to serve users in different time slots. Nevertheless, the user association (UA), i.e.,  UL users' decoding order and user clustering for DL users, is not fully exploited in the aforementioned works, leading to a suboptimal solution.

\vspace{-15pt}
\subsection{Main Contributions}\vspace{-5pt}
To that end, we formulate a novel optimization problem to maximize the total SE in FD-NOMA multiuser MISO (MU-MISO) systems, where each user is guaranteed a minimum data rate. Our formulation explicitly considers the effects of user association in both DL and UL channels. For UL reception, we adopt the SIC technique that results in a permutation problem to optimize UL users' decoding order. For DL transmission, it may not be realistic to require all users in the FD system to jointly perform NOMA. Therefore, a promising alternative is to divide  DL users into multiple clusters with different channel conditions by introducing a tensor of binary numbers, where NOMA is implemented within each cluster. The optimization problem of interest is a mixed-integer non-convex programming, which often requires exponential complexity to find its globally optimal solution. To tackle it, we propose novel transformations so that the popular solvers can be applied to address the problem efficiently. Our main contributions are summarized as follows:
\begin{itemize}
	\item {\hili Aiming at SE, we introduce new binary variables to establish BS-UE associations in both DL and UL transmissions, which help not only alleviate network interference (SI, CCI and MUI) but also better exploit different channel conditions among users. We then formulate a novel SE maximization problem of joint power control and user association in	FD-NOMA systems. The formulated problem is a mixed-integer non-convex programming, which is generally NP-hard.}
	\item We then propose two suboptimal low-complexity algorithms by relaxing the binary variables. In the first one,  we resort to the inner convex approximation (ICA) framework \cite{Marks:78,Beck:JGO:10} to tackle the non-convex relaxed problem. Via our novel approximations, the convex program solved at each iteration can be cast as a second-order cone (SOC) program for which the modern convex solvers are very efficient. In the second design, we apply the penalty method to control the tightness of the continuous relaxation (CR) problem, which helps improve the system performance in terms of convergence speed and the SE. Numerical simulations later show that the continuous variables found at the convergence are nearly exact binary, suggesting a close-to-optimal solution.
	\item For a benchmarking purpose, we present an optimal design using the brute-force search (BFS) to find the best user association among all possible cases, combined with the ICA method for the problem of power control. 
	\item  Numerical results are provided to demonstrate the convergence of the proposed algorithms and the achieved SE gains of the proposed FD-NOMA schemes over state-of-the-art approaches, i.e., the conventional FD \cite{Dan:TWC:14}, FD-NOMA with random UA (RUA) and HD-NOMA.
\end{itemize}
\vspace{-15pt}
\subsection{Paper Organization and Notation}
The remainder of this paper is organized as follows. Section \ref{System Model and Problem Formulation} presents the system model and  the problem formulation for FD-NOMA systems. Two proposed suboptimal algorithms using ICA to solve the CR problem are introduced in Section \ref{sec: Continuous Relaxation Problems}. Section \ref{sec: sum rate maximization} discusses the BFS algorithm, while the analysis of initial point, convergence and complexity is shown in Section \ref{sec: Init, Converg, Complexity}. Numerical results are presented in Section \ref{NumericalResults}, and Section \ref{Conclusion} concludes the paper.

\emph{Notation}:  $\mathbf{X}^{T}$, $\mathbf{X}^{H}$ and $\tr(\mathbf{X})$ are the transpose, Hermitian transpose and trace of a matrix $\mathbf{X}$, respectively.  $\|\cdot\|$ denotes the Euclidean norm of a matrix or vector, while $|\cdot|$ stands for the absolute value of a complex scalar.  $\Re\{\cdot\}$ returns the real part of the argument. The notations $\mathbf{X}\succeq\mathbf{0}$, $\mathbf{X}\succ\mathbf{0}$ mean that
$\mathbf{X}$ is a positive-semidefinite or positive-definite matrix, respectively. $ \mathbf{x} \preceq \mathbf{y}  $ denotes the element-wise comparison of vectors, in which for the same index, a certain element of $ \mathbf{x} $ is not larger than the corresponding element of $ \mathbf{y} $.  $\mathbf{x}\sim\mathcal{CN}(\boldsymbol{\eta},\boldsymbol{Z})$ means that $\mathbf{x}$ is a random vector following a circularly symmetric complex Gaussian distribution with mean  $\boldsymbol{\eta}$ and covariance matrix $\boldsymbol{Z}$.
\section{System Model and Problem Formulation} \label{System Model and Problem Formulation}

\begin{figure}[t]
	\begin{minipage}[t]{0.48\columnwidth}
	\centering
	\includegraphics[width=0.81\columnwidth,trim={0cm 0.0cm 0cm 0.0cm}]{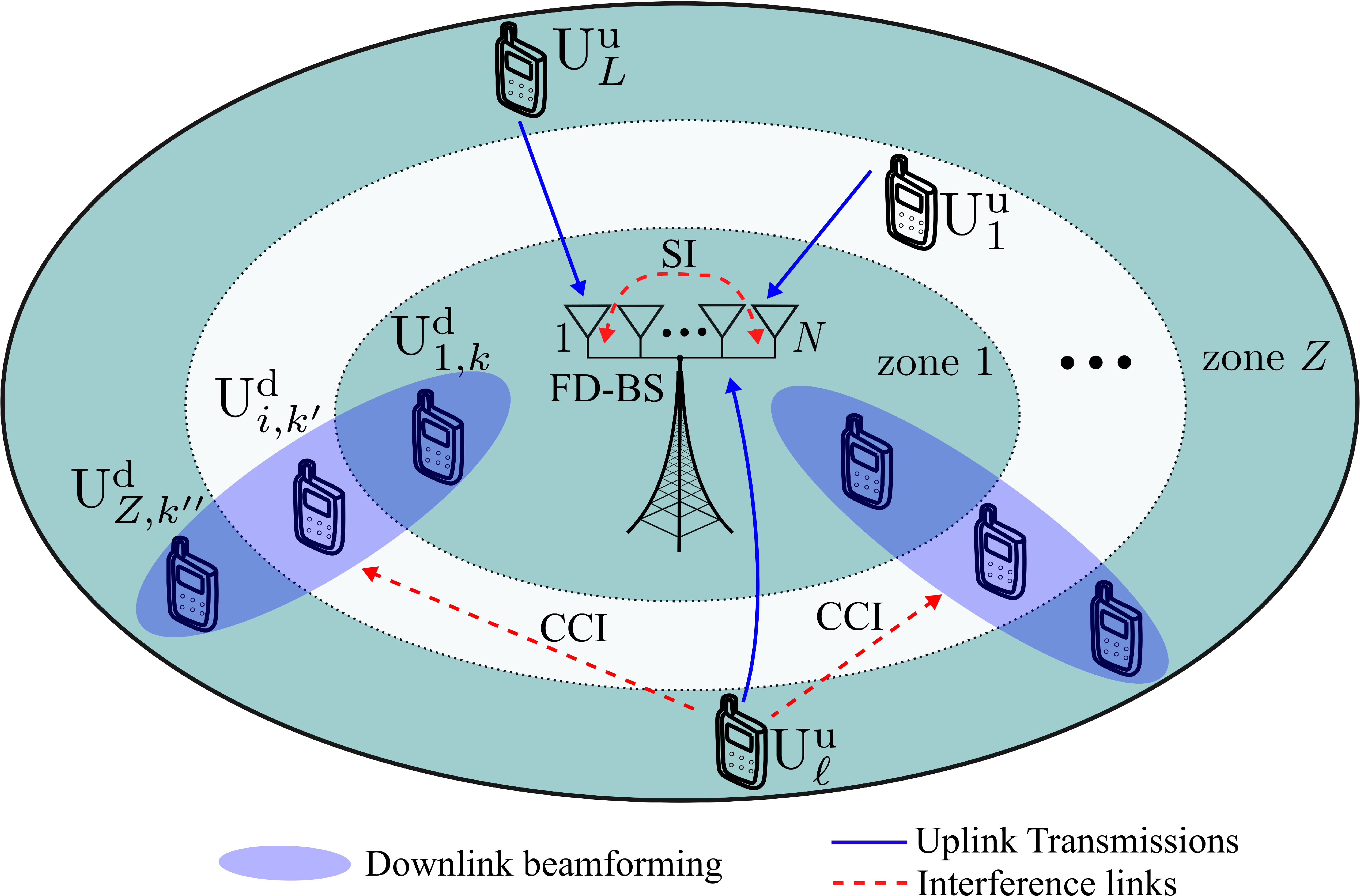}
	\caption{A small cell FD-NOMA MU-MISO system. FD-BS serves $ M=ZK $ DL users, with $ K $ DL users in each of $ Z $ zones, and $ L $ UL users which are assumed to be uniformly deployed in the cell.}
	\label{fig: system model}
	\end{minipage}
	\hfill
	\begin{minipage}[t]{0.48\columnwidth}
		\centering
		\includegraphics[width=0.72\columnwidth,trim={0cm 0.0cm 0cm 1.0cm}]{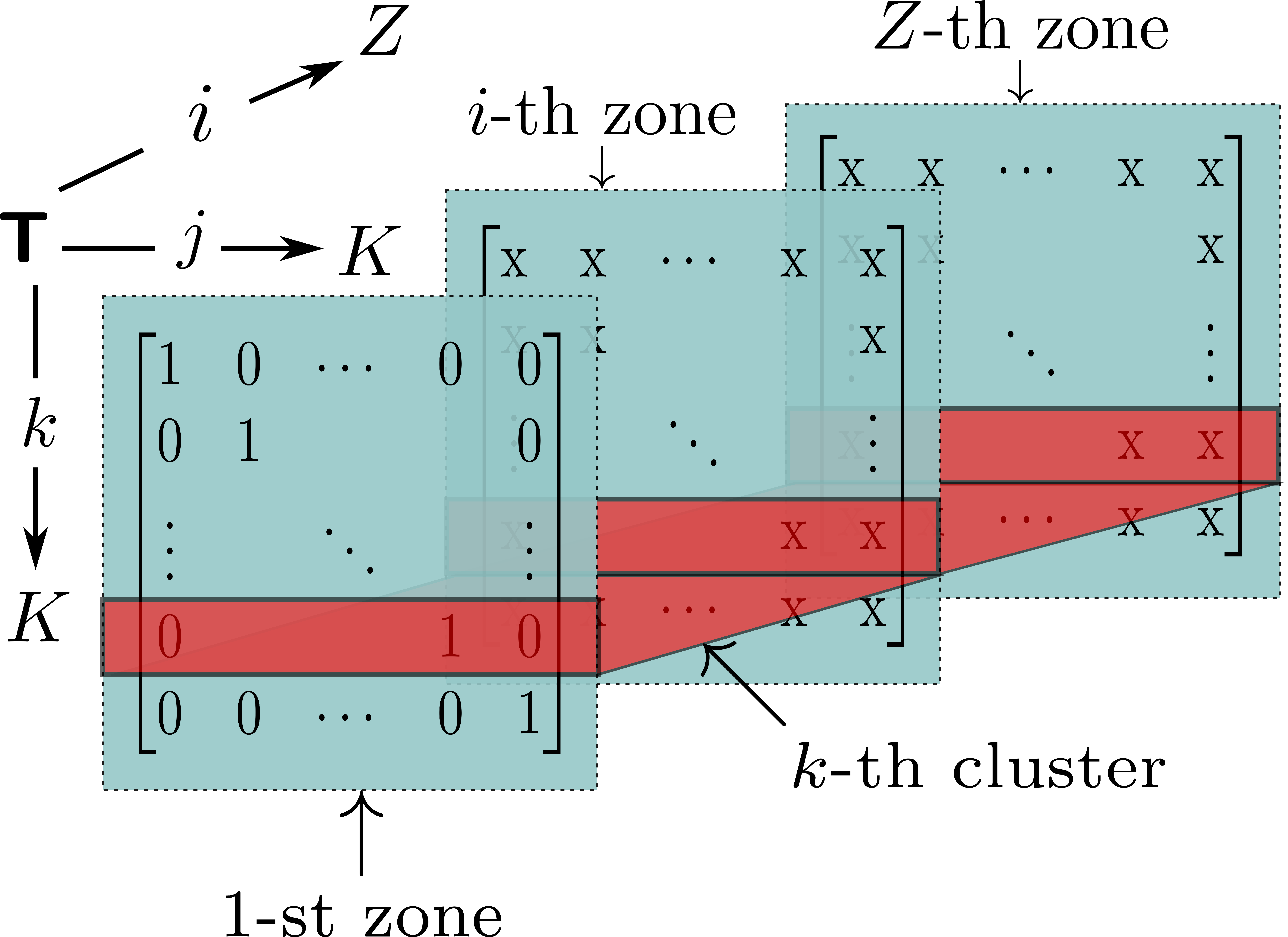}
		\caption{The structure of  tensor $ \tensor{T} $, representing the DL user association. The $ i $-th plane containing matrix $ \mathbf{C}_i\in\{0,1\}^{K\times K},\;i\in\mathcal{Z} $ indicates whether the $ j $-th DL user in zone $ i $ is assigned to the $ k $-th cluster. }
		\label{fig: system model - tensor}
	\end{minipage}
\vspace{-10pt}
\end{figure}

{\hilidra Consider a small cell in which the BS is equipped with $ N > 1$ antennas. To facilitate the NOMA operation, as in \cite{DSP16, Dinh:JSAC:Dec2017,ChenJSAC2017}, the cell is virtually partitioned into $Z$ annular regions (or zones) whose channel conditions are as much different as possible. Specifically, we number the zones/regions of $ \mathcal{Z}\triangleq\{1,2,\dots,Z\} $ in an ascending order with respect to their distance from the BS, i.e., the $ 1 $-st and $ Z $-th zones are the nearest and farthest zones, respectively.  In practice, we do not know how many zones are chosen to provide the best performance, since it depends on different criteria, such as channel condition, cell size,  the number of antennas and users. Finding the optimal number of zones is out of scope of this work. For the sake of mathematical convenience, we assume each zone contains $ K $ DL users leading $ M=ZK $ DL users in total (note that the following analysis is also applicable when zones have different numbers of DL users).

The BS is assumed to be equipped with FD capability, e.g., using the circulator-based FD radio prototypes \cite{Bharadia13} to simultaneously serve  $M$ and $L$ single-antenna DL and UL users in the same frequency band, respectively. We denote the $ k $-th DL user in zone $ i $ by $\DLU, \; \forall i \in \mathcal{Z},  k \in \mathcal{K}\triangleq \{1,2,\dots,K\} $, while the $ \ell $-th UL user at an arbitrary location  is represented by $\ULU, \; \forall \ell \in \mathcal{L}\triangleq \{1,2,\dots,L\} $. The channel vectors from the BS to $\DLU$ and from $\ULU$ to the BS are denoted by $ \mathbf{h}_{ik}^\dl\in \mathbb{C}^{N\times1} $ and $ \mathbf{h}_{\ell}^\ul \in \mathbb{C}^{N\times1}$, respectively. To capture the imperfect SiS at the BS, let $ \mathbf{G}_{\SI} \in \mathbb{C}^{N\times N} $ and $ \rho\in[0,1) $  be the SI channel matrix and  the residual SiS level, respectively. Further, let $ g_{\ell,ik} $ denote  the CCI channel from $\ULU$ to $\DLU$. }


\vspace{-12pt}
\subsection{Downlink Transmission}\vspace{-5pt}
 Before proceeding further, we first lay a foundation on third-order tensor to generalize the DL user clustering through the following definitions.

\begin{definition} \label{def: tensor for assignment}
{\hilidra A group of DL users consisting of $ Z $ DL users, in which no two DL users come from the same zone is called a cluster (of users). The NOMA beamforming is thus applied to $ K $ different clusters.} The third-order tensor $ \tensor{T}\triangleq [T_{kji}]_{k,j\in\mathcal{K},i\in\mathcal{Z}} $ is used for user associations, where $ T_{kji}\in\{0,1\} $. If $ T_{kji}=1 $, the $ j $-th DL user in zone $ i $ is admitted to the $ k $-th cluster, and vice versa.
\end{definition}

\begin{definition} \label{def: index clusters}
{\hilidr It can be foreseen that $ K $ users in each zone will result in $ K! $ possible permutations of clusters. To simplify the considered problem we utilize DL user indices in the first zone to index the clusters. In other words, the $ k $-th DL user in the first zone is always admitted to the $ k $-th cluster.} As illustrated in Fig. \ref{fig: system model - tensor}, $ \tensor{T} $ is formed by $ Z $ matrices w.r.t. the index $ i $ as $ \tensor{T}=\{\mathbf{C}_i\}_{i\in\mathcal{Z}} $, with $ \mathbf{C}_i\triangleq\bigl[T_{kji}\bigr]_{k,j\in\mathcal{K}}\in\{0,1\}^{K\times K} $ representing the $ i $-th zone, and thus the first matrix of $ \tensor{T} $ is assigned to the identity matrix, i.e., $ \mathbf{C}_1\triangleq\bigl[T_{kj1}\bigr]=\mathbf{I}_{K} $. According to \textbf{Definition} \textbf{\ref{def: tensor for assignment}}, $ \mathbf{C}_{i}$ with $ i\in\mathcal{Z}\backslash\{1\} $ is considered as the association variables of DL users. 
\end{definition}
    
From the two  definitions above, we now establish the UA between two arbitrary zones as follows. 
\begin{theorem} \label{thm: user association matrix}
Let $ \mathbf{T}^{iz}\in\{0,1\}^{K\times K}, \; \forall i,z\in\mathcal{Z} $ be an UA matrix between zones $ i $ and $ z $. If the entry $ T_{kj}^{iz},\;\forall k,j\in\mathcal{K} $ is set to 1, the $ k $-th DL user in zone $ i $ and the $ j $-th DL user in zone $ z $ are grouped into the same cluster, and vice versa. Based on the structure of $ \tensor{T} $, the matrix $ \mathbf{T}^{iz} $ is simply calculated as
	\begin{align} \label{eq: user association matrix}
		\mathbf{T}^{iz} = \mathbf{C}_{i}^T \mathbf{C}_{z}.
	\end{align} 
\end{theorem}
\begin{proof}
	Please see Appendix \ref{app: user association matrix}.
\end{proof}

In the DL channel,  BS employs a linear beamforming vector $ \mathbf{w}_{ik} \in \mathbb{C}^{N\times1} $ to precode the data symbol $ x_{ik}^{\dl} $, with $ \mathbb{E}\bigl[|x_{ik}^{\dl}|^2\bigr]=1$, intended to $\DLU$. The  received signal at $\DLU$ can be expressed as
\begin{align} \label{eq: received signal at DL user}
	\hspace{-8pt} y_{ik}^{\dl} = \sum\nolimits_{i'\in\mathcal{Z}}\sum\nolimits_{k'\in\mathcal{K}} (\mathbf{h}_{ik}^\dl)^H \mathbf{w}_{i'k'} x_{i'k'}^{\dl} + \underbrace{\sum\nolimits_{\ell\in\mathcal{L}} p_{\ell} g_{\ell,ik} x_{\ell}^{\ul}}_{\text{CCI}} + n_{ik}, \hspace{-5pt}
\end{align}
where $ p_{\ell} $ and $ x_{\ell}^{\ul} $, with $ \mathbb{E}\bigl[|x_{\ell}^{\ul}|^2\bigr]=1$,  are the transmit power coefficient and data symbol of $ \ULU $, respectively; and $ n_{ik} \sim \mathcal{CN}(0, \sigma_{ik}^2) $ is the additive white Gaussian noise (AWGN) at $\DLU$.
The messages intended to DL user in cluster $k$ are sequentially decoded as follows.  $\DLUi{ik}$ first decodes the  messages of $\DLUi{i'j}$ with  $ i'\in\mathcal{Z}_{i}^{+}\triangleq\{i+1,\dots,Z\} $  for $ T_{kj}^{ii'}=1 $, and then removes them by using the SIC technique before decoding its own message. {\hili The received signal-to-interference-plus-noise ratio (SINR) at  $\DLUi{ik}$ can be generally expressed as
	\begin{align} \label{eq: DL SINR general form}
	\gamma_{ik}^{\dl}(\mathbf{w},\mathbf{p},\tensor{T}) = \underset{ z\in\mathcal{Z}_{i}}{\min}\;\underset{ j\in\mathcal{K}}{\max}\Biggl\{\mfrac{T_{jk}^{zi}|(\mathbf{h}_{zj}^\dl)^H \mathbf{w}_{ik}|^2}{\Theta_{jk}^{zi}(\mathbf{w},\mathbf{p},\tensor{T})}\Biggr\},
	\end{align}
where $ \mathcal{Z}_{i}\triangleq\{1,\dots,i\} $, $ \mathbf{p}=[p_{\ell}]_{\ell\in\mathcal{L}} $, $ \mathbf{w}=[\mathbf{w}_i^H]^H_{i\in\mathcal{Z}} $ with $ \mathbf{w}_i\triangleq[\mathbf{w}_{ik}^H]_{k\in \mathcal{K}}^H $, and the interference-plus-noise (IN) for decoding the $\DLUi{ik}$'s message at $\DLUi{zj}$, denoted by $  \Theta_{jk}^{z i}(\mathbf{w},\mathbf{p},\tensor{T}) $, is given as
\begin{IEEEeqnarray}{l}
\Theta_{jk}^{zi}(\mathbf{w},\mathbf{p},\tensor{T}) = \underset{(z',j')\neq(i,k)}{\sum_{z'\in\mathcal{Z}_{i} }\;\sum_{j'\in\mathcal{K}}} |(\mathbf{h}_{zj}^\dl)^H \mathbf{w}_{z'j'}|^2 + \sum_{i'\in\mathcal{Z}_{i}^{+} }\sum_{k'\in\mathcal{K}}(1-T_{jk'}^{zi'})|(\mathbf{h}_{zj}^\dl)^H \mathbf{w}_{i'k'}|^2 
+ \sum_{\ell\in\mathcal{L}} p_{\ell}^2 |g_{\ell,zj}|^2 + \sigma_{zj}^2. \quad\;
\end{IEEEeqnarray} 
}

\vspace{-20pt}
\subsection{Uplink Transmission}
\vspace{-10pt}
The  received signal vector at the FD-BS in  the UL transmission can be expressed as
\begin{align} \label{eq: received signal at BS}
\mathbf{y}^{\ul} & = \sum\nolimits_{\ell\in\mathcal{L}} p_{\ell}  \mathbf{h}_{\ell}^{\ul} x_{\ell}^{\ul} + \underbrace{\rho\sum\nolimits_{i\in\mathcal{Z}}\sum\nolimits_{k\in\mathcal{K}} \mathbf{G}_{\SI}^H \mathbf{w}_{ik} x_{ik}^{\dl}}_{\text{SI}}  + \mathbf{n},
\end{align}
where $ \mathbf{n}\sim \mathcal{CN}(0,\sigma_{\mathtt{U}}^2\mathbf{I}) $ is the AWGN. To decode the UL messages, we adopt the minimum mean-square error and SIC (MMSE-SIC) decoder  at the FD-BS \cite{Tse:book:05}. To jointly optimize the UL users' decoding order, we  introduce binary variables $ \beta_{\ell m}\in \{0,1\}, \forall \ell, m \in \mathcal{L} $. Specifically, the message of the $ \ell $-th UL user is successfully decoded prior to that of the $ m $-th UL user if $ \beta_{\ell m}=1 $ in sync with $ \beta_{m \ell}=0 $, and they are in reverse order if $ \beta_{\ell m}=0 $.  On this basis and by treating the residual SiS as  noise, the received SINR of $ \ULU $ at the FD-BS can be expressed as
\begin{equation} \label{eq: UL SINR}
\gamma_{\ell}^{\ul}\bigl(\mathbf{w}, \mathbf{p}, \boldsymbol{\beta}\bigr) = p_{\ell}^2 (\mathbf{h}_{\ell}^{\ul})^H \bigl(\boldsymbol{\Psi}_{\ell}(\mathbf{w}, \mathbf{p}, \boldsymbol{\beta})\bigr)^{-1} \mathbf{h}_{\ell}^{\ul},
\end{equation}
where $ \boldsymbol{\beta}\triangleq[\beta_{\ell m}]_{\ell, m\in \mathcal{L}} $ and
\begin{align}
	\boldsymbol{\Psi}_{\ell}(\mathbf{w}, \mathbf{p}, \boldsymbol{\beta}) \triangleq  \sum\nolimits_{m\in\mathcal{L}} \beta_{\ell m} p_{m}^2 \mathbf{h}_{m}^{\ul}(\mathbf{h}_{m}^{\ul})^H  + \rho^2\sum\nolimits_{i\in\mathcal{Z}}\sum\nolimits_{k\in\mathcal{K}} \mathbf{G}_{\SI}^H \mathbf{w}_{ik} \mathbf{w}_{ik}^H\mathbf{G}_{\SI} + \sigma_{\mathtt{U}}^2\mathbf{I}. \nonumber
\end{align}

\subsection{Optimization Problem Formulation}
{\hilidra With the above discussion, the achievable rates (measured in nats/s/Hz) of $\DLUi{ik}$ and $ \ULU $ are respectively given as
\begin{align} 
R_{ik}^{\dl}\bigl(\mathbf{w}, \mathbf{p},\tensor{T}\bigr) & = \ln\bigl(1+\gamma_{ik}^{\dl}(\mathbf{w}, \mathbf{p},\tensor{T})\bigr), \label{eq: DL rate - nonconvex} \\
R_{\ell}^{\ul}\bigl(\mathbf{w}, \mathbf{p},\boldsymbol{\beta}\bigr) & = \ln\bigl(1+\gamma_{\ell}^{\ul}(\mathbf{w}, \mathbf{p},\boldsymbol{\beta})\bigr) \label{eq: UL rate - nonconvex}.
\end{align}
Herein, our main goal is to jointly optimize beamformers and  transmit power ($\mathbf{w}, \mathbf{p}$) and binary variables $(\tensor{T},\boldsymbol{\beta})$, so that the total SE is maximized subject to QoS and power constraints. We can now state  the SE maximization problem,  referred to as SEM problem for short,  as}
\begingroup
\allowdisplaybreaks\begin{subequations} \label{eq: prob. general form tensor}
	\setlength{\jot}{0.05pt}
	\begin{align}
		\hspace{-15pt}\underset{\mathbf{w}, \mathbf{p},\tensor{T},\boldsymbol{\beta}}{\max} &\   R_{\Sigma}\triangleq\sum\nolimits_{i\in\mathcal{Z}}\sum\nolimits_{k\in\mathcal{K}}R_{ik}^{\dl}\bigl(\mathbf{w}, \mathbf{p},\tensor{T}\bigr)+\sum\nolimits_{\ell\in\mathcal{L}}R_{\ell}^{\ul}\bigl(\mathbf{w}, \mathbf{p},\boldsymbol{\beta}\bigr) \label{eq: prob. general form tensor :: a} \\ 
		\hspace{-15pt}\st & \; \|\mathbf{w}\|^2 \leq P_{\bs}^{\text{max}}, \label{eq: prob. general form tensor :: b} \\
		&  p_{\ell}^2 \leq P_{\ell}^{\text{max}},\; p_{\ell} \geq 0,\ \forall \ell \in \mathcal{L}, \label{eq: prob. general form tensor :: c} \\
		&  
		R_{ik}^{\dl}\bigl(\mathbf{w}, \mathbf{p},\tensor{T}\bigr) \geq \bar{R}_{ik}^{\dl}, \; \forall i\in \mathcal{Z},\;  k \in \mathcal{K}, \label{eq: prob. general form tensor :: e} \\
		&  
		R_{\ell}^{\ul}\bigl(\mathbf{w}, \mathbf{p},\boldsymbol{\beta}\bigr) \geq \bar{R}_{\ell}^{\ul}, \; \forall \ell \in \mathcal{L}, \label{eq: prob. general form tensor :: f} \\
		&  T_{kj}^{iz} \in \{0,1\}, \;\forall i,z\in\mathcal{Z},\; \forall k,j \in \mathcal{K}, \label{eq: prob. general form tensor :: g} \\
		&  \sum\nolimits_{k\in\mathcal{K}}T_{kj}^{iz} = 1,\; \sum\nolimits_{j\in\mathcal{K}}T_{kj}^{iz} = 1,\forall i,z\in\mathcal{Z},\  \forall k, j \in \mathcal{K}, \label{eq: prob. general form tensor :: h} \;\\
		&  \beta_{\ell m} \in \{0,1\}, \; \forall \ell, m \in \mathcal{L}, \label{eq: prob. general form tensor :: j} \\
		&  \beta_{\ell\ell} = 0, \; \forall \ell \in \mathcal{L}, \label{eq: prob. general form tensor :: k} \\
		&  \beta_{\ell m} + \beta_{m \ell} = 1, \;  \ell \neq m,\; \forall \ell,m \in \mathcal{L}, \label{eq: prob. general form tensor :: l} \\
		& \bigl|\sum\nolimits_{m\in\mathcal{L}}\beta_{\ell m} - \sum\nolimits_{m\in\mathcal{L}}\beta_{\ell' m}\bigr| \geq 1, \;  \ell \neq \ell',\; \forall \ell, \ell' \in \mathcal{L}, \label{eq: prob. general form tensor :: m} 
	\end{align}							
\end{subequations}\endgroup
where $ P_{\bs}^{\text{max}} $ and $ P_{\ell}^{\text{max}} $ in \eqref{eq: prob. general form tensor :: b} and \eqref{eq: prob. general form tensor :: c}  are the transmit power budgets at the BS and $\ULU$, respectively. In \eqref{eq: prob. general form tensor :: e} and \eqref{eq: prob. general form tensor :: f}, we impose   the minimum QoS requirements  $\bar{R}_{ik}^{\dl}\geq 0$ and $\bar{R}_\ell^{\ul}\geq 0$ in order to maintain some degree of fairness among users. Constraints \eqref{eq: prob. general form tensor :: g} and \eqref{eq: prob. general form tensor :: h} establish the criteria for DL user clustering, in which $ \mathbf{T}^{iz} $ satisfies the property of tensor $ \tensor{T} $ given in \textbf{Theorem} \textbf{\ref{thm: user association matrix}}, while constraints \eqref{eq: prob. general form tensor :: j}-\eqref{eq: prob. general form tensor :: m} determine the decoding orders of UL users. Due to the non-concavity of the objective \eqref{eq: prob. general form tensor :: a} and the non-convexity of  QoS constraints w.r.t. the decision variables (i.e.,  by examining the Hessian matrix),  problem \eqref{eq: prob. general form tensor} belongs to a  class of mixed-integer non-convex problem.

\begin{remark}The merits of \textbf{Theorem} \textbf{\ref{thm: user association matrix}} to problem \eqref{eq: prob. general form tensor} are as follows. {\hili First, the search region for DL user clustering via tensor $ \tensor{T} $ is significantly reduced compared to the exhaustive search method, since $ \mathbf{C}_1 $ is fixed to identity matrix. Accordingly, the association problem avoids searching all permutations of clusters while still achieving a close-to-optimal solution.} The second advantage is to reduce the number of association variables. For NOMA decoding at a certain zone, the knowledge of UA with  other zones is required, i.e., $ K^2 $ variables for each of $ {2\choose Z} $ couples of zones,  leading to $ {2\choose Z}K^2 $ UA variables required zone-by-zone. In comparison, by exploiting the tensor structure of $ \tensor{T} $, the number of decision variables for user clustering is significantly reduced to $ (Z-1)K^2 $. Finally, the association among DL users is properly defined by \textbf{Theorem} \textbf{\ref{thm: user association matrix}}, in which the tensor $ \tensor{T} $ storing the change-of-basis matrices w.r.t. the identity-matrix basis is sufficient to recover the UA matrix of any two zones via \eqref{eq: user association matrix}. Note that $ T_{kji} $ is an element of $ \tensor{T} $, while $ T_{kj}^{iz} $ is an entry of the matrix $ \mathbf{T}^{iz} $ derived from \eqref{eq: user association matrix}. From the fact that $ \mathbf{C}_1=\mathbf{I}_K $ as  in \textbf{Definition} \textbf{\ref{def: index clusters}}, $ \mathbf{C}_i, i\in\mathcal{Z}\backslash\{1\} $ becomes the UA matrix between  zone 1 and  zone $i$, leading to $ \mathbf{T}^{1i}\equiv\mathbf{C}_i $. {\hili When zones have different numbers of DL users, let us denote the number of users in zone $ z $ by $ K_z,\;z\in\mathcal{Z}\triangleq\{1,\dots,Z\} $.  If $ K_z < K \triangleq \underset{z\in\mathcal{Z}}{\max} \{K_z\} $, some new users with zero channel vectors can be added to zone $ z $, such that each zone has $ K $ DL users. The third-order tensor remains unchanged. Accordingly, problem \eqref{eq: prob. general form tensor} can be easily re-expressed by forcing beamforming vectors of newly added users to be zeros and skipping their QoS constraints.}
\end{remark}


{\hili 
\begin{remark} As pointed out in \cite{Dinh:JSAC:Dec2017}, a larger cluster size with more distinct channel conditions among DL users is more desirable. This is because a large size of a DL cluster with very different channel gains plays a vital role in NOMA systems. This mainly depends on the number of DL users and cell size. In this paper, we focus on a small-cell setup, which is merely due to current practical limitations of FD radios \cite{Dan:TWC:14,Aquilina:TCOMM:2017,Dinh:Access,Yadav:Access, Dinh:JSAC:18,Tam:TCOM:16}. As such, it is reasonable to consider a two-zone NOMA for FD small-cell systems, as in the following case study. 
\end{remark}
}

\textit{{Case study with $Z = 2$}:} {\hili To reduce the system load and processing delay, we will study user pairing for DL transmission, where  NOMA is applied to  pairs of two DL users. Here, each pair includes one near DL user $\DLUi{1k}$ in the inner zone and one far DL user $\DLUi{2j}$ in the outer zone.} We should note that the two-zone NOMA  for DL transmission has been widely adopted in the literature \cite{Dinh:JSAC:Dec2017,Dinh:JSAC:18,DSP16}. In this case, $ \tensor{T} $  includes two UA matrices as $ \mathbf{C}_1=\mathbf{I}_K $ and  $ \mathbf{C}_2\in\{0,1\}^{K\times K} $. Without loss of generality, let $ \boldsymbol{\alpha}=\mathbf{T}^{12}=\mathbf{C}_2 $ be a unique matrix of UA variables in $ \tensor{T} $, and then the entries of  $ \boldsymbol{\alpha} $ are  $ \alpha_{kj} \in \{0,1\} $, indicating whether $\DLUi{1k}$ in zone 1 is paired with $\DLUi{2j}$ in zone 2.  From \eqref{eq: DL SINR general form}, the  SINRs at $\DLUi{1k}$ and $\DLUi{2j}$ are respectively simplified as
\begingroup
\allowdisplaybreaks\begin{subequations} \label{eq: DL SINR}
	\begin{align} 
	\gamma_{1k}^{\dl}(\mathbf{w},\mathbf{p},\boldsymbol{\alpha}) & = \mfrac{|(\mathbf{h}_{1k}^\dl)^H \mathbf{w}_{1k}|^2}{\phi_{k}(\mathbf{w},\mathbf{p},\boldsymbol{\alpha})}, \label{eq: DL SINR zone 1} \\
	\gamma_{2j}^{\dl}(\mathbf{w},\mathbf{p},\boldsymbol{\alpha}) & = \min\Biggl\{ \underset{ k\in\mathcal{K}}{\max}\Bigl\{\mfrac{\alpha_{kj}|(\mathbf{h}_{1k}^\dl)^H \mathbf{w}_{2j}|^2}{\psi_{j}^{k}(\mathbf{w},\mathbf{p})}\Bigr\},  \mfrac{|(\mathbf{h}_{2j}^\dl)^H \mathbf{w}_{2j}|^2}{\varphi_{j}(\mathbf{w},\mathbf{p})}\Biggr\}, \label{eq: DL SINR zone 2}
	\end{align}
\end{subequations}\endgroup
where the IN $\phi_{k}(\mathbf{w},\mathbf{p},\boldsymbol{\alpha})$ experienced by $\DLUi{1k}$ is
\begin{align} 
\phi_{k}(\mathbf{w},\mathbf{p},\boldsymbol{\alpha}) =  \sum\nolimits_{k'\in\mathcal{K}\setminus k} |(\mathbf{h}_{1k}^\dl)^H \mathbf{w}_{1k'}|^2 + \sum\nolimits_{\ell\in\mathcal{L}} p_{\ell}^2 |g_{\ell,1k}|^2  
 + \sum\nolimits_{j'\in\mathcal{K}} (1-\alpha_{kj'})|(\mathbf{h}_{1k}^\dl)^H \mathbf{w}_{2j'}|^2 + \sigma_{1k}^2,\nonumber
\end{align}
while the INs involved in the SINRs for decoding the $\DLUi{2j}$'s message at $\DLUi{1k}$ and itself, denoted by $\psi_{j}^{k}(\mathbf{w},\mathbf{p})$ and $\varphi_{j}(\mathbf{w},\mathbf{p})$, are respectively given as
\begin{align}
\psi_{j}^{k}(\mathbf{w},\mathbf{p}) = & \sum\nolimits_{k'\in\mathcal{K}} |(\mathbf{h}_{1k}^\dl)^H \mathbf{w}_{1k'}|^2 + \sum\nolimits_{j'\in\mathcal{K}\setminus j}|(\mathbf{h}_{1k}^\dl)^H \mathbf{w}_{2j'}|^2  + \sum\nolimits_{\ell\in\mathcal{L}} p_{\ell}^2 |g_{\ell,1k}|^2 + \sigma_{1k}^2,\nonumber \\
\varphi_{j}(\mathbf{w},\mathbf{p}) = & \sum\nolimits_{k'\in\mathcal{K}} |(\mathbf{h}_{2j}^\dl)^H \mathbf{w}_{1k'}|^2 + \sum\nolimits_{j'\in\mathcal{K}\setminus j}|(\mathbf{h}_{2j}^\dl)^H \mathbf{w}_{2j'}|^2  + \sum\nolimits_{\ell\in\mathcal{L}} p_{\ell}^2 |g_{\ell,2j}|^2 + \sigma_{2j}^2.\nonumber
\end{align}
We remark that the first term in \eqref{eq: DL SINR zone 2} is the SINR for decoding the $\DLUi{2j}$'s message at $\DLUi{1k}$, which is imposed on $\gamma_{2j}^{\dl}(\mathbf{w},\mathbf{p},\boldsymbol{\alpha})$ to ensure that $\DLUi{1k}$ can successfully decode the $\DLUi{2j}$'s message by SIC \cite{Dinh:JSAC:Dec2017,ChenTSP16}. Toward this end, we consider the following modified problem of \eqref{eq: prob. general form tensor} 
\begingroup
\allowdisplaybreaks\begin{subequations} \label{eq: prob. general form}
	\setlength{\jot}{0.05pt}
	\begin{align}
		\underset{\mathbf{w}, \mathbf{p},\boldsymbol{\alpha},\boldsymbol{\beta}}{\max} & \; R_{\Sigma}\triangleq\sum\nolimits_{i\in\mathcal{Z}}\sum\nolimits_{k\in\mathcal{K}}R_{ik}^{\dl}\bigl(\mathbf{w}, \mathbf{p},\boldsymbol{\alpha}\bigr)+\sum\nolimits_{\ell\in\mathcal{L}}R_{\ell}^{\ul}\bigl(\mathbf{w}, \mathbf{p},\boldsymbol{\beta}\bigr)  \label{eq: prob. general form a} \qquad\;\\
		\st & \; \eqref{eq: prob. general form tensor :: b}, \eqref{eq: prob. general form tensor :: c}, \eqref{eq: prob. general form tensor :: f}, \eqref{eq: prob. general form tensor :: j}-\eqref{eq: prob. general form tensor :: m}, \label{eq: prob. general form b}\\
		& \; 
		R_{ik}^{\dl}\bigl(\mathbf{w}, \mathbf{p},\boldsymbol{\alpha}\bigr) \geq \bar{R}_{ik}^{\dl}, \; \forall i\in\mathcal{Z},  k \in \mathcal{K}, \label{eq: prob. general form c} \\
		& \; \alpha_{kj} \in \{0,1\}, \; \forall k,j \in \mathcal{K}, \label{eq: prob. general form d} \\
		& \; \sum\nolimits_{k\in\mathcal{K}}\alpha_{kj} = 1, \; \sum\nolimits_{j\in\mathcal{K}}\alpha_{kj} = 1, \forall k,j \in \mathcal{K}. \label{eq: prob. general form e} 
	\end{align}							
\end{subequations}\endgroup
{\hili where $ R_{ik}^{\dl}\bigl(\mathbf{w}, \mathbf{p},\boldsymbol{\alpha}\bigr) $, is derived by replacing $ \gamma_{ik}^{\dl}(\mathbf{w}, \mathbf{p},\tensor{T}) $ in \eqref{eq: DL rate - nonconvex} with $ \gamma_{ik}^{\dl}(\mathbf{w}, \mathbf{p},\boldsymbol{\alpha}) $ given by \eqref{eq: DL SINR}.} In what follows, three iterative algorithms to solve the SEM problem \eqref{eq: prob. general form} will be presented. However, the mathematical presentation can be slightly modified to address the general SEM problem \eqref{eq: prob. general form tensor}, and numerical results for the three-zone NOMA scenario will be elaborated in Section \ref{NumericalResults}.

{\hili In this paper, we focus on slowly time-varying channels in small cell systems and adopt the channel reciprocity of UL and DL channels in time division duplex (TDD) mode. At the beginning of each time block, BS obtains the CSI of DL and UL channels by requesting all users to  simultaneously transmit mutually orthogonal pilot sequences of length $ \tau $ symbols (i.e., $ \tau \geq M+L $). Generally, the pilots are transmitted via public channel and the framework of training sequence can be known at UL users. Therefore, for the acquisition of the CSI of CCI channels, UL users  can take advantage of the pilots sent from DL users and then estimate CCI channels through TDD. Then, during the transmission phase, the estimated CCI channels acquired by the BS can be periodically embedded in the data packets of UL users. For slowly time-varying channels, it is reasonable to assume that the  CSI of all links in the network is available at the BS, where the system optimization is carried out. The performance under perfect CSI unveils an upper bound for the achievement of FD-NOMA systems.}

\vspace{-5pt}
\section{Proposed Suboptimal Designs} \label{sec: Continuous Relaxation Problems}

{\hili
In general, finding a globally optimal solution to \eqref{eq: prob. general form} is very challenging due to two obvious reasons. First, solving problem \eqref{eq: prob. general form} for an optimal solution may end up with an exhaustive search due to the strong coupling between the continuous variables $ (\mathbf{w},\mathbf{p}) $ and binary variables $ (\boldsymbol{\alpha},\boldsymbol{\beta}) $. Second, even if the binary variables are fixed,  problem \eqref{eq: prob. general form} still remains highly non-convex in the continuous variables. This method is of little practical use in wireless communication designs  since the resulting computational complexity grows exponentially when the number of users increases. This motivates us to develop  more practically appealing methods which can find a good solution with much less complexity. In particular, we first relax the binary UA variables to tackle the binary nature of the problem \eqref{eq: prob. general form}, and then, two suboptimal low-complexity algorithms are proposed to solve the resulting non-convex CR problem. In the first algorithm, we  employ the ICA framework \cite{Marks:78,Beck:JGO:10} to convexify the non-convex continuous parts, which aims at finding the approximate, but efficient solution to the CR problem. In the second algorithm, we propose an approach combining ICA and penalty method to produce a higher quality solution.}


\subsection{ICA-CR based Design}\vspace{-5pt}
We first consider the following CR of problem \eqref{eq: prob. general form} 
\begingroup
\allowdisplaybreaks
\begin{subequations} \label{eq: prob. general form - relax.}
	\setlength{\jot}{0.05pt}
	\begin{align}
		\underset{\mathbf{w}, \mathbf{p},\boldsymbol{\alpha},\boldsymbol{\beta}}{\max} & \quad \sum\nolimits_{i\in\mathcal{Z}}\sum\nolimits_{k\in\mathcal{K}}R_{ik}^{\dl}\bigl(\mathbf{w}, \mathbf{p},\boldsymbol{\alpha}\bigr)+\sum\nolimits_{\ell\in\mathcal{L}}R_{\ell}^{\ul}\bigl(\mathbf{w}, \mathbf{p},\boldsymbol{\beta}\bigr) \quad\quad \label{eq: prob. general form - relax. a} \\
		\st & \quad \eqref{eq: prob. general form tensor :: b}, \eqref{eq: prob. general form tensor :: c}, \eqref{eq: prob. general form tensor :: f}, \eqref{eq: prob. general form tensor :: k}-\eqref{eq: prob. general form tensor :: m}, 
		\eqref{eq: prob. general form c}, \eqref{eq: prob. general form e},  \label{eq: prob. general form - relax. b} \\
		& \quad 0 \leq \alpha_{kj} \leq 1,\;\forall k,j\in\mathcal{K}, \label{eq: prob. general form - relax. c} \\
		& \quad 0 \leq \beta_{\ell m}\leq 1,\;\forall \ell,m\in\mathcal{L}, \label{eq: prob. general form - relax. d}
	\end{align}							
\end{subequations}
\endgroup
where $\boldsymbol{\alpha}$ and $\boldsymbol{\beta}$ are relaxed to be continuous as in \eqref{eq: prob. general form - relax. c} and \eqref{eq: prob. general form - relax. d}, respectively.
By following the relaxation property, the feasible region of \eqref{eq: prob. general form - relax.} is larger than that of \eqref{eq: prob. general form}, and thus any feasible point of the latter is also feasible for the former. It is obvious that the difficulty in solving problem \eqref{eq: prob. general form - relax.} is due to  the non-concave objective  \eqref{eq: prob. general form - relax. a} and non-convex constraints in \eqref{eq: prob. general form tensor :: f}, \eqref{eq: prob. general form tensor :: m} and \eqref{eq: prob. general form c}. In light of the ICA method, convex approximations are required to tackle the non-convexity of  \eqref{eq: prob. general form - relax.}.

\textit{Concavity of the objective \eqref{eq: prob. general form - relax. a}:}
Let us start by handling  the non-concavity of $ R_{ik}^{\dl}(\mathbf{w}, \mathbf{p},\boldsymbol{\alpha}) $.  For $ \DLUi{1k}$ (DL users in zone 1), we introduce new variables $ \omega_{1k}, \forall k\in \mathcal{K} $ to explicitly expose the non-convex parts of $R_{1k}^{\dl}(\mathbf{w},\mathbf{p},\boldsymbol{\alpha})$ as
\begingroup
\allowdisplaybreaks\begin{subequations} \label{eq: SINR zone-1 convexification}
	\setlength{\jot}{0.05pt}
	\begin{align} 
	R_{1k}^{\dl}(\mathbf{w},\mathbf{p},\boldsymbol{\alpha}) & \geq \ln\bigl(1+\mfrac{1}{\omega_{1k}}\bigr), \label{eq: objective for omega1}\\
	\mfrac{|(\mathbf{h}_{1k}^\dl)^H \mathbf{w}_{1k}|^2}{\phi_{k}(\mathbf{w},\mathbf{p},\boldsymbol{\alpha})} & \geq \mfrac{1}{\omega_{1k}}, \label{eq: var. change omega1}
	\end{align}\end{subequations}\endgroup
which does not affect the optimality. The reason is that \eqref{eq: var. change omega1}  must hold with equality at the optimum; otherwise, we can decrease $\omega_{1k}$ to obtain a higher objective without violating the constraint. We are now ready to use ICA for approximating \eqref{eq: SINR zone-1 convexification}. {\hili Suppose the value of $ x $ at iteration $ \kappa $ in the proposed iterative algorithm is denoted by $ x^{(\kappa)} $, which is referred as a feasible point of $ x $ at iteration $ \kappa+1 $.} From \eqref{eq: objective for omega1} and as an effort to reduce the complexity of log function, the concave minorant of $ R_{1k}^{\dl}(\mathbf{w}, \mathbf{p},\boldsymbol{\alpha}) $ at the $ (\kappa+1) $-th iteration is derived as
\begin{align} 
R_{1k}^{\dl}(\mathbf{w},\mathbf{p},\boldsymbol{\alpha})  \geq \mathtt{A}(\omega_{1k}^{(\kappa)}) + \mathtt{B}(\omega_{1k}^{(\kappa)})\omega_{1k} := \ddot{R}_{1k}^{\dl, (\kappa)}, \label{eq: DL rate concave minorant1}
\end{align}
due to the convexity of $\ln\bigl(1+\mfrac{1}{\omega_{1k}}\bigr), $ where $\mathtt{A}(\omega_{1k}^{(\kappa)})  \triangleq \ln\bigl(1+\mfrac{1}{\omega_{1k}^{(\kappa)}}\bigr) + \mfrac{1}{(\omega_{1k}^{(\kappa)}+1)}$ and $
\mathtt{B}(\omega_{1k}^{(\kappa)})  \triangleq -\mfrac{1}{\omega_{1k}^{(\kappa)}(\omega_{1k}^{(\kappa)}+1)}$ \cite[Eq. (82)]{Dinh:JSAC:Dec2017}.
It is obvious that  the equality in \eqref{eq: DL rate concave minorant1} holds true whenever $\omega_{1k}=\omega_{1k}^{(\kappa)}$.  In other words, we obtain $ R_{1k}^{\dl}(\mathbf{w},\mathbf{p},\boldsymbol{\alpha}) = \ddot{R}_{1k}^{\dl, (\kappa)} $ as $ \kappa \rightarrow \infty $.
By applying the ICA method, we iteratively replace the non-convex constraint \eqref{eq: var. change omega1} by
\begin{equation} \label{eq: var. change omega convex - relax.}
\phi_{k}(\mathbf{w},\mathbf{p},\boldsymbol{\alpha}) \leq \omega_{1k} \tilde{\gamma}_{1k}^{(\kappa)}\bigl(\mathbf{w}\bigr),
\end{equation}
over the trust region (i.e., the feasible domain):
\begin{align} \label{eq: positive effective channel DL 1,k 1}
\tilde{\gamma}_{1k}^{(\kappa)}(\mathbf{w}) \triangleq \; 2\Re\{(\mathbf{h}_{1k}^\dl)^H\mathbf{w}_{1k}^{(\kappa)}\}\Re\{(\mathbf{h}_{1k}^\dl)^H \mathbf{w}_{1k}\}  - \bigl(\Re\{(\mathbf{h}_{1k}^\dl)^H \mathbf{w}_{1k}^{(\kappa)}\}\bigr)^2  > 0,\; \forall k\in\mathcal{K}, 
\end{align}
where $\tilde{\gamma}_{1k}^{(\kappa)}(\mathbf{w})$ is the first order approximation of $|(\mathbf{h}_{1k}^\dl)^H \mathbf{w}_{1k}|^2$ around the point $\mathbf{w}_{1k}^{(\kappa)}$ found at iteration $\kappa$. It can be seen that \eqref{eq: var. change omega convex - relax.} is still a non-convex constraint due to non-convexity in the third term of function $\phi_{k}(\mathbf{w},\mathbf{p},\boldsymbol{\alpha})$. To overcome this issue, we introduce the following lemma.
\begin{lemma} \label{lem: UB xy2}
	Consider a function $ h(x,y)\triangleq xy^2,\; x>0 $. The convex majorant of $ h(x,y) $ is expressed as
	\begin{align} \label{eq: UB xy2}
		h(x,y)\leq \tilde{h}(x,z) \triangleq xz \leq \mfrac{z^{(\kappa)}}{2x^{(\kappa)}}x^2 + \mfrac{x^{(\kappa)}}{2z^{(\kappa)}}z^2 \triangleq \tilde{h}^{(\kappa)}(x,z),
	\end{align}
	by imposing an SOC constraint: $ y^2\leq z$, where $ z > 0 $ is a new variable.
\end{lemma}
\begin{proof}
	Please see Appendix \ref{app: UB xy2}.
\end{proof}
{\hili By applying \textbf{Lemma} \textbf{\ref{lem: UB xy2}}, the convex upper bound of $\phi_{k}(\mathbf{w},\mathbf{p},\boldsymbol{\alpha})$ is
\begin{align} \label{eq: var. change lambda}
\phi_{k}(\mathbf{w},\mathbf{p},\boldsymbol{\alpha}) &\leq\sum_{k'\in\mathcal{K}\setminus k} |(\mathbf{h}_{1k}^\dl)^H \mathbf{w}_{1k'}|^2 + \sum_{\ell\in\mathcal{L}} p_{\ell}^2 |g_{\ell,1k}|^2  \nonumber + \sum_{j'\in\mathcal{K}}\tilde{h}^{(\kappa)}(\lambda_{kj'},\mu_{kj'}) 
+ \sigma_{1k}^2 
:=  \tilde{\phi}_{k}^{(\kappa)}(\mathbf{w},\mathbf{p},\boldsymbol{\lambda},\boldsymbol{\mu}), \nonumber
\end{align}
\color{black} where} $ \boldsymbol{\lambda}\triangleq[\lambda_{kj}]_{k,j\in\mathcal{K}} $ and $ \boldsymbol{\mu}\triangleq[\mu_{kj}]_{k,j\in\mathcal{K}} $ are  alternative variables, in which $ \lambda_{kj} $ and $ \mu_{kj} $ satisfy the following convex constraints:
\begin{subequations} \label{eq: var. change mu}
	\setlength{\jot}{0.05pt}
	\begin{align} 
	\lambda_{kj}&=1-\alpha_{kj},\; \forall k,j\in\mathcal{K}, \\
	|(\mathbf{h}_{1k}^{\dl})^H\mathbf{w}_{2j}|^2&\leq\mu_{kj},\; \; \forall k,j\in\mathcal{K}.
	\end{align}
\end{subequations}
In this regard, we iteratively replace \eqref{eq: var. change omega convex - relax.} by the convex constraint:
\begin{equation} \label{eq: var. change omega convex - relax. 1}
\tilde{\phi}_{k}^{(\kappa)}(\mathbf{w},\mathbf{p},\boldsymbol{\lambda},\boldsymbol{\mu}) \leq \omega_{1k} \tilde{\gamma}_{1k}^{(\kappa)}\bigl(\mathbf{w}\bigr),\; \forall k\in\mathcal{K},
\end{equation} 
which is the SOC representative.
 To address $ R_{2j}^{\dl}(\mathbf{w},\mathbf{p},\boldsymbol{\alpha})$ (DL users in zone 2), we  can  equivalently express  the SINR $ \gamma_{2j}^{\dl} $ of $ \DLUi{2j} $  as
\begin{align} \label{eq: DL SINR zone 2 - modified}
\gamma_{2j}^{\dl}(\mathbf{w},\mathbf{p},\boldsymbol{\alpha})  = \min\Biggl\{ \underset{ k\in\mathcal{K}}{\min}\Bigl\{\mfrac{|(\mathbf{h}_{1k}^\dl)^H \mathbf{w}_{2j}|^2}{\alpha_{kj}\psi_{j}^{k}(\mathbf{w},\mathbf{p})}\Bigr\}, 
 \mfrac{|(\mathbf{h}_{2j}^\dl)^H \mathbf{w}_{2j}|^2}{\varphi_{j}(\mathbf{w},\mathbf{p})}\Biggr\}.
\end{align}
Note that $ \underset{ k\in\mathcal{K}}{\max}\Bigl\{\mfrac{\alpha_{kj}|(\mathbf{h}_{1k}^\dl)^H \mathbf{w}_{2j}|^2}{\psi_{j}^{k}(\mathbf{w},\mathbf{p})}\Bigr\} = \underset{ k\in\mathcal{K}}{\min}\Bigl\{\mfrac{|(\mathbf{h}_{1k}^\dl)^H \mathbf{w}_{2j}|^2}{\alpha_{kj}\psi_{j}^{k}(\mathbf{w},\mathbf{p})}\Bigr\}$ at the optimality due to the inverse function of $ \alpha_{kj} $. We replace $\alpha_{kj}$ by $\alpha_{kj}+\varepsilon$ to avoid the numerical problem design when $\alpha_{kj}=0$, where $ \varepsilon $ is a given small number. In the same manner to \eqref{eq: DL rate concave minorant1}, the non-smoothness and non-concavity of $ R_{2j}^{\dl}(\mathbf{w},\mathbf{p},\boldsymbol{\alpha}) $ are tackled as \setlength{\jot}{0.05pt}
\begin{align} 
R_{2j}^{\dl}(\mathbf{w},\mathbf{p},\boldsymbol{\alpha}) & \geq \ln\Bigl(1+\mfrac{1}{\omega_{2j}}\Bigr)\geq \mathtt{A}(\omega_{2j}^{(\kappa)}) + \mathtt{B}(\omega_{2j}^{(\kappa)})\omega_{2j} := \ddot{R}_{2j}^{\dl, (\kappa)}, \label{eq: DL rate var. change 2 1}
\end{align}
by imposing the following  constraints
\begin{subequations} \label{eq: var. change omega 2 1}
	\setlength{\jot}{0.1pt}
	\begin{align}
	\mfrac{|(\mathbf{h}_{1k}^\dl)^H \mathbf{w}_{2j}|^2}{(\alpha_{kj}+\varepsilon)\psi_{j}^{k}(\mathbf{w},\mathbf{p})} & \geq \mfrac{1}{\omega_{2j}}, \label{eq: var. change omega 2 a 1}\\
	\mfrac{|(\mathbf{h}_{2j}^\dl)^H \mathbf{w}_{2j}|^2}{\varphi_{j}(\mathbf{w},\mathbf{p})} & \geq \mfrac{1}{\omega_{2j}}. \label{eq: var. change omega 2 b 1}
	\end{align}
\end{subequations}
As in \eqref{eq: var. change omega1}, the non-convex constraints \eqref{eq: var. change omega 2 a 1} and \eqref{eq: var. change omega 2 b 1} are  innerly convexified by
\begin{subequations} \label{eq: var. change omega convex 2 1}
	\setlength{\jot}{0.05pt}
	\begin{align}
	\psi_{j}^{k}(\mathbf{w},\mathbf{p}) & \leq \omega_{2j} \ddot{\gamma}_{2,kj}^{(\kappa)}\bigl(\mathbf{w},\boldsymbol{\alpha}\bigr),\; \forall k,j\in\mathcal{K}, \label{eq: var. change omega convex 2 a 1}\\
	\varphi_{j}(\mathbf{w},\mathbf{p}) & \leq \omega_{2j} \tilde{\gamma}_{2j}^{(\kappa)}\bigl(\mathbf{w}\bigr), \; \forall j\in\mathcal{K}\label{eq: var. change omega convex 2 b 1},
	\end{align}
\end{subequations}
over the trust regions:
\begin{subequations} \label{eq: positive effective channel DL 2,k 1}
	\setlength{\jot}{0.05pt}
	\begin{align}
	\ddot{\gamma}_{2,kj}^{(\kappa)}(\mathbf{w},\boldsymbol{\alpha}) \triangleq & \; \mfrac{2\Re\bigl\{\bigl((\mathbf{h}_{1k}^\dl)^H \mathbf{w}_{2j}^{(\kappa)}\bigr)^*\bigl((\mathbf{h}_{1k}^\dl)^H \mathbf{w}_{2j}\bigr)\bigr\} }{\alpha_{kj}^{(\kappa)}+\varepsilon}
	-\mfrac{|(\mathbf{h}_{1k}^\dl)^H \mathbf{w}_{2j}^{(\kappa)}|^2}{(\alpha_{kj}^{(\kappa)}+\varepsilon)^2}(\alpha_{kj}+\varepsilon) > 0, \; \forall k,j\in\mathcal{K}, \label{eq: positive effective channel DL 2,k a 1}\\
	\tilde{\gamma}_{2j}^{(\kappa)}(\mathbf{w}) \triangleq & \; 2\Re\{(\mathbf{h}_{2j}^\dl)^H\mathbf{w}_{2j}^{(\kappa)}\}\Re\{(\mathbf{h}_{2j}^\dl)^H \mathbf{w}_{2j}\}  - \bigl(\Re\{(\mathbf{h}_{2j}^\dl)^H \mathbf{w}_{2j}^{(\kappa)}\}\bigr)^2  > 0, \; \forall j\in\mathcal{K}. \label{eq: positive effective channel DL 2,k b 1}
	\end{align}
\end{subequations}
Next, at the feasible point $(\mathbf{w}^{(\kappa)}, \mathbf{p}^{(\kappa)},\boldsymbol{\beta}^{(\kappa)})$, the UL rate $ R_{\ell}^{\ul}\bigl(\mathbf{w}, \mathbf{p},\boldsymbol{\beta}\bigr) $ is globally lower bounded by
\begingroup
\allowdisplaybreaks
\begin{align} \label{eq: UL rate convex - relax.}
	R_{\ell}^{\ul}(\mathbf{w}, \mathbf{p},\boldsymbol{\beta})  \geq  \tilde{\mathtt{A}}\bigl(\gamma_{\ell}^{\ul}\bigl(\mathbf{w}^{(\kappa)}, \mathbf{p}^{(\kappa)},\boldsymbol{\beta}^{(\kappa)}\bigr)\bigr) - \Phi_{\ell}^{(\kappa)}(\mathbf{w}, \mathbf{p},\boldsymbol{\beta}) +  \mfrac{2\gamma_{\ell}^{\ul}\bigl(\mathbf{w}^{(\kappa)}, \mathbf{p}^{(\kappa)},\boldsymbol{\beta}^{(\kappa)}\bigr)}{p_{\ell}^{(\kappa)}}p_{\ell},
\end{align}
where 
{\small \setlength{\jot}{0.05pt}\begin{align}
&\tilde{\mathtt{A}}\bigl(\gamma_{\ell}^{\ul}(\mathbf{w}^{(\kappa)}, \mathbf{p}^{(\kappa)},\boldsymbol{\beta}^{(\kappa)})\bigr)\triangleq  \ln\bigl(1+\gamma_{\ell}^{\ul}(\mathbf{w}^{(\kappa)}, \mathbf{p}^{(\kappa)},\boldsymbol{\beta}^{(\kappa)})\bigr)-\gamma_{\ell}^{\ul}(\mathbf{w}^{(\kappa)}, \mathbf{p}^{(\kappa)},\boldsymbol{\beta}^{(\kappa)}), \nonumber \\
& \Phi_{\ell}^{(\kappa)}(\mathbf{w}, \mathbf{p},\boldsymbol{\beta}) \triangleq
p_{\ell}^2\Lambda_{\ell}     +\sum\nolimits_{m\in\mathcal{L}\setminus \ell}\beta_{\ell m}p_{m}^2\Lambda_{m}
+ \rho^2\sum\nolimits_{k\in\mathcal{K}} \bigl(\mathbf{w}_{k}\bigr)^H \mathbf{G}_{\SI}\boldsymbol{\Xi}_{\ell}^{(\kappa)}\mathbf{G}_{\SI}^H \mathbf{w}_{k}+\sigma_{\mathtt{U}}^2\tr(\boldsymbol{\Xi}_{\ell}^{(\kappa)}) \nonumber, \\
& \Lambda_{\ell}\triangleq(\mathbf{h}_{\ell}^{\ul})^H\boldsymbol{\Xi}_{\ell}^{(\kappa)}\mathbf{h}_{\ell}^{\ul}, \quad \boldsymbol{\Xi}_{\ell}^{(\kappa)} \triangleq (\boldsymbol{\Psi}_{\ell}^{(\kappa)})^{-1} - ((p_{\ell}^{(\kappa)})^2\mathbf{h}_{\ell}^{\ul}(\mathbf{h}_{\ell}^{\ul})^H+\boldsymbol{\Psi}_{\ell}^{(\kappa)})^{-1}, \nonumber \\
& \boldsymbol{\Psi}_{\ell}^{(\kappa)} \triangleq  \sum\nolimits_{m\in\mathcal{L}\setminus \ell}\beta_{\ell m}^{(\kappa)} \bigl(p_{m}^{(\kappa)}\bigr)^2 \mathbf{h}_{m}^{\ul}\bigl(\mathbf{h}_{m}^{\ul}\bigr)^H    + \rho^2\sum\nolimits_{k\in\mathcal{K}} \mathbf{G}_{\SI}^H \mathbf{w}_{k}^{(\kappa)} \bigl(\mathbf{w}_{k}^{(\kappa)}\bigr)^H\mathbf{G}_{\SI}+\sigma_{\mathtt{U}}^2\mathbf{I}. \nonumber
\end{align}}
\endgroup
The right-hand side (RHS) of \eqref{eq: UL rate convex - relax.} is still non-concave, resulting from the fact that $ \Phi_{\ell}^{(\kappa)}(\mathbf{w}, \mathbf{p},\boldsymbol{\beta}) $ is a non-convex function. {\hili By applying \textbf{Lemma} \textbf{\ref{lem: UB xy2}} to the second term of $ \Phi_{\ell}^{(\kappa)}(\mathbf{w}, \mathbf{p},\boldsymbol{\beta}) $, it is true that
\begingroup \setlength{\jot}{0.05pt}
\allowdisplaybreaks
\begin{align}
\Phi_{\ell}^{(\kappa)}(\mathbf{w}, \mathbf{p},\boldsymbol{\beta}) 
& \leq p_{\ell}^2\Lambda_{\ell} + \rho^2\sum\nolimits_{k\in\mathcal{K}} \bigl(\mathbf{w}_{k}\bigr)^H \mathbf{G}_{\SI}\boldsymbol{\Xi}_{\ell}^{(\kappa)}\mathbf{G}_{\SI}^H \mathbf{w}_{k} 
+\sum\nolimits_{m\in\mathcal{L}\setminus\ell}\Lambda_{m}\tilde{h}^{(\kappa)}(\beta_{\ell m},\nu_{ m})
+\sigma_{\mathtt{U}}^2\tr(\boldsymbol{\Xi}_{\ell}^{(\kappa)}) \nonumber \\ & 
:= \tilde{\Phi}_{\ell}^{(\kappa)}(\mathbf{w}, \mathbf{p},\boldsymbol{\beta},\boldsymbol{\nu}),
\end{align}
\endgroup
\color{black} with }the additional SOC and linear constraints: \setlength{\jot}{0.05pt}
\begin{align}  \label{eq: var. change nu}
p_{m}^2\leq\nu_{m}\leq P_{m}^{\text{max}},\; \;\forall m\in\mathcal{L},
\end{align}
where $ \boldsymbol{\nu}\triangleq[\nu_{m}]_{ m\in\mathcal{L}} $ are new variables. The concave quadratic minorant of $R_{\ell}^{\ul}(\mathbf{w}, \mathbf{p},\boldsymbol{\beta})$ at  iteration $\kappa+1$ is given by
\begin{align} \label{eq: UL rate convex - relax. 1}
R_{\ell}^{\ul}(\mathbf{w}, \mathbf{p},\boldsymbol{\beta})  \geq \tilde{\mathtt{A}}\bigl(\gamma_{\ell}^{\ul}\bigl(\mathbf{w}^{(\kappa)}, \mathbf{p}^{(\kappa)},\boldsymbol{\beta}^{(\kappa)}\bigr)\bigr)  +  \mfrac{2\gamma_{\ell}^{\ul}\bigl(\mathbf{w}^{(\kappa)}, \mathbf{p}^{(\kappa)},\boldsymbol{\beta}^{(\kappa)}\bigr)}{p_{\ell}^{(\kappa)}}p_{\ell}- \tilde{\Phi}_{\ell}^{(\kappa)}(\mathbf{w}, \mathbf{p},\boldsymbol{\beta},\boldsymbol{\nu})     := \ddot{R}_{\ell}^{\ul,(\kappa)}.
\end{align}
Clearly, the objective function and constraints \eqref{eq: prob. general form tensor :: f}, \eqref{eq: prob. general form c} are innerly convexified by replacing non-concave functions $R_{ik}^{\dl}(\mathbf{w}, \mathbf{p},\boldsymbol{\alpha})$ and $R_{\ell}^{\ul}(\mathbf{w}, \mathbf{p},\boldsymbol{\beta})$ with concave functions $\ddot{R}_{ik}^{\dl,(\kappa)}$ and $\ddot{R}_{\ell}^{\ul,(\kappa)}$, respectively.

{\textit{Convexity of constraint \eqref{eq: prob. general form tensor :: m}:}}  We now are in position to tackle  constraint \eqref{eq: prob. general form tensor :: m}. It can be observed that the absolute function $\bigl|\sum\nolimits_{m\in\mathcal{L}}\beta_{\ell m} - \sum\nolimits_{m\in\mathcal{L}}\beta_{\ell' m}\bigr|$  is quasi-convex, leading to the non-convexity of constraint \eqref{eq: prob. general form tensor :: m}. To overcome this issue, we replace the absolute function with the maximum form, i.e., $ |s_{\ell \ell'}|=\max(s_{\ell \ell'},-s_{\ell \ell'}), \;\ell\neq\ell',\; \forall\ell, \ell'\in\mathcal{L} $, where $ s_{\ell \ell'}\triangleq\sum_{m=1}^{L}\beta_{\ell m}-\sum_{m=1}^{L}\beta_{\ell' m} $, and thus, constraint \eqref{eq: prob. general form tensor :: m} becomes
\begin{align}\label{eq:43from9k}
	\max(s_{\ell \ell'},-s_{\ell \ell'}) \geq 1, \;\ell\neq\ell',\; \forall\ell, \ell'\in\mathcal{L}.
\end{align}
We then apply a smooth approximation via the log-sum-exp (LSE) function \cite{Stephen}, to handle \eqref{eq:43from9k} as 
\begin{align}
	\max(s_{\ell \ell'},-s_{\ell \ell'}) \geq \mfrac{1}{\Omega}\ln\bigl(\exp(\Omega s_{\ell \ell'})+\exp(-\Omega s_{\ell \ell'})\bigr)-\mfrac{1}{\Omega}\ln(2)
	 := f_{\mathrm{LSE}}(s_{\ell \ell'}),
\end{align}
where $ \Omega $ is a predefined large number. {\hili This naturally leads to a direct application of the ICA method to approximate $ f_{\mathrm{LSE}}(s_{\ell \ell'}) $ around the point $s_{\ell \ell'}^{(\kappa)}$ as
\begin{align}
	f_{\mathrm{LSE}}(s_{\ell \ell'}) & \geq f_{\mathrm{LSE}}(s_{\ell \ell'}^{(\kappa)})
	+\mfrac{\partial f_{\mathrm{LSE}}}{\partial s_{\ell \ell'} }\Big|_{s_{\ell \ell'}=s_{\ell \ell'}^{(\kappa)}}\bigl(s_{\ell \ell'}-s_{\ell \ell'}^{(\kappa)}\bigr) 
	= f_{\mathrm{LSE}}(s_{\ell \ell'}^{(\kappa)})
	+ \tanh\bigl(\Omega s_{\ell \ell'}^{(\kappa)}\bigr) 
	:= f_{\mathrm{LSE}}^{(\kappa)}(s_{\ell\ell'}), \nonumber
\end{align}
where $ s_{\ell\ell'}^{(\kappa)}\triangleq\sum_{m=1}^{L}\beta_{\ell m}^{(\kappa)} - \sum_{m=1}^{L}\beta_{\ell' m}^{(\kappa)}$ and $ \tanh(x) $ is the hyperbolic tangent function of $ x $.} As a result, we iteratively replace \eqref{eq:43from9k} with the linear constraint:
\begin{align} \label{eq: abs function smooth approx.}
	f_{\mathrm{LSE}}^{(\kappa)}(s_{\ell\ell'}) \geq 1, \;\ell\neq\ell',\; \forall\ell, \ell'\in\mathcal{L}.
\end{align}

From the discussions above, the successive convex program  to solve \eqref{eq: prob. general form - relax.} at  iteration $ \kappa+1 $ is given~as
\vspace{-20pt}
\begingroup
\allowdisplaybreaks
\begin{subequations} \label{eq: prob. general form - relax. 1}
	\setlength{\jot}{0.05pt}
	\begin{align}
		\underset{\substack{\mathbf{w}, \mathbf{p},\boldsymbol{\alpha},\boldsymbol{\beta}, \\ \boldsymbol{\omega},\boldsymbol{\lambda},\boldsymbol{\mu},\boldsymbol{\nu}}}{\max} & \quad \ddot{R}_{\Sigma}^{(\kappa+1)}\triangleq\sum\nolimits_{i\in\mathcal{Z}}\sum\nolimits_{k\in\mathcal{K}}\ddot{R}_{ik}^{\dl,(\kappa)}+\sum\nolimits_{\ell\in\mathcal{L}}\ddot{R}_{\ell}^{\ul,(\kappa)}\label{eq: prob. general form - relax. 1 a} \\
		\st & \quad \eqref{eq: prob. general form tensor :: b}, \eqref{eq: prob. general form tensor :: k}, \eqref{eq: prob. general form tensor :: l}, \eqref{eq: prob. general form e}, \eqref{eq: prob. general form - relax. c},\eqref{eq: prob. general form - relax. d},  
		\eqref{eq: positive effective channel DL 1,k 1}, \eqref{eq: var. change mu}, \eqref{eq: var. change omega convex - relax. 1},  \eqref{eq: var. change omega convex 2 1}, \eqref{eq: positive effective channel DL 2,k 1},  \eqref{eq: var. change nu}, \eqref{eq: abs function smooth approx.}, \qquad \label{eq: prob. general form - relax. 1 b} \quad\\
		& \quad 
		\ddot{R}_{ik}^{\dl,(\kappa)} \geq \bar{R}_{ik}^{\dl}, \; \forall i\in\mathcal{Z},\;  k \in \mathcal{K}, \label{eq: prob. general form - relax. 1 c} \\
		& \quad 
		\ddot{R}_{\ell}^{\ul,(\kappa)} \geq \bar{R}_{\ell}^{\ul}, \; \forall \ell \in \mathcal{L}, \label{eq: prob. general form - relax. 1 d} 
	\end{align}							
\end{subequations}
\endgroup
where $ \boldsymbol{\omega}\triangleq\bigl[\omega_{ik}\bigr]_{i\in\mathcal{Z},\;k\in\mathcal{K}} $. We have numerically observed that the ICA-CR based algorithm for solving \eqref{eq: prob. general form - relax.}  yields nearly binary values at convergence, but some relaxed variables are non-binary. It simply means that the solution obtained by \eqref{eq: prob. general form - relax. 1} can be infeasible to problem \eqref{eq: prob. general form}. To make $\alpha_{kj}$ and $\beta_{\ell m}$ exact binary, we further introduce the round function at iteration $ \kappa $ as
	\begin{align}\label{eq: rounding binary var.}
	\alpha_{kj}^{*}  = \Bigl\lfloor\alpha_{kj}^{(\kappa)}+\mfrac{1}{2}\Bigr\rfloor,\; \forall k,j\in\mathcal{K}, \ \text{and}\	\beta_{\ell m}^{*}  = \Bigl\lfloor\beta_{\ell m}^{(\kappa)}+\mfrac{1}{2}\Bigr\rfloor,\; \forall \ell,m\in\mathcal{L}.
	\end{align}
 Upon defining $ \mathbf{X}\triangleq \bigl(\mathbf{w},\mathbf{p},\boldsymbol{\alpha}, \boldsymbol{\beta}, \boldsymbol{\omega},\boldsymbol{\lambda}, \boldsymbol{\mu}, \boldsymbol{\nu}\bigr) $, the iterative algorithm based on the ICA-CR method for solving \eqref{eq: prob. general form} is briefly described in \textbf{Algorithm} \textbf{\ref{alg: Continuous relaxation problem}}, where the procedure for finding an initial feasible point will be detailed in Section \ref{sec: Init, Converg, Complexity}. We note that the actual solutions of $(\mathbf{w},\mathbf{p})$ are straightforward to find by a post-processing method, i.e., repeating Steps 4-6 in \textbf{Algorithm} \textbf{\ref{alg: Continuous relaxation problem}} for  the fixed values of ($\boldsymbol{\alpha}, \boldsymbol{\beta}$) obtained at Step 8 until convergence.

\begin{algorithm}[t]
	\begin{algorithmic}[1]
		\fontsize{8.5}{9.5}\selectfont
		\protect\caption{\fontsize{9}{10}\selectfont Proposed ICA-CR Based Design for the Relaxed SEM Problem \eqref{eq: prob. general form - relax.}}
		
		\label{alg: Continuous relaxation problem}

		\STATE \textbf{Initialization:} Set $ R_{\Sigma} := -\infty $ and $ (\mathbf{w}^{*},\mathbf{p}^{*},\boldsymbol{\alpha}^{*},\boldsymbol{\beta}^{*}):=\mathbf{0} $.

		\STATE Set $\kappa:=0$ and solve \eqref{eq: prob. CRP initialization} and \eqref{eq: prob. CRP initialization 1} to generate an initial feasible point $ \mathbf{X}^{(0)} $.
		
		\REPEAT
		\STATE Solve \eqref{eq: prob. general form - relax. 1} to obtain  $ \mathbf{X}^{\star} $ and  $\ddot{R}_{\Sigma}^{(\kappa+1)}$.
		
		\STATE Update $ \mathbf{X}^{(\kappa+1)} :=\mathbf{X}^{\star} $.
		\STATE Set $ \kappa := \kappa + 1 $.
		\UNTIL Convergence
		
		\STATE\ Update $ (\mathbf{w}^{*},\mathbf{p}^{*}) := (\mathbf{w}^{(\kappa)},\mathbf{p}^{(\kappa)})$ and ($ \boldsymbol{\alpha}^{*},\boldsymbol{\beta}^{*} $) as in \eqref{eq: rounding binary var.}.
		\STATE\ Use $ (\mathbf{w}^{*},\mathbf{p}^{*},\boldsymbol{\alpha}^{*},\boldsymbol{\beta}^{*}) $ to compute $ R_{\Sigma} $  in \eqref{eq: prob. general form a}.

		\STATE {\textbf{Output:}  $ R_{\Sigma} $ and the optimal solution $ ( \mathbf{w}^{*},\mathbf{p}^{*},\boldsymbol{\alpha}^{*},\boldsymbol{\beta}^{*}) $.}
	\end{algorithmic} 
\end{algorithm}

\vspace{-8pt}
\subsection{ICA-CR based Design with Penalty Function}
\vspace{-8pt}
We now present the  ICA-CR based design with the penalty function (i.e., ICA-CR-PF) is to further process the uncertain values of  binary variables, which helps speed up the convergence. We can see that $\alpha_n^2=\alpha_n$ for any $\alpha_n\in\{0,1\}$. On the other hand, it is also true that $\alpha_n^2\leq\alpha_n$ for $\alpha_n\in[0,1]$. The idea of the PF method is to allow   the uncertainties of  binary variables to be penalized, which enforce $\alpha_n^2=\alpha_n$ quickly. In this regard, we first introduce the following theorem to derive the property of the PF method.
\begin{theorem} \label{thm: relax. prob with pen.}
	Suppose that $ (\mathbf{x},\boldsymbol{\alpha}) $ is a 2-tuple variable of the following mixed-integer problem:
	\begin{subequations} \label{eq: prob. mixed-integer theorem}
		\setlength{\jot}{0.05pt}
		\begin{align}
			\underset{\mathbf{x}, \boldsymbol{\alpha}}{\max} & \quad f_0(\mathbf{x}, \boldsymbol{\alpha})\label{eq: prob. mixed-integer theorem a} \\
			\st & \quad \mathbf{x}\in\mathcal{X}, \label{eq: prob. mixed-integer theorem b} \\
			& \quad 
			\boldsymbol{\alpha}\in\{0,1\}^{N_{\alpha}\times 1}, \label{eq: prob. mixed-integer theorem c}
		\end{align}					
	\end{subequations}
	where the objective function $ f_0(\mathbf{x}, \boldsymbol{\alpha}) $ is closed and bounded. The optimality of \eqref{eq: prob. mixed-integer theorem} is guaranteed by using the following CR form:
	\begin{subequations} \label{eq: prob. relax theorem}
		\setlength{\jot}{0.05pt}
		\begin{align}
			\underset{\mathbf{x}, \boldsymbol{\alpha}}{\max} & \quad f_0(\mathbf{x}, \boldsymbol{\alpha})+\sum\nolimits_{n\in\mathcal{N}} f_p(\alpha_n) \label{eq: prob. relax theorem a} \\
			\st & \quad \mathbf{x}\in\mathcal{X}, \label{eq: prob. relax theorem b} \\
			& \quad 
			\boldsymbol{0} \preceq \boldsymbol{\alpha} \preceq \boldsymbol{1}, \label{eq: prob. relax theorem c} 
		\end{align}					
	\end{subequations}
	where $ f_p(\alpha_n)\triangleq\varrho_n(\alpha_n^2-\alpha_n) $ is the PF, with $ \varrho_n>0 $ as the constant penalty parameter w.r.t.  $ \alpha_n $, and $ \mathcal{N}\triangleq\{1,2,\dots,N_{\alpha}\} $. If $ \mathcal{X} $ is a compact convex set, the feasible region of \eqref{eq: prob. relax theorem} is a compact convex set.
\end{theorem} 
\begin{proof}
	Please see Appendix \ref{app: relax. prob with pen.}.
\end{proof}

\begin{remark} \label{rem: pen. par. selection}
	The values of $ \varrho_n>0 $ can be theoretically selected as in \eqref{eq: weight of pen. func.} of Appendix \ref{app: relax. prob with pen.} to guarantee an optimal solution. In implementation, $ \varrho_n $ should not be too large to adapt to the error tolerance of the binary variables, such that the corresponding values of $\sum\nolimits_{n\in\mathcal{N}} f_p(\alpha_n)$ can be comparable to the actual objective value of $f_0(\mathbf{x}, \boldsymbol{\alpha})$. In addition, a small value of $\varrho_n$ yields to a significant gap between the penalty and objective functions, leading to a slow convergence. Hence, $ \varrho_n $ must be adaptively selected (e.g., w.r.t. the iteration index of the  iterative algorithm) to speed up the convergence. 
\end{remark}

For given binary variables $ \boldsymbol{\alpha} $ and $ \boldsymbol{\beta} $, the convex program \eqref{eq: prob. general form - relax. 1} can be reformulated  as 
\begin{subequations} \label{eq: prob. general form - relax. pen.}
	\setlength{\jot}{0.05pt}
	\begin{align}
		\underset{\substack{\mathbf{w}, \mathbf{p},\boldsymbol{\alpha},\boldsymbol{\beta}, \\ \boldsymbol{\omega},\boldsymbol{\lambda},\boldsymbol{\mu},\boldsymbol{\nu}}}{\max} & \quad \ddot{R}_{\Sigma}^{(\kappa+1)} \label{eq: prob. general form - relax. pen. a} \\
		\st & \quad \eqref{eq: prob. general form tensor :: b}, \eqref{eq: prob. general form tensor :: k}, \eqref{eq: prob. general form tensor :: l}, \eqref{eq: prob. general form e},  \eqref{eq: positive effective channel DL 1,k 1}, 
		\eqref{eq: var. change mu}, \eqref{eq: var. change omega convex - relax. 1},  \eqref{eq: var. change omega convex 2 1}, \eqref{eq: positive effective channel DL 2,k 1},  \eqref{eq: var. change nu}, \eqref{eq: abs function smooth approx.}, \eqref{eq: prob. general form - relax. 1 c}, \eqref{eq: prob. general form - relax. 1 d}, \label{eq: prob. general form - relax. pen. b} \\
		& \quad 
		\eqref{eq: prob. general form tensor :: j},\eqref{eq: prob. general form d} \label{eq: prob. general form - relax. pen. c},
	\end{align}							
\end{subequations}
where  constraints \eqref{eq: prob. general form - relax. pen. b} and \eqref{eq: prob. general form - relax. pen. c} are specified by  sets \eqref{eq: prob. mixed-integer theorem b} and \eqref{eq: prob. mixed-integer theorem c}, respectively. By  \textbf{Theorem}  \textbf{\ref{thm: relax. prob with pen.}}, the CR problem of \eqref{eq: prob. general form - relax. pen.}  with  the PF can be expressed as
\begin{subequations} \label{eq: prob. general form - relax. pen. 1}
	\setlength{\jot}{0.05pt}
	\begin{align}
		\underset{\substack{\mathbf{w}, \mathbf{p},\boldsymbol{\alpha},\boldsymbol{\beta}, \\ \boldsymbol{\omega},\boldsymbol{\lambda},\boldsymbol{\mu},\boldsymbol{\nu}}}{\max} & \quad \ddot{R}_{\Sigma}^{(\kappa+1)}+\sum\nolimits_{k\in\mathcal{K}}\sum\nolimits_{j\in\mathcal{K}}f_p^{\dl}(\alpha_{kj})+\sum\nolimits_{\ell\in\mathcal{L}}\sum\nolimits_{m\in\mathcal{L}}f_p^{\ul}(\beta_{\ell m}) \qquad\; \label{eq: prob. general form - relax. pen. 1 a} \\
		\st & \quad \eqref{eq: prob. general form - relax. pen. b}, \label{eq: prob. general form - relax. pen. 1 b} \\
		& \quad \eqref{eq: prob. general form - relax. c},\eqref{eq: prob. general form - relax. d}, \label{eq: prob. general form - relax. pen. 1 c}
	\end{align}							
\end{subequations}
where $ f_p^{\dl}(\alpha_{kj})\triangleq\varrho_{kj}^{\dl}(\alpha_{kj}^2-\alpha_{kj}) $ and $ f_p^{\ul}(\beta_{\ell m})\triangleq\varrho_{\ell m}^{\ul}(\beta_{\ell m}^2-\beta_{\ell m}) $, with $ \varrho_{kj}^{\dl}>0 $ and $ \varrho_{\ell m}^{\ul}>0 $ being the constant penalty parameters. 

{\hili
\begin{remark} In \eqref{eq: prob. general form - relax. pen. 1},  the penalty terms $(\alpha_{kj}^2-\alpha_{kj})$ and $(\beta_{\ell m}^2-\beta_{\ell m})$ are always non-positive, and thus,  are very useful to assess the degree of satisfaction of  binary
constraints \eqref{eq: prob. general form tensor :: j} and \eqref{eq: prob. general form d}. Moreover,  the positive values of $ \varrho_{kj}^{\dl} $ and $ \varrho_{\ell m}^{\ul} $ are of crucial  importance to recover binary values for $(\boldsymbol{\alpha},\boldsymbol{\beta})$ through the SEM problem. With a proper selection of penalty parameters to force the PF values to zeros, problem \eqref{eq: prob. general form - relax. pen. 1} is equivalent to problem \eqref{eq: prob. general form - relax. pen.} in the sense that they share the same objective value.  
\end{remark}
}

It is clear that the quadratic PF $\sum_{k\in\mathcal{K}}\sum_{j\in\mathcal{K}}f_p^{\dl}(\alpha_{kj})+\sum_{\ell\in\mathcal{L}}\sum_{m\in\mathcal{L}}f_p^{\ul}(\beta_{\ell m})$ provides strict convexity for the objective  \eqref{eq: prob. general form - relax. pen. 1 a} w.r.t. $(\boldsymbol{\alpha},\boldsymbol{\beta})$. In other words, problem \eqref{eq: prob. general form - relax. pen. 1} is always solvable. Here we apply the ICA method to iteratively approximate $f_p^{\dl}(\alpha_{kj})$ and $f_p^{\ul}(\beta_{\ell m})$ at the feasible points $\alpha_{kj}^{(\kappa)}$ and $\beta_{\ell m}^{(\kappa)}$, respectively, as
\begingroup
\allowdisplaybreaks
\begin{subequations}
	\setlength{\jot}{0.05pt}
	\begin{align}
	f_p^{\dl}(\alpha_{kj}) & \geq \varrho_{kj}^{\dl}\bigl((2\alpha_{kj}^{(\kappa)}-1)\alpha_{kj}-(\alpha_{kj}^{(\kappa)})^2\bigr) := \tilde{f}_p^{\dl,(\kappa)}(\alpha_{kj}), \\
	f_p^{\ul}(\beta_{\ell m}) & \geq \varrho_{\ell m}^{\ul}\bigl((2\beta_{\ell m}^{(\kappa)}-1)\beta_{\ell m}-(\beta_{\ell m}^{(\kappa)})^2\bigr)  := \tilde{f}_p^{\ul,(\kappa)}(\beta_{\ell m}).
	\end{align}
\end{subequations}\endgroup

{\color{black}
	\begin{algorithm}[t]
		\begin{algorithmic}[1]
			\fontsize{8.5}{9.5}\selectfont
			\protect\caption{\fontsize{9}{10}\selectfont Proposed ICA-CR Based Design with Penalty Function  for the Relaxed SEM Problem \eqref{eq: prob. general form - relax.}}

			\label{alg: Continuous relaxation problem - pen}
			
			\STATE \textbf{Initialization:} Set $ R_{\Sigma} := -\infty $ and $ \bigl(\mathbf{w}^{*},\mathbf{p}^{*},\boldsymbol{\alpha}^{*},\boldsymbol{\beta}^{*}\bigr):=\mathbf{0} $.

			\STATE Set $\kappa:=0$ and generate an initial feasible point $ \mathbf{X}^{(0)} $.
			
			\REPEAT
			\STATE Set $ \varrho_{kj}^{\dl} = \varrho_{\ell m}^{\ul}=\varrho,\;\forall k,j\in\mathcal{K},\; \forall \ell,m\in\mathcal{L} $. 
			\STATE Solve \eqref{eq: prob. general form - relax. pen. 2} to obtain  $ \mathbf{X}^{\star} $ and  $\mathring{R}_{\Sigma}^{(\kappa+1)}$.
			
			\STATE Update $ \mathbf{X}^{(\kappa+1)} :=\mathbf{X}^{\star} $.
			\STATE Set $ \kappa := \kappa + 1 $.
			\UNTIL Convergence
			
			\STATE\ Update $ \bigl(\mathbf{w}^{*},\mathbf{p}^{*}\bigr) := \bigl(\mathbf{w}^{(\kappa)},\mathbf{p}^{(\kappa)}\bigr)$ and ($ \boldsymbol{\alpha}^{*},\boldsymbol{\beta}^{*} $) as in \eqref{eq: rounding binary var.}.
			\STATE\ Use $ \bigl(\mathbf{w}^{*},\mathbf{p}^{*},\boldsymbol{\alpha}^{*},\boldsymbol{\beta}^{*}\bigr) $ to compute $ R_{\Sigma} $ as in \eqref{eq: prob. general form a}.

			\STATE {\textbf{Output:}  $ R_{\Sigma} $ and the optimal solution $ \bigl( \mathbf{w}^{*},\mathbf{p}^{*},\boldsymbol{\alpha}^{*},\boldsymbol{\beta}^{*}\bigr) $.}
		\end{algorithmic} 
	\end{algorithm}
}

At iteration $ \kappa+1 $ of the proposed algorithm, the successive convex program follows as
\begin{subequations} \label{eq: prob. general form - relax. pen. 2}
	\setlength{\jot}{0.05pt}
	\begin{align}
		\underset{\substack{\mathbf{w}, \mathbf{p},\boldsymbol{\alpha},\boldsymbol{\beta}, \\ \boldsymbol{\omega},\boldsymbol{\lambda},\boldsymbol{\mu},\boldsymbol{\nu}}}{\max} & \quad \mathring{R}_{\Sigma}^{(\kappa+1)}\triangleq \ddot{R}_{\Sigma}^{(\kappa+1)}+\tilde{f}_p^{(\kappa)} \label{eq: prob. general form - relax. pen. 2 a} \\
		\st & \quad \eqref{eq: prob. general form - relax. c}, \eqref{eq: prob. general form - relax. d}, \eqref{eq: prob. general form - relax. pen. b}, \label{eq: prob. general form - relax. pen. 2 b}
	\end{align}							
\end{subequations}
where $ \tilde{f}_p^{(\kappa)} \triangleq \sum_{k\in\mathcal{K}}\sum_{j\in\mathcal{K}}\tilde{f}_p^{\dl,(\kappa)}(\alpha_{kj})+\sum_{\ell\in\mathcal{L}}\sum_{m\in\mathcal{L}}\tilde{f}_p^{\ul,(\kappa)}(\beta_{\ell m}) $.
To summarize,  \textbf{Algorithm} \textbf{\ref{alg: Continuous relaxation problem - pen}} outlines the proposed ICA-CR-PF algorithm to solve \eqref{eq: prob. general form - relax.}. Without loss of optimality, the constant penalty parameters can be uniformly selected through  $ \varrho\triangleq\max\{\varrho_{kj}^{\dl}, \varrho_{\ell m}^{\ul}\}_{k,j\in\mathcal{K},\;\ell,m\in\mathcal{L}} $, as proved in Appendix \ref{app: relax. prob with pen.}. With an appropriate choice of $ \varrho $, the ICA-CR-PF algorithm  provides an optimal solution to the original problem \eqref{eq: prob. general form}. Further discussions on the choice of the  penalty parameter $\varrho$ are given in Section \ref{NumericalResults}-C.

{\hili
\subsection{Discussion on Solution for the General Case $ (Z>=2) $:}
	Our emphasis on the case study with $ Z=2 $ is merely due to current practical limitations of small cell FD-NOMA systems. However, finding an optimal solution for the general case of FD-NOMA systems (if practically implementable) is straightforward after slight modifications. To see this, by using the binary variable \tensor{T}, we first transform the SINR of DL users in \eqref{eq: DL SINR general form} into a more tractable form as:
\begin{align} \label{eq: DL SINR - ext}
\gamma_{ik}^{\dl}(\mathbf{w},\mathbf{p},\tensor{T}) & = \min\Biggl\{ \underset{ j\in\mathcal{K}}{\max}\Bigl\{\mfrac{T_{jk}^{1i}|(\mathbf{h}_{1j}^\dl)^H \mathbf{w}_{ik}|^2}{\Theta_{jk}^{1i}(\mathbf{w},\mathbf{p},\tensor{T})}\Bigr\}, \underset{z\in\mathcal{Z}_{i}\setminus \{1\} }{\min}\; \underset{ j\in\mathcal{K}}{\max}\Bigl\{\mfrac{T_{jk}^{zi}|(\mathbf{h}_{zj}^\dl)^H \mathbf{w}_{ik}|^2}{\Theta_{jk}^{zi}(\mathbf{w},\mathbf{p},\tensor{T})}\Bigr\} \Biggr\}, \; \forall i\in \mathcal{Z}.
\end{align}
It is observed that $ \tensor{T} $ includes one identity matrix and $ (Z-1) $ variable matrices, as equivalent to $ \mathbf{T}^{1i},\; \forall i\in\mathcal{Z}\setminus\{1\} $ in \textbf{Theorem} \textbf{\ref{thm: user association matrix}}. We then relax $ \mathbf{T}^{1i} $ to be continuous as $ \mathbf{T}^{1i} \in [0,1] $, while the matrix $ \mathbf{T}^{zi} $, $ \forall z \in\mathcal{Z}_{i}\setminus\{1\} $, is treated as a constant updated on $ \mathbf{T}^{1z} $ and $ \mathbf{T}^{1i} $. By applying the ICA method and exploiting the structure of $ \tensor{T} $, at each iteration, the UA variable matrix $ \mathbf{T}^{1i} $ is jointly optimized to find a minorant maximization, while $ \mathbf{T}^{zi} $ is calculated via \eqref{eq: user association matrix} with $ \mathbf{T}^{1z} $ and $ \mathbf{T}^{1i} $ given in the previous iteration. Consequently, all the fractional functions in \eqref{eq: DL SINR - ext} can be tackled by following the same steps as \eqref{eq: SINR zone-1 convexification}-\eqref{eq: positive effective channel DL 2,k 1}. Note that although \eqref{eq: DL SINR} can be expressed by \eqref{eq: DL SINR - ext}, the expressions \eqref{eq: DL SINR zone 1} and \eqref{eq: DL SINR zone 2} help to show up all possible cases of numerators and denominators of the fractions which need to be convexified in the general case.
}

\section{Optimal Solution based on Brute-force Search}\label{sec: sum rate maximization}
For benchmarking purposes, this section presents the BFS combining with the ICA method, where the optimization subproblems corresponding to all feasible cases of $(\boldsymbol{\alpha},\boldsymbol{\beta})$ are successively solved to find the optimal solution. We note that a typical approach for BFS is to generate all possible cases of  $( \boldsymbol{\alpha},  \boldsymbol{\beta} )$, and then check the feasibility through the branch-and-bound method before handling the optimization problem of the power control. However, this approach takes the complexity $ \mathcal{O}\bigl(2^{K^2}2^{L^2}\bigr) $ as the number of subproblems, which is computationally expensive even for networks of medium size. Alternatively, we generate the permutations of the orders for DL and UL users, which strictly satisfy  constraints \eqref{eq: prob. general form tensor :: k}-\eqref{eq: prob. general form tensor :: m} and \eqref{eq: prob. general form d}-\eqref{eq: prob. general form e}, to relieve the complexity of searching. The values of $ \boldsymbol{\alpha} $ and $ \boldsymbol{\beta} $ are then derived for each permutation. In this case, the number of possible cases is calculated from enumerative combinatorics $ (K!\times L!) $. {\hilidra In particular, the inner-zone DL users keep the orders as their indices, while the orders of the outer-zone DL users are permuted in all possible cases, i.e., a set $ \mathcal{P}_K $ contains $ K $-permutations of $ \{1,2,\dots,K\} $. For each permutation of $ \mathcal{P}_K $, one inner-zone DL user is paired with one outer-zone DL user with the same index.} For UL decoding, we also generate $ L $-permutations of $ \{1,2,\dots,L\} $, establishing the set $ \mathcal{P}_L $, in which each permutation represents the UL users' decoding order. Therefore,   the SEM subproblem of \eqref{eq: prob. general form} generated by each element of $ \mathcal{P}_K\times\mathcal{P}_L $ is expressed as
\begin{subequations} \label{eq: prob. BB form}
	\setlength{\jot}{0.5pt}
	\begin{align}
		\underset{\mathbf{w}, \mathbf{p}}{\max} & \quad \sum\nolimits_{i\in\mathcal{Z}}\sum\nolimits_{k\in\mathcal{K}}R_{ik}^{\dl}\bigl(\mathbf{w}, \mathbf{p}|\boldsymbol{\alpha}\bigr)+\sum\nolimits_{\ell\in\mathcal{L}}R_{\ell}^{\ul}\bigl(\mathbf{w}, \mathbf{p}|\boldsymbol{\beta}\bigr) \qquad \label{eq: prob. BB form a} \\
		\st & \quad \eqref{eq: prob. general form tensor :: b}, \eqref{eq: prob. general form tensor :: c}, \label{eq: prob. BB form b} \\
		& \quad 
		R_{ik}^{\dl}(\mathbf{w}, \mathbf{p}|\boldsymbol{\alpha}) \geq \bar{R}_{ik}^{\dl}, \; \forall i\in\mathcal{Z},\;  k \in \mathcal{K}, \label{eq: prob. BB form e} \\
		& \quad 
		R_{\ell}^{\ul}(\mathbf{w}, \mathbf{p}|\boldsymbol{\beta}) \geq \bar{R}_{\ell}^{\ul}, \; \forall \ell \in \mathcal{L}, \label{eq: prob. BB form f} 
	\end{align}						
\end{subequations}
where $ R_{ik}^{\dl}(\mathbf{w}, \mathbf{p}|\boldsymbol{\alpha}) $ and $ R_{\ell}^{\ul}(\mathbf{w}, \mathbf{p}|\boldsymbol{\beta}) $ are obtained from $ R_{ik}^{\dl}(\mathbf{w}, \mathbf{p},\boldsymbol{\alpha}) $ and $ R_{\ell}^{\ul}(\mathbf{w}, \mathbf{p},\boldsymbol{\beta}) $ for given values of $ \boldsymbol{\alpha} $ and $ \boldsymbol{\beta} $, respectively. 

Even for  subproblem \eqref{eq: prob. BB form} at hand, the optimization problem in $(\mathbf{w}, \mathbf{p})$ remains non-convex. Specifically,  subproblem \eqref{eq: prob. BB form} has a non-concave objective  \eqref{eq: prob. BB form a} and  non-convex QoS constraints \eqref{eq: prob. BB form e} and \eqref{eq: prob. BB form f}. Fortunately, we will show shortly that the developments presented in Section \ref{sec: Continuous Relaxation Problems} are useful to convexify the non-convex parts of  \eqref{eq: prob. BB form}. In what follows, the notations defined in the previous section will be reused, unless otherwise mentioned.
Similarly to \eqref{eq: DL rate concave minorant1}, the concave minorant of $ R_{1k}^{\dl}(\mathbf{w}, \mathbf{p}|\boldsymbol{\alpha}) $ at the $ (\kappa+1) $-th iteration is 
\begin{align} 
R_{1k}^{\dl}(\mathbf{w},\mathbf{p}|\boldsymbol{\alpha}) & \geq \ln\bigl(1+\mfrac{1}{\omega_{1k}}\bigr) 
\geq \mathtt{A}(\omega_{1k}^{(\kappa)}) + \mathtt{B}(\omega_{1k}^{(\kappa)})\omega_{1k} := \tilde{R}_{1k}^{\dl, (\kappa)}, \label{eq: DL rate concave minorant}
\end{align}
with the trust region in \eqref{eq: positive effective channel DL 1,k 1} and the following  constraint
\begin{equation} \label{eq: var. change omega convex - relax. 2}
\phi_{k}(\mathbf{w},\mathbf{p}|\boldsymbol{\alpha}) \leq \omega_{1k} \tilde{\gamma}_{1k}^{(\kappa)}\bigl(\mathbf{w}\bigr), \; \omega_{1k} > 0,
\end{equation}
where $ \phi_{k}(\mathbf{w},\mathbf{p}|\boldsymbol{\alpha}) $ is the IN of $ \DLUi{1k} $ for a given value of $ \boldsymbol{\alpha} $. Different from $ \phi_{k}(\mathbf{w},\mathbf{p},\boldsymbol{\alpha}) $, $ \phi_{k}(\mathbf{w},\mathbf{p}|\boldsymbol{\alpha}) $ is a quadratic function, leading to the convex constraint \eqref{eq: var. change omega convex - relax. 2}.   The SINR of $ \DLUi{2j}$ in \eqref{eq: DL SINR zone 2} is reformulated as
\begin{equation} 
\hspace{-0pt} \gamma_{2j}^{\dl}(\mathbf{w},\mathbf{p}|\boldsymbol{\alpha}) = \min\Bigl\{ \mfrac{|(\mathbf{h}_{1k}^\dl)^H \mathbf{w}_{2j}|^2}{\psi_{j}^{k}(\mathbf{w},\mathbf{p})}, \mfrac{|(\mathbf{h}_{2j}^\dl)^H \mathbf{w}_{2j}|^2}{\varphi_{j}(\mathbf{w},\mathbf{p})}\Bigr\}, \hspace{-7pt} \label{eq: DL SINR zone 2 - fixed alpha}
\end{equation} 
for $ \alpha_{kj}=1 $. From \eqref{eq: DL rate var. change 2 1}, the concave minorant of $ R_{2j}^{\dl}\bigl(\mathbf{w}, \mathbf{p}|\boldsymbol{\alpha}\bigr) $ at the $ (\kappa+1) $-th iteration is given as
\begin{align} 
R_{2j}^{\dl}(\mathbf{w},\mathbf{p}|\boldsymbol{\alpha}) & \geq \ln\Bigl(1+\mfrac{1}{\omega_{2j}}\Bigr)\geq \mathtt{A}(\omega_{2j}^{(\kappa)}) + \mathtt{B}(\omega_{2j}^{(\kappa)})\omega_{2j} := \tilde{R}_{2j}^{\dl, (\kappa)}, \label{eq: DL rate concave minorant 2} 
\end{align}
by imposing the constraints \eqref{eq: var. change omega 2 b 1} and replacing constraint \eqref{eq: var. change omega 2 a 1} with
$	\mfrac{|(\mathbf{h}_{1k}^\dl)^H \mathbf{w}_{2j}|^2}{\psi_{j}^{k}(\mathbf{w},\mathbf{p})} \geq \mfrac{1}{\omega_{2j}}$, which is innerly convexified by
	\begin{align} \label{eq: var. change omega convex 2}
	\psi_{j}^{k}(\mathbf{w},\mathbf{p}) & \leq \omega_{2j} \bar{\gamma}_{2,kj}^{(\kappa)}\bigl(\mathbf{w}\bigr),\; k\in \mathcal{K}_j\triangleq\{k\in\mathcal{K}|\alpha_{kj}=1\} , 
	\end{align}
over the trust regions:
	\begin{align}\label{eq: positive effective channel DL 2,k}
	\bar{\gamma}_{2,kj}^{(\kappa)}(\mathbf{w}) \triangleq & \; 2\Re \bigl\{\bigl((\mathbf{h}_{1k}^\dl)^H\mathbf{w}_{2j}^{(\kappa)}\bigr)^*\bigl((\mathbf{h}_{1k}^\dl)^H\mathbf{w}_{2j}\bigr)\bigr\}  - \bigl|(\mathbf{h}_{1k}^\dl)^H \mathbf{w}_{2j}^{(\kappa)}\bigr|^2  > 0, \; k\in \mathcal{K}_j.
	\end{align}
We note that $\gamma_{2j}^{\dl}(\mathbf{w},\mathbf{p}|\boldsymbol{\alpha})$ in \eqref{eq: DL SINR zone 2 - fixed alpha} can be simplified to $\gamma_{2j}^{\dl}(\mathbf{w},\mathbf{p}|\boldsymbol{\alpha})=|(\mathbf{h}_{2j}^\dl)^H \mathbf{w}_{2j}|^2/\varphi_{j}(\mathbf{w},\mathbf{p})$ for  $\alpha_{kj}=0,\forall k,j$, and thus constraints \eqref{eq: var. change omega convex 2} and \eqref{eq: positive effective channel DL 2,k} are neglected in this case.

Let us turn our attention to handle the non-concavity of $ R_{\ell}^{\ul}\bigl(\mathbf{w}, \mathbf{p}|\boldsymbol{\beta}\bigr) $. By following the development in \cite[Eq. (20)]{Dinh:TCOMM:2017}, a concave quadratic minorant of $ R_{\ell}^{\ul}(\mathbf{w}, \mathbf{p}|\boldsymbol{\beta})$ at the feasible point $(\mathbf{w}^{(\kappa)}, \mathbf{p}^{(\kappa)})$ is 
\begin{IEEEeqnarray} {cl} \label{eq: UL rate convex}
	R_{\ell}^{\ul}(\mathbf{w}, \mathbf{p}|\boldsymbol{\beta}) \geq  \tilde{\mathtt{A}}\bigl(\gamma_{\ell}^{\ul}\bigl(\mathbf{w}^{(\kappa)}, \mathbf{p}^{(\kappa)}|\boldsymbol{\beta}\bigr)\bigr) +  \mfrac{2\gamma_{\ell}^{\ul}\bigl(\mathbf{w}^{(\kappa)}, \mathbf{p}^{(\kappa)}|\boldsymbol{\beta}\bigr)}{p_{\ell}^{(\kappa)}}p_{\ell}  - \Phi_{\ell}^{(\kappa)}(\mathbf{w}, \mathbf{p}|\boldsymbol{\beta}) := \tilde{R}_{\ell}^{\ul,(\kappa)},\qquad
\end{IEEEeqnarray}
where 
\begingroup \setlength{\jot}{0.05pt}
\allowdisplaybreaks\begin{align}
\tilde{\mathtt{A}}\bigl(\gamma_{\ell}^{\ul}\bigl(\mathbf{w}^{(\kappa)}, \mathbf{p}^{(\kappa)}|\boldsymbol{\beta}\bigr)\bigr)&\triangleq \ln\bigl(1+\gamma_{\ell}^{\ul}\bigl(\mathbf{w}^{(\kappa)}, \mathbf{p}^{(\kappa)}|\boldsymbol{\beta}\bigr)\bigr)  -\gamma_{\ell}^{\ul}\bigl(\mathbf{w}^{(\kappa)}, \mathbf{p}^{(\kappa)}|\boldsymbol{\beta}\bigr), \nonumber \\
\Phi_{\ell}^{(\kappa)}(\mathbf{w}, \mathbf{p}|\boldsymbol{\beta})  &\triangleq
\ds\tr\Bigl(\bigl(p_{\ell}^2\mathbf{h}_{\ell}^{\ul}(\mathbf{h}_{\ell}^{\ul})^H + \boldsymbol{\Psi}_{\ell}\bigl)\bigl( (\boldsymbol{\Psi}_{\ell}^{(\kappa)})^{-1}  - ((p_{\ell}^{(\kappa)})^2\mathbf{h}_{\ell}^{\ul}(\mathbf{h}_{\ell}^{\ul})^H+\boldsymbol{\Psi}_{\ell}^{(\kappa)})^{-1}\bigr)\Bigr),   \nonumber \\
\boldsymbol{\Psi}_{\ell}^{(\kappa)} &\triangleq  \sum\nolimits_{m\in\mathcal{L}}\beta_{\ell m} \bigl(p_{m}^{(\kappa)}\bigr)^2 \mathbf{h}_{m}^{\ul}\bigl(\mathbf{h}_{m}^{\ul}\bigr)^H  + \rho^2\sum\nolimits_{i\in\mathcal{Z}}\sum\nolimits_{k\in\mathcal{K}} \mathbf{G}_{\SI}^H \mathbf{w}_{ik}^{(\kappa)} \bigl(\mathbf{w}_{ik}^{(\kappa)}\bigr)^H\mathbf{G}_{\SI} +\sigma_{\mathtt{U}}^2\mathbf{I}.\nonumber
\end{align}\endgroup

\begin{algorithm}[t]
	\begin{algorithmic}[1]
		\fontsize{8.5}{9.5}\selectfont
		\protect\caption{\fontsize{9}{10}\selectfont Proposed ICA-BFS Based Design for the SEM Problem \eqref{eq: prob. general form}}

		\label{alg: Brute-Force Search}
		
		\STATE \textbf{Initialization:} Set $ R_{\Sigma} := -\infty $ and $ \bigl(\mathbf{w}^{*},\mathbf{p}^{*},\boldsymbol{\alpha}^{*},\boldsymbol{\beta}^{*}\bigr)$\footnotemark $:=\mathbf{0} $.
		
		\FOR[solving subproblem \eqref{eq: prob. BB form}] {each   element of $ \mathcal{P}_K\times\mathcal{P}_L $}	
		
		\STATE Compute ($ \boldsymbol{\alpha},\boldsymbol{\beta} $) for each   element of $ \mathcal{P}_K\times\mathcal{P}_L $.
		\STATE Set $\kappa:=0$ and solve \eqref{eq: prob. BFS initialization} to generate an initial feasible point $ (\mathbf{w}^{(0)},\mathbf{p}^{(0)},\boldsymbol{\omega}^{(0)}) $.
		
		\REPEAT
		\STATE Solve \eqref{eq: prob. BB form convex prog.} to obtain  $ \bigl(\mathbf{w}^{\star},\mathbf{p}^{\star},\boldsymbol{\omega}^{\star}\bigr) $ and  $\tilde{R}_{\Sigma}^{(\kappa+1)}$.
		
		\STATE Update $(\mathbf{w}^{(\kappa+1)},\mathbf{p}^{(\kappa+1)}, \boldsymbol{\omega}^{(\kappa+1)}) :=(\mathbf{w}^{\star}, \mathbf{p}^{\star}, \boldsymbol{\omega}^{\star}) $.
		\STATE Set $ \kappa := \kappa + 1 $.
		\UNTIL Convergence
		\IF {$\tilde{R}_{\Sigma}^{(\kappa)}>R_{\Sigma} $}
		\STATE Update $ R_{\Sigma}:= \tilde{R}_{\Sigma}^{(\kappa)} $ and $ \bigl(\mathbf{w}^{*},\mathbf{p}^{*},\boldsymbol{\alpha}^{*},\boldsymbol{\beta}^{*}\bigr) := \bigl(\mathbf{w}^{(\kappa)},\mathbf{p}^{(\kappa)},\boldsymbol{\alpha},\boldsymbol{\beta}\bigr)$.
		\ENDIF
		\ENDFOR
		\STATE {\textbf{Output:}  $ R_{\Sigma} $ and the optimal solution $ \bigl( \mathbf{w}^{*},\mathbf{p}^{*},\boldsymbol{\alpha}^{*},\boldsymbol{\beta}^{*}\bigr) $.}
	\end{algorithmic} 
\end{algorithm}
\footnotetext{Note that $ \bigl( \mathbf{w}^{*},\mathbf{p}^{*},\boldsymbol{\alpha}^{*},\boldsymbol{\beta}^{*}\bigr) $ represents the optimal solution for \eqref{eq: prob. general form}, while $ \bigl(\mathbf{w}^{\star},\mathbf{p}^{\star},\boldsymbol{\omega}^{\star}\bigr) $ only denotes a per-iteration optimal solution for subproblem \eqref{eq: prob. BB form}.}

In summary, to solve  subproblem  \eqref{eq: prob. BB form}, the successive convex program  at the $ (\kappa+1) $-th iteration of the proposed algorithm is given by
\begin{subequations} \label{eq: prob. BB form convex prog.}
	\setlength{\jot}{0.01pt}
	\begin{align}
		\underset{\mathbf{w}, \mathbf{p},\boldsymbol{\omega}}{\max} & \quad \tilde{R}_{\Sigma}^{(\kappa+1)}\triangleq \sum\nolimits_{i\in\mathcal{Z}}\sum\nolimits_{k\in\mathcal{K}}\tilde{R}_{ik}^{\dl,(\kappa)}+\sum\nolimits_{\ell\in\mathcal{L}}\tilde{R}_{\ell}^{\ul,(\kappa)}\label{eq: prob. BB form convex prog. a} \\
		\st & \quad \eqref{eq: positive effective channel DL 1,k 1},  \eqref{eq: var. change omega convex 2 b 1}, \eqref{eq: positive effective channel DL 2,k b 1}, \eqref{eq: prob. BB form b}, \eqref{eq: var. change omega convex - relax. 2},  \eqref{eq: var. change omega convex 2}, \eqref{eq: positive effective channel DL 2,k}, \label{eq: prob. BB form convex prog. b} \quad\\
		& \quad 
		\tilde{R}_{ik}^{\dl,(\kappa)} \geq \bar{R}_{ik}^{\dl}, \; \forall i\in\mathcal{Z},\;  k \in \mathcal{K}, \label{eq: prob. BB form convex prog. c} \\
		& \quad 
		\tilde{R}_{\ell}^{\ul,(\kappa)} \geq \bar{R}_{\ell}^{\ul}, \; \forall \ell \in \mathcal{L}, \label{eq: prob. BB form convex prog. d} 
	\end{align}						
\end{subequations}
The iterative algorithm for solving \eqref{eq: prob. general form} based on the ICA-BFS method is summarized in \textbf{Algorithm} \textbf{\ref{alg: Brute-Force Search}}.

%

\section{Initial Point, Convergence and Complexity Analysis} \label{sec: Init, Converg, Complexity}

\subsection{Finding Initial Feasible Points}
\vspace{-5pt}
In general, the ICA framework used in solving a non-convex problem requires an initial feasible point to start the iterative algorithm, which is challenging to find due to the QoS constraints. A simple way using a random initial point may take many iterations or even fail to initialize the computational procedure.   Thus, it is nontrivial  to find an initial feasible point such that the proposed algorithms are successfully solved in the first iteration. Here we will present a heuristic way to generate a feasible point of the three proposed algorithms. 

\subsubsection{Initial Feasible Point of the ICA-CR and ICA-CR-PF Based Algorithms}
In some problems such as the one in \cite{Dinh:JSAC:Dec2017}, the feasible point can be found easily by violating the QoS constraints in \eqref{eq: prob. general form - relax.} and maximizing $\chi \triangleq \underset{i\in \mathcal{Z},k \in \mathcal{K},\ell \in \mathcal{L}}{\min}\{R_{ik}^{\dl}\bigl(\mathbf{w}, \mathbf{p},\boldsymbol{\alpha}\bigr) - \bar{R}_{ik}^{\dl}, R_{\ell}^{\ul}\bigl(\mathbf{w}, \mathbf{p},\boldsymbol{\beta}\bigr) - \bar{R}_{\ell}^{\ul}\}$ until reaching a positive objective value, i.e., $\chi \geq 0$. Unfortunately, such simple manipulations do not apply to our considered problem due to a joint UA in both DL and UL transmissions. For this reason, we propose a two-step initialization for the ICA-CR-based methods. In the ICA-CR based algorithm, we first randomly generate $ \boldsymbol{\beta} $ as $ \boldsymbol{\bar{\beta}} $ over the interval $(\boldsymbol{0},\boldsymbol{1})$ and solve the following convex program:
	\begin{align}\label{eq: prob. CRP initialization}
		\underset{\substack{\mathbf{w}, \mathbf{p},\boldsymbol{\alpha}, \boldsymbol{\beta}=\boldsymbol{\bar{\beta}}, \\ \boldsymbol{\omega},\boldsymbol{\lambda},\boldsymbol{\mu},\boldsymbol{\nu}}}{\max} & \quad \sum_{i\in\mathcal{Z}}\sum_{k\in\mathcal{K}}\ddot{R}_{ik}^{\dl,(\kappa)}, \qquad  
		\st \quad \eqref{eq: prob. general form - relax. 1 b}-\eqref{eq: prob. general form - relax. 1 d}.
	\end{align}							
In the second step, by using the optimal solution obtained from \eqref{eq: prob. CRP initialization} as the initial point and fixing $ \boldsymbol{\alpha} = \boldsymbol{\bar{\alpha}} $, we successively solve the following convex program:
	\begin{align}\label{eq: prob. CRP initialization 1}
		\underset{\substack{\mathbf{w}, \mathbf{p},\boldsymbol{\alpha}=\boldsymbol{\bar{\alpha}}, \boldsymbol{\beta}, \\ \boldsymbol{\omega},\boldsymbol{\lambda},\boldsymbol{\mu},\boldsymbol{\nu}}}{\max} & \quad \sum_{\ell\in\mathcal{L}}\ddot{R}_{\ell}^{\ul,(\kappa)},\quad 
		\st \qquad \eqref{eq: prob. general form - relax. 1 b}-\eqref{eq: prob. general form - relax. 1 d}.
	\end{align}							
Whenever problems \eqref{eq: prob. CRP initialization} and \eqref{eq: prob. CRP initialization 1} are found to be feasible, we output the initial feasible point as $\mathbf{X}^{(0)}\triangleq(\mathbf{w}^{(0)}, \mathbf{p}^{(0)},\boldsymbol{\alpha}^{(0)}, \boldsymbol{\beta}^{(0)},  \boldsymbol{\omega}^{(0)},\boldsymbol{\lambda}^{(0)},\boldsymbol{\mu}^{(0)},\boldsymbol{\nu}^{(0)})$ and start the ICA-CR based algorithm.

Similarly, the initial feasible point for the ICA-CR-PF based algorithm is easily found by replacing the objective functions in \eqref{eq: prob. CRP initialization} and  \eqref{eq: prob. CRP initialization 1} with $\sum_{i\in\mathcal{Z}}\sum_{k\in\mathcal{K}}\ddot{R}_{ik}^{\dl,(\kappa)} + \sum_{k\in\mathcal{K}}\sum_{j\in\mathcal{K}}\tilde{f}_p^{\dl,(\kappa)}(\alpha_{kj})$ and $\sum_{\ell\in\mathcal{L}}\ddot{R}_{\ell}^{\ul,(\kappa)} + \sum_{\ell\in\mathcal{L}}\sum_{m\in\mathcal{L}}\tilde{f}_p^{\ul,(\kappa)}(\beta_{\ell m}) $, respectively.

\subsubsection{Initial Feasible Point of the ICA-BFS Based Algorithm}
For given values of $(\boldsymbol{\alpha},\boldsymbol{\beta})$, we propose the one-step initialization for the ICA-BFS based algorithm inspired by \cite{Dinh:JSAC:Dec2017}. Particularly, let us consider the following modification of \eqref{eq: prob. BB form convex prog.}:
\begingroup
\allowdisplaybreaks\begin{subequations} \label{eq: prob. BFS initialization}
	\setlength{\jot}{0.05pt}
	\begin{align}
		\underset{\mathbf{w}, \mathbf{p},\boldsymbol{\omega}, \eta}{\max} & \quad \eta \label{eq: prob. BFS initialization :: a} \\
		\st 
		& \quad \eqref{eq: prob. BB form convex prog. b}, \\
		& \quad \tilde{R}_{ik}^{\dl, (\kappa)}-\bar{R}_{ik}^{\dl} \geq \eta, \; \forall i\in \mathcal{Z},  k \in \mathcal{K}, \label{eq: prob. BFS initialization :: c} \\
		& \quad \tilde{R}_{\ell}^{\ul, (\kappa)}-\bar{R}_{\ell}^{\ul} \geq \eta,\ \forall \ell \in \mathcal{L}, \label{eq: prob. BFS initialization :: d}
	\end{align}
\end{subequations}\endgroup
where $ \eta $ is a new variable.  Note that 
constraints \eqref{eq: prob. BFS initialization :: c} and \eqref{eq: prob. BFS initialization :: d} are satisfied if the objective  of \eqref{eq: prob. BFS initialization :: a} is a positive value. Initialized by any feasible point $ (\mathbf{w}^{(0)},\mathbf{p}^{(0)},\boldsymbol{\omega}^{(0)}) $ to the convex constraints in \eqref{eq: prob. BB form convex prog. b}, we successively solve \eqref{eq: prob. BFS initialization} at iteration $\kappa$ until reaching $\eta \geq 0 $.

\vspace{-10pt}
\subsection{Convergence Analysis}
\vspace{-5pt}
For the sake of notational convenience, we define the convex feasible sets of \eqref{eq: prob. general form - relax. 1}, \eqref{eq: prob. general form - relax. pen. 2}, and \eqref{eq: prob. BB form convex prog.} at iteration $ \kappa $ as
\begin{align}
	\mathcal{F}_{1}^{(\kappa)} & \triangleq \bigl\{\mathbf{X}^{(\kappa)}\bigr| \eqref{eq: prob. general form - relax. 1 b}-\eqref{eq: prob. general form - relax. 1 d} \text{ hold}\bigr\}, \\
	\mathcal{F}_{2}^{(\kappa)} & \triangleq \bigl\{\mathbf{X}^{(\kappa)}\bigr| \eqref{eq: prob. general form - relax. pen. 2 b} \text{ holds}\bigr\}, \\
	\mathcal{F}_{3}^{(\kappa)} & \triangleq \bigl\{(\mathbf{w}^{(\kappa)},\mathbf{p}^{(\kappa)},\boldsymbol{\omega}^{(\kappa)})\bigr| \eqref{eq: prob. BB form convex prog. b}-\eqref{eq: prob. BB form convex prog. d} \text{ hold}\bigr\}.
\end{align}
\begin{theorem}\label{pro:Convergence}
The proposed \textbf{Algorithms} $p\in\{1,2,3\}$ yield a sequence of  improved solutions converging to at least to a local optimum. 
\end{theorem}

\begin{proof}
The convergence analysis of the ICA-based algorithms is similar to that in \cite{Beck:JGO:10}. Thus, we only need to verify the convergent conditions of the proposed algorithms. 
First, we recall that the approximate functions presented in Section \ref{sec: Continuous Relaxation Problems} and \ref{sec: sum rate maximization} satisfy the properties listed in \cite[Property A]{Beck:JGO:10}. In other words, the optimal solutions achieved by the proposed algorithms at  iteration $\kappa$ are also feasible for the problem at  iteration $\kappa+1$, i.e., $ \mathcal{F}_{p}^{(\kappa)}\subseteq\mathcal{F}_{p}^{(\kappa+1)},\;p\in\{1,2,3\} $ \cite{Beck:JGO:10}. Moreover, $ \mathcal{F}_{p}^{(\kappa)} $ is closed and bounded due to the ICA method and  the power constraints \eqref{eq: prob. general form tensor :: b} and \eqref{eq: prob. general form tensor :: c}. According to \textbf{Theorem \ref{thm: relax. prob with pen.}}, the feasible sets $ \mathcal{F}_{p}^{(\kappa)} $ are compact and nonempty. Thus it satisfies the connectedness condition for Karush-Kuhn-Tucker (KKT) invexity as $ \kappa\rightarrow\infty $ \cite{Bestuzheva:InvexOp:Jul2017}. Therefore, the iterative solutions of the proposed convex programs towards the KKT-point are monotonically improved and converge at least to a local
optimum. For the practical implementation, a given error tolerance between two successive objective values can be used to terminate the proposed algorithms, i.e., $ d_{\mathtt{R}}\triangleq \mathtt{R}^{(\kappa+1)}-\mathtt{R}^{(\kappa)}<\varepsilon$ for $ \mathtt{R}\in\bigl\{\ddot{R}_{\Sigma},\mathring{R}_{\Sigma},\tilde{R}_{\Sigma}\bigr\} $.
\end{proof}

\vspace{-10pt}
\subsection{Complexity Analysis}
\vspace{-5pt}
We now provide  the worst-case per-iteration complexity analysis of \textbf{Algorithm} $ p\in\{1,2,3\} $. 
We first observe that the convex programs given in  \eqref{eq: prob. general form - relax. 1},  \eqref{eq: prob. general form - relax. pen. 2} and \eqref{eq: prob. BB form convex prog.} involve only the SOC and linear constraints, thus leading to a low computational complexity. For ease of presentation, let us define $ \upsilon_p $ and $ c_p $ be the numbers of  scalar variables and linear/SOC constraints, respectively. As a result, the per-iteration cost of solving  \eqref{eq: prob. general form - relax. 1}, \eqref{eq: prob. general form - relax. pen. 2} and \eqref{eq: prob. BB form convex prog.} is a polynomial time complexity of $ \mathcal{O}\bigl(c_p^{2.5}\upsilon_p^{2}+c_p^{3.5}\bigr) $ \cite{SeDuMi:2002}, which is detailed in Table \ref{tab: Complexities}.

{\hili From the above complexity estimates, it can be seen that   the per-iteration complexity of the ICA-CR and ICA-CR-PF based algorithms is higher than that of the ICA-BFS based algorithm for one subproblem due to the newly-added binary variables.} However, the computational complexity of the latter is extremely large when  the number of users increases resulting in a large number of subproblems, while the formers require solving  one convex program only.  

\begin{table*}[t]
	\centering
	\captionof{table}{Complexity Comparison}
	\label{tab: Complexities}
	\vspace{-10pt}
	\scalebox{0.7}{
		\begin{tabular}{l|c|c|c|c}
			\hline
			Method & \# of subproblems & $ \upsilon_p $ & $ c_p $ & Per-iteration Complexity\\
			\hline
			ICA-CR (\textbf{Alg. \ref{alg: Continuous relaxation problem}}) & 1 & $ \upsilon_{1}=\upsilon_{3}+3K^2+L^2+L $ & $ c_{1}=c_{3} + 6K^2+3L^2 $ & $ \mathcal{O}\bigl(c_1^{2.5}\upsilon_1^{2}+c_1^{3.5}\bigr) $ \\ \hline
			ICA-CR-PF (\textbf{Alg. \ref{alg: Continuous relaxation problem - pen}}) &  1 & $ \upsilon_{2}=\upsilon_{3}+3K^2+L^2+L $ & $ c_{2}=c_{3} + 6K^2+3L^2 $ & $ \mathcal{O}\bigl(c_2^{2.5}\upsilon_2^{2}+c_2^{3.5}\bigr) $ \\
			\hline\hline
			ICA-BFS (\textbf{Alg. \ref{alg: Brute-Force Search}})   & $ K!\times L! $ & $ \upsilon_{3}=2K(N+1)+L $ & $ c_{3}=8K+3L+1 $ & $ (K!\times L!).\mathcal{O}\bigl(c_3^{2.5}\upsilon_3^{2}+c_3^{3.5}\bigr) $ \\
			\hline		   				
	\end{tabular} }
	\vspace{5pt}
\end{table*}

{\hili
\begin{remark}
	Recall that $ N $, $ L $ and $ M = ZK $ denote the number of antennas at the BS, the number of UL users and the number of DL users, respectively, where $ Z $ is the number of zones for the general case. For any Z, the numbers of variables and constraints required by \textbf{Algorithm} \ref{alg: Brute-Force Search} are $ \upsilon_{3,Z}=NZK+ZK+L $ and $ c_{3,Z}=(Z+2)ZK+3L+1 $, respectively. \textbf{Algorithms} \ref{alg: Continuous relaxation problem} and \ref{alg: Continuous relaxation problem - pen} require the same numbers of variables and constraints as $ \upsilon_{1,Z}=\upsilon_{2,Z}=\upsilon_{3,Z}+3(Z-1)K^2+L^2+L $ and $ c_{1,Z}=c_{2,Z}=(Z-1)(Z+4)K^2+(5Z-2)K+3L^2+3L+1 $, respectively. The total complexities for $ Z\geq 2 $ are derived by replacing $ \upsilon_p $ and $ c_p $ in Table \ref{tab: Complexities} with $ \upsilon_{p,Z} $ and $ c_{p,Z} $, respectively. Note that the number of subproblems for ICA-CR-based methods is still one, while that for ICA-BFS is generalized as $ (K!)^{Z-1}\times L! $.
\end{remark}
}


\section{Numerical Results}\label{NumericalResults}

\subsection{Simulation Setup}
In this section, we numerically evaluate the performance of the developed algorithms in the MATLAB environment. The results are based on the following settings, unless otherwise specified. The radius of small-cell is set to 100 meters, and the BS located at the cell-center serves $L = 4$ UL users and $M = 8$ DL users. All UL users are randomly placed in the cell, while four DL users are randomly placed in zone-$1$ between 10 and 50 meters, and the other four DL users are randomly located in zone-2  between 50 and 100 meters. The channel vector between the BS and user $ \mathtt{U} \in \{\ULU, \DLU\} $ is generated as $\mathbf{h}=\sqrt{10^{-\mathtt{PL}_{\mathtt{BS,U}}/10}}\hat{\mathbf{h}}$, with  $\mathbf{h}\in\{\mathbf{h}_{\ell}^{\ul},\mathbf{h}_{ik}^{\dl}\}$, while the channel response from $\ULU$ to $\DLU$ is generated as $ g_{\ell k}=\sqrt{10^{-\mathtt{PL}_{\ell k}/10}}\hat{g}_{\ell k} $. Here, $\mathtt{PL}_{\mathtt{BS,U}}$ and $\mathtt{PL}_{\ell k}$ represent the path loss (in dB), as given in Table \ref{tab: parameter}, with $d_{\mathtt{BS,U}}$ and $d_{\ell k}$ being  the distances (in km) between BS and user $\mathtt{U}$ and  from $ \ULU $ to $ \DLU $, respectively.  $\hat{\mathbf{h}}$ and $ \hat{g}_{\ell k} $ are the small-scale fading and distributed as $\mathcal{CN}(\boldsymbol{0},\mathbf{I})$. The entries of the SI channel $ \mathbf{G}_{\SI} $ are modeled as independent and identically distributed Rician random variables, with the Rician factor of $ 5 $ dB \cite{Dinh:JSAC:18}. {\hili Unless specifically stated otherwise, the other parameters are set as given in Table~\ref{tab: parameter}, following the studies in \cite{Bharadia13,Bharadia14,Duarte:TWC:12,3GPP}.} The average SEs are plotted for 1000 topologies and 500 random channel realizations for each topology. The SEs are divided by $ \ln2 $ to be presented in bits/s/Hz. The scheme proposed in this paper is referred to as FD-NOMA with three different algorithms: ICA-CR (\textbf{Alg. \ref{alg: Continuous relaxation problem}}), ICA-CR-PF (\textbf{Alg. \ref{alg: Continuous relaxation problem - pen}}), and ICA-BFS (\textbf{Alg. \ref{alg: Brute-Force Search}}). For comparison purpose, the following schemes are also considered:
\begin{enumerate}
	\item ``FD-NOMA + RUA'':  Both UL users' decoding order and DL user pairing are randomly selected, which is referred to as the strategy of RUA. For a given random user association, the ICA-BFS based design is applied to handle the problem of power control.
	\item ``FD-Conventional'': A conventional FD scheme  in \cite{Dan:TWC:14} without  applying NOMA is used.
	\item ``HD-NOMA'': Under FD-NOMA with RUA, BS serves DL  and UL users separately in two  independent communication time blocks. In this scheme, there are no  SI and CCI, but the effective SEs are divided by two.
\end{enumerate}

\begin{table}[t]
	\centering
	\captionof{table}{Simulation Parameters}
	\label{tab: parameter}
	\vspace{-10pt}
	\scalebox{0.68}{
		\begin{tabular}{l|l}
			\hline
			Parameter & Value \\
			\hline\hline
			Radius of small cell                   & 100 m \\
			Noise power at receivers, $\sigma^2_{ik}=\sigma^2_{\mathtt{U}}\equiv\sigma^2$ & -104 dBm \\
			Residual SiS parameter, $\rho^2$                                          & -90 dB\\
			Distance from BS to closest users & $>$ 10 m\\
			PL  between BS and a user $\mathtt{U}$,  $\mathtt{PL}_{\mathtt{BS,U}}$ & 103.8 + 20.9$\log_{10}(d_{\mathtt{BS,U}})$ dB\\
			PL from $\ULU$ to $\DLU$,	$\mathtt{PL}_{\ell k} $ & 145.4 + 37.5$\log_{10}(d_{\ell k})$ dB\\
			Power budget at UL users, $P_{\ell}^{\max},\forall\ell $ & 18 dBm \\
			Power budget at BS, $ P_{\mathtt{bs}}^{\max}$  & 38 dBm \\
			Number of antennas at BS, $ N $ & 10\\
			Rate threshold, $ \bar{R}_{\ell}^{\ul} = \bar{R}_{ik}^{\dl} \equiv \bar{R},\;\forall \ell, i, k $	&  1 bits/s/Hz\\
			Error tolerance, $\varepsilon$ & $10^{-3}$\\
			\hline		   				
		\end{tabular}
	}
\end{table}

\subsection{Performance Evaluation}
\vspace{-5pt}
	\begin{figure}[t]
	\begin{subfigure}[Average SE versus $ P_{\mathtt{bs}}^{\max}$ for different schemes.]
				{
			\includegraphics[width=0.45\columnwidth, trim={0cm, 0cm, 0cm, 0cm}]{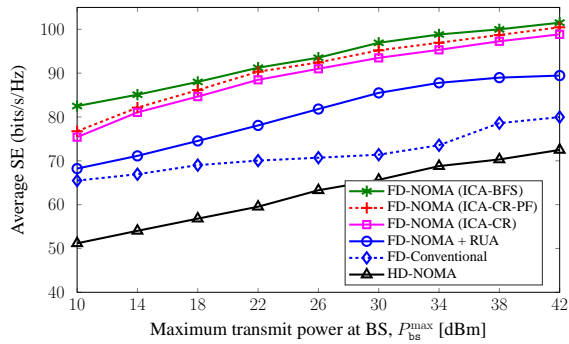}
						\label{fig: SumRate vs Pbs}
					}
	\end{subfigure}
		\hfill
	\begin{subfigure}[Average SE versus $ P_{\mathtt{bs}}^{\max}$ with RUA for DL and/or UL transmission.]
		{\includegraphics[width=0.45\columnwidth, trim={0cm, 0cm, 0cm, 0cm}]{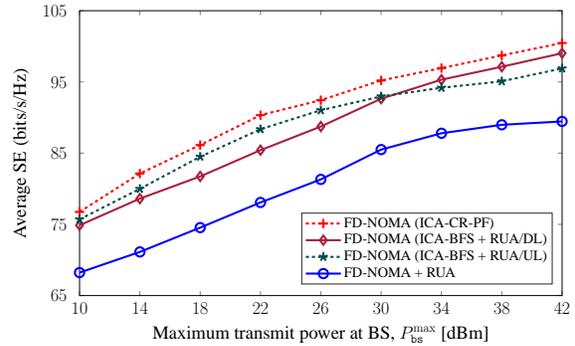}
						\label{fig: Gain SumRate vs Pbs}
				}
	\end{subfigure}
	\vspace{-5pt}
	\caption{Average SE versus the BS transmit power, $ P_{\mathtt{bs}}^{\max}$.}
	\label{fig: Aver SumRate vs Pbs}
\end{figure}
\hspace{6pt} \vrule \hspace{2pt}
\begin{figure}[t] 
	\centering
	\begin{subfigure}
		[Average SE versus $N$.]
		{
\includegraphics[width=0.45\columnwidth, trim={0cm, 0cm, 0cm, 0cm}]{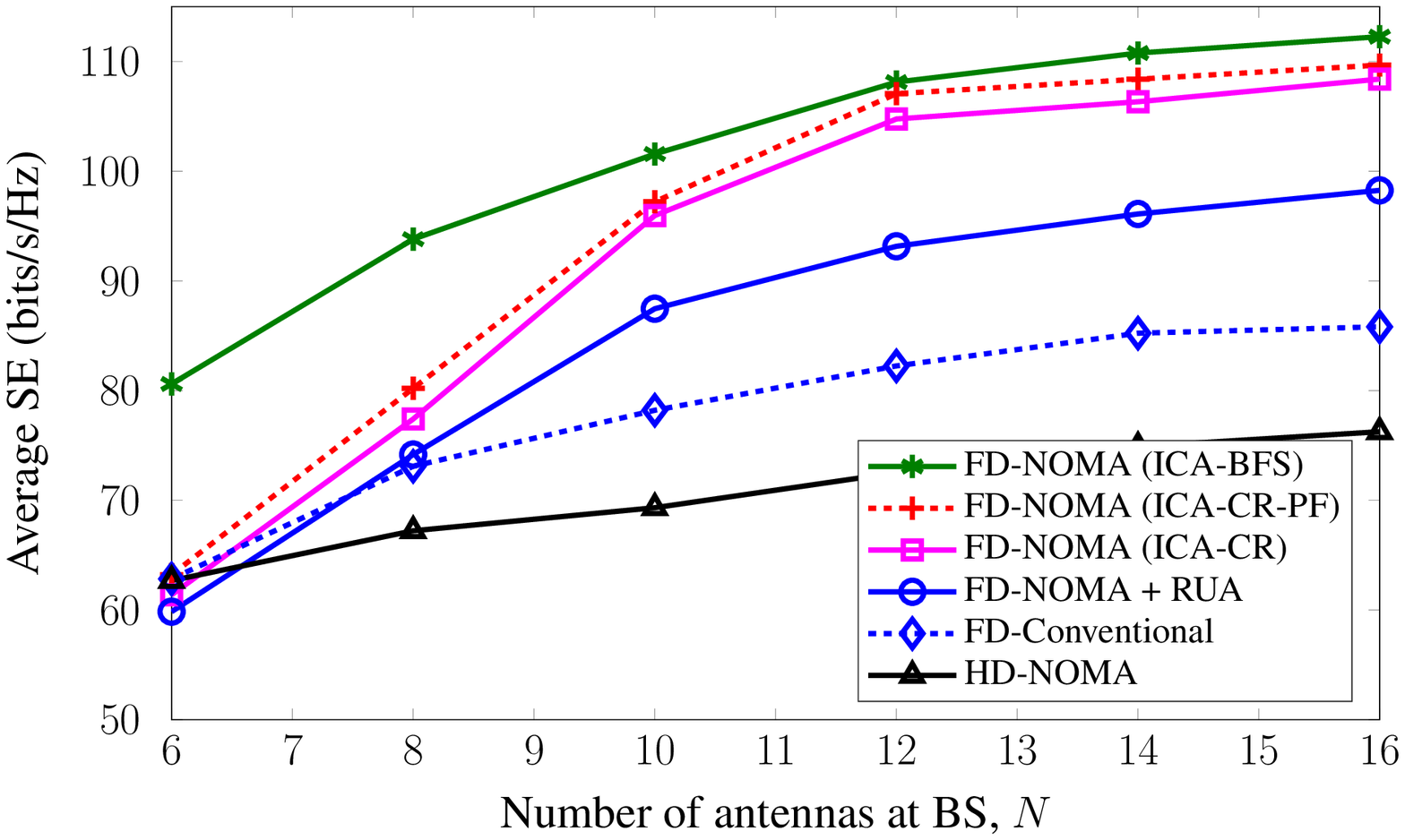}
			\label{fig: SumRate vs No. Antennas}
		}
	\end{subfigure}
	\hfill 
	\begin{subfigure}
	[Percentage of loss-SE versus  $N$.]
		{\includegraphics[width=0.45\columnwidth, trim={-0cm, 0cm, 0cm, 0cm}]{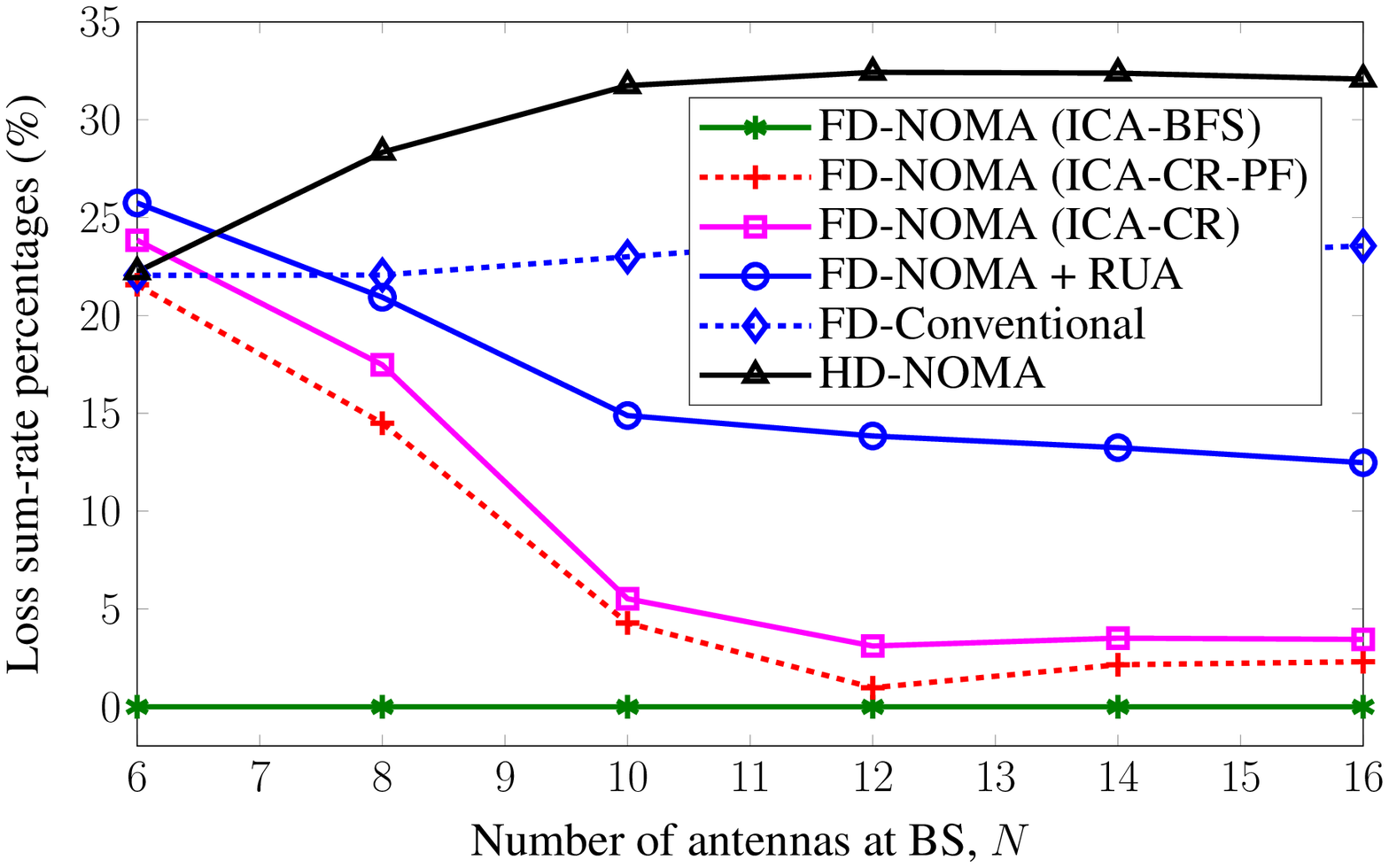}
			\label{fig: LossRate vs No. Antennas}
		}
	\end{subfigure}
	\vspace{-8pt}
	\caption{The change of average SE versus the number of antennas at the BS, $N$.}
\end{figure}


Fig.~\ref{fig: Aver SumRate vs Pbs} depicts the average SEs versus the maximum transmit power at the BS for different resource allocation schemes. As expected in Fig. \ref{fig: SumRate vs Pbs}, the ICA-BFS based algorithm provides the best SE due to finding the best UA for the FD-NOMA. However, we recall that it requires extremely high complexity, and thus, only plays as a benchmark. In addition, the ICA-CR-PF based algorithm tends to outperform the ICA-CR based algorithm as it aims at finding a high-performance UA solution.  On the other hand, it can be seen that both ICA-CR based algorithms deviate only 1$\%\sim2\%$ from the optimal SE, meaning that  performances are very good but with much less complexity compared to the ICA-BFS based algorithm. Unsurprisingly, our proposed FD-NOMA schemes outperform the conventional ones.  Fig. \ref{fig: Gain SumRate vs Pbs} further demonstrates the role of UA in DL and UL transmissions. Upon utilizing the BFS method, two other cases are considered: ($i$) ``ICA-BFS + RUA/DL'': We use the ICA-BFS for optimizing UL users' decoding order and random DL user pairing, and $(ii)$ ``ICA-BFS + RUA/UL'': The ICA-BFS is used for random UL users' decoding order and optimizing DL user pairing. As shown in Fig. \ref{fig: Gain SumRate vs Pbs}, the optimization of UA  in either DL or UL transmission enjoys a significant improvement of the SE as compared to FD-NOMA with RUA. The results also show that the UA in both DL and UL transmissions has significant influence on the performance. Although performances of ``ICA-BFS + RUA/DL'' and ``ICA-BFS + RUA/UL'' are comparable, they are distinguished from each other when $ P_{\mathtt{bs}}^{\max}$ varies. It is clear that increasing $ P_{\mathtt{bs}}^{\max} $ merely assists the DL transmission. The performance of ``ICA-BFS + RUA/DL'' is inferior when $ P_{\mathtt{bs}}^{\max}<30 $ dBm, since the total SE  is mainly determined by the DL transmission. Increasing $ P_{\mathtt{bs}}^{\max}$ brings much benefit to the optimization of UL users' decoding order as the performance of ``ICA-BFS + RUA/UL'' is lower than that of ``ICA-BFS + RUA/DL'' when  $ P_{\mathtt{bs}}^{\max} > 32$ dBm.  These results again confirm the important roles of user association in FD-NOMA systems.


In Fig. \ref{fig: SumRate vs No. Antennas}, we plot the average SE versus the number of antennas at the BS, $N\in[6,16]$. First, for small $N$ (i.e., $N\leq 6$),  ICA-CR based schemes provide the same performance (even worse) compared to the conventional ones. This is because the BS has a limited degrees-of-freedom to exploit multiuser diversity, which may result in a severe network interference situation in FD-NOMA based schemes. In this case, the use of UA in ICA-CR based schemes becomes less efficient. However, when $N$ increases, the SE gains of the proposed FD-NOMA schemes over the other ones are remarkable.  The reason is that the BS in FD-NOMA has sufficient degrees-of-freedom to select the best UA solutions, without causing much interference to the other users.
To further comprehend the benefit of using UA, Fig. \ref{fig: LossRate vs No. Antennas} shows the percentages of  loss-SE in comparison with the optimal performance obtained by the ICA-BFS.


Fig. \ref{fig: CDF vs Rate Threshold} shows the cumulative distribution function (CDF) of the FD-based schemes as a function of  the QoS requirement, $ \bar{R} $. The HD-NOMA scheme is omitted here  since its DL and UL transmissions are separately executed. It can be seen that the probabilities of feasibility of all the considered schemes are smaller when $ \bar{R} $ is higher. As expected, FD-NOMA schemes can maintain much rate fairness among all the DL and UL users as compared to the FD-Conventional scheme. In addition,  FD-NOMA schemes using the ICA-CR and ICA-CR-PF based algorithms offset about 0.5 bits/s/Hz and 2 bits/s/Hz of the rate threshold more than the scheme of FD-NOMA with RUA, respectively, in about 50$\%$ of the simulated trials. It further confirms that the joint optimization of UA  might help satisfy  higher QoS levels for the FD-NOMA schemes. Clearly, the ICA-CR-PF based algorithm assisted by the PF outperforms the ICA-CR one. The reason is that the former can quickly find the satisfactory solution of DL user pairing and UL users' decoding, at which the QoS constraints are satisfied even under unexpected channel conditions. The CDFs w.r.t. the rate threshold validate the advantage of UA, and reflect the characteristics of the proposed algorithms.

\begin{figure}
	\vspace{-35pt}
	\centering
	\begin{minipage}[t]{0.45\columnwidth}
		\centering
		\includegraphics[width=\columnwidth, trim={0cm, 0cm, 0cm, 0cm}]{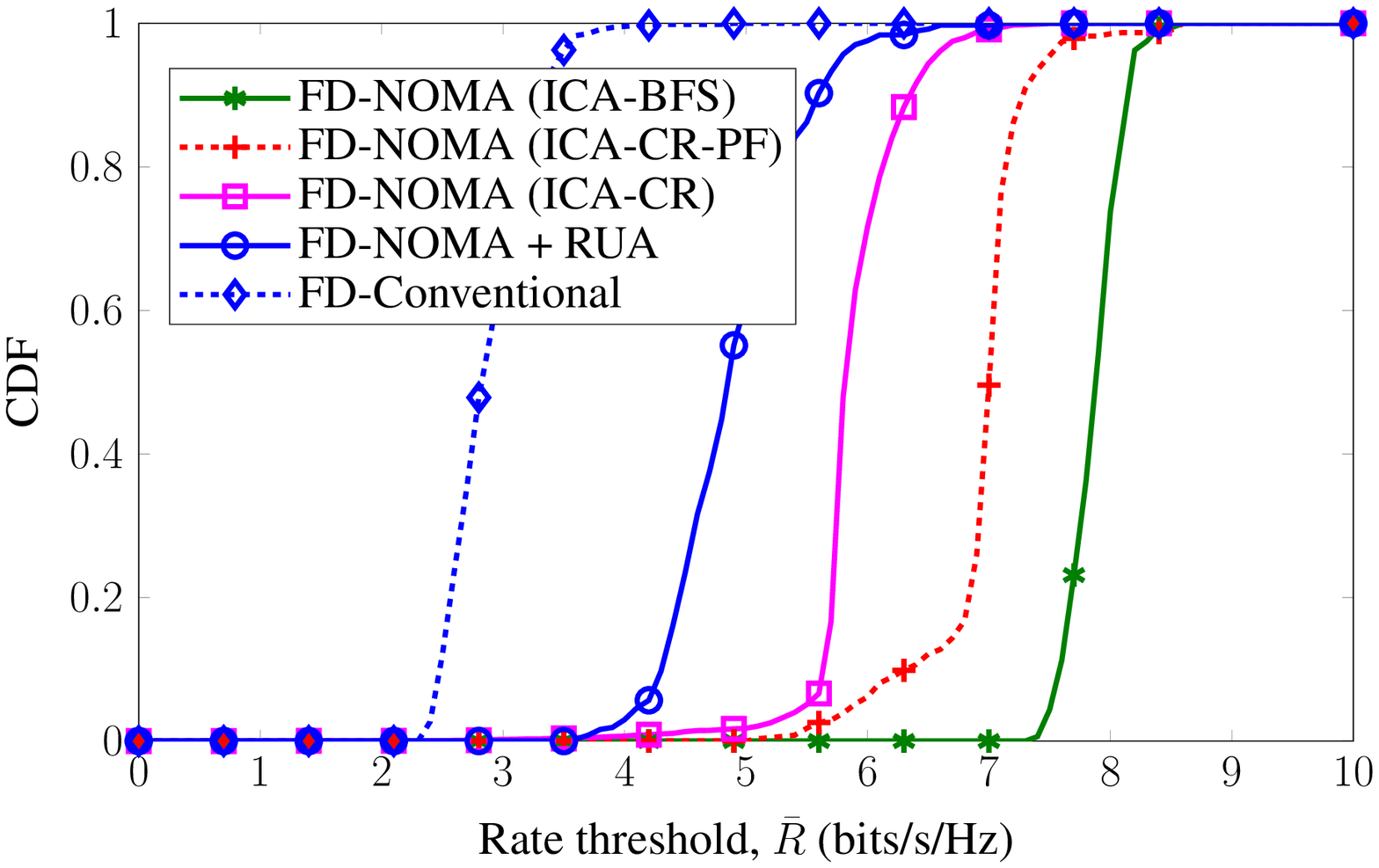} 
				\vspace{-20pt}
		\caption{Cumulative distribution function versus the rate threshold.}
		\label{fig: CDF vs Rate Threshold}
	\end{minipage}
	\hfill
	\begin{minipage}[t]{0.45\columnwidth}
				\centering
	\includegraphics[width=\columnwidth, trim={0cm, 0cm, 0cm, 0cm}]{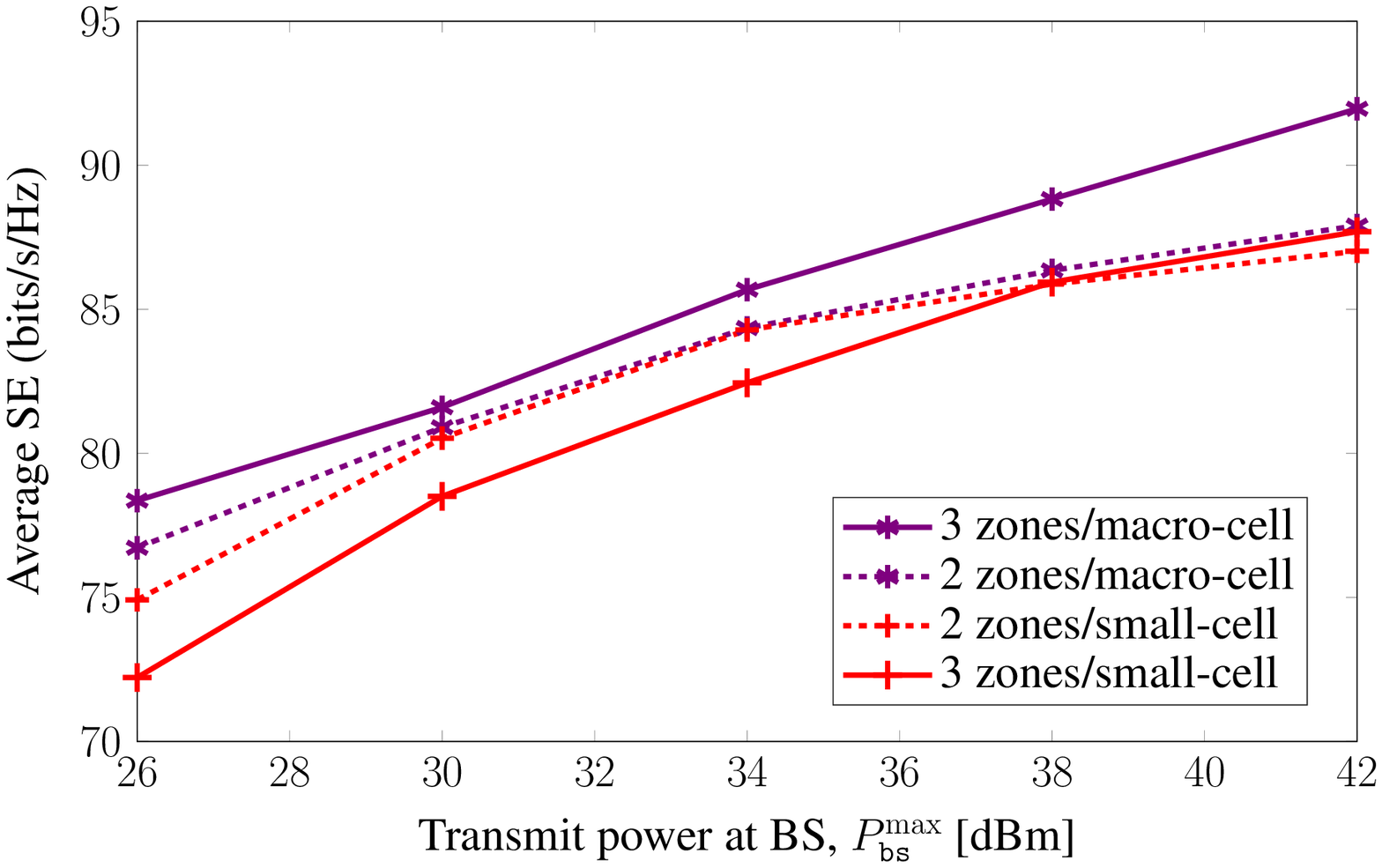}
		\vspace{-20pt}
		\caption{Average SE versus $ P_{\mathtt{bs}}^{\max}$ for different numbers of zones/clusters using the ICA-CR-PF based algorithm.}
		\label{fig: SumRate vs Pbs - multizones}
	\end{minipage}
	\vspace{-10pt}
\end{figure}

As mentioned previously in Section \ref{System Model and Problem Formulation}-C, we now provide simulation results for two- and three-zone in  scenarios of small- (100-meter radius) and macro-cells (500-meter radius), as illustrated in Fig. \ref{fig: SumRate vs Pbs - multizones}. The number of UL users is the same as before. We place 12 DL users to fairly compare the system performance of 2- and 3-zone DL transmissions, where NOMA is applied to 6 clusters (2 users in each) and 4 clusters (3 users in each), respectively. We use the ICA-CR-PF based algorithm for this investigation since it provides a good performance, with low complexity per iteration and fast convergence rate. In general, DL NOMA used for the macro-cell offers more efficient than that in small-cell, since DL users in the same cluster in macro-cell have significantly different channel gains. In addition, the dense multi-user interference in the small-cell deteriorates the system performance. Numerically, it is observed that the 3-zone NOMA  provides the best SE for the macro-cell. However,  it tends to  perform the worst in the small-cell. This is because  FD-NOMA systems for the small-cell scenario suffer from both the similar channel conditions and mutually strong interference. These results corroborate that in a small-cell scenario, the NOMA using 2-zone model outperforms that using the 3-zone one. In other words, a number of DL clusters should be properly chosen depending on the realistic scenarios.

{\hili 
\subsection{Effects of SI, CCI and Rate Threshold}

\begin{figure}[t]
	\centering
	\vspace{-35pt}
	\begin{subfigure}[Average DURR versus residual SiS, $\rho^2$.]
		{	
		\includegraphics[width=.45\columnwidth, trim={0cm, 0cm, 0cm, 0cm}]{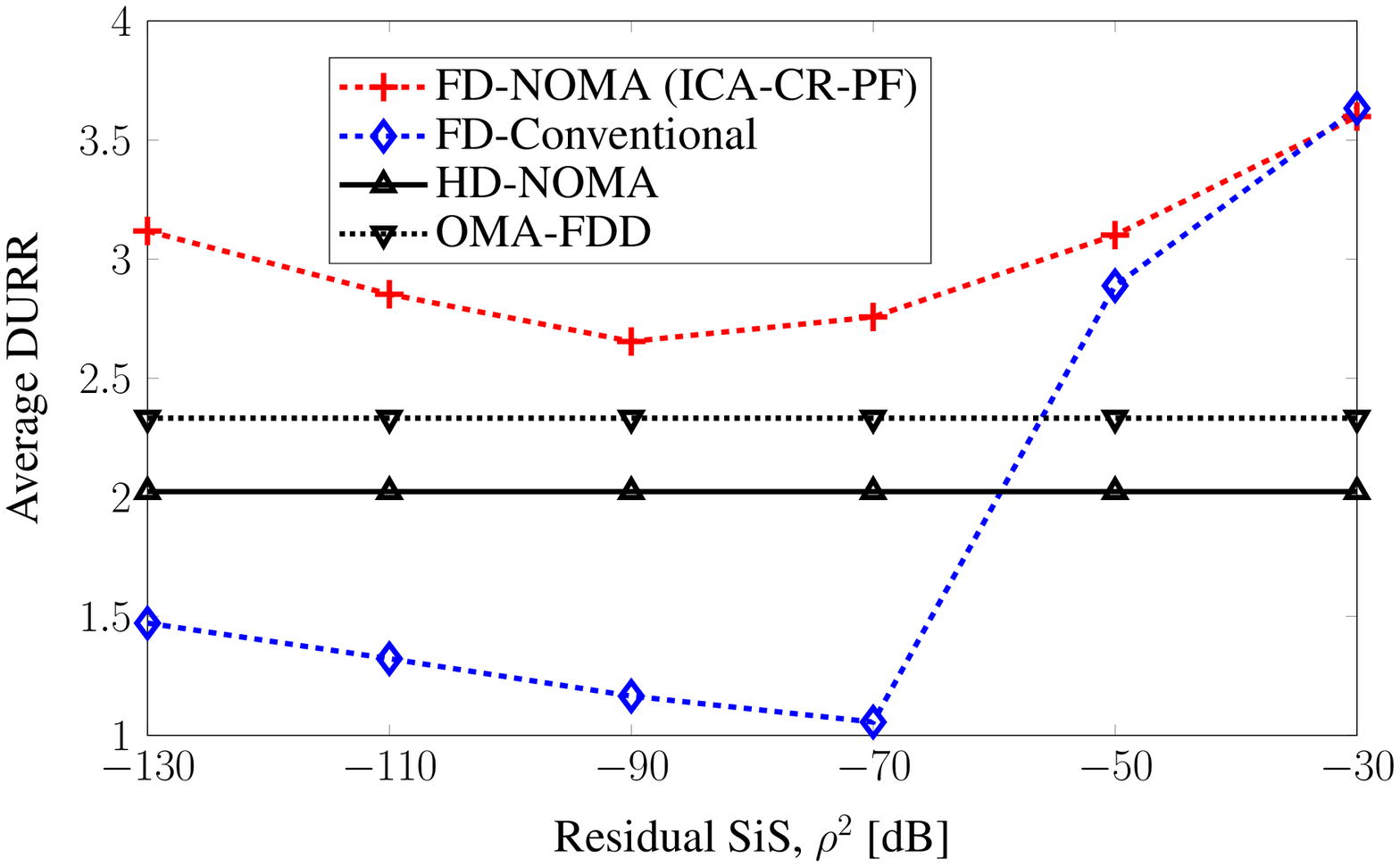}
			\label{fig: SumRate vs rho a}
		}
	\end{subfigure}
	\begin{subfigure}[Average SE versus residual SiS, $\rho^2$.]
		{\includegraphics[width=.45\columnwidth, trim={0cm, 0cm, 0cm, 0cm}]{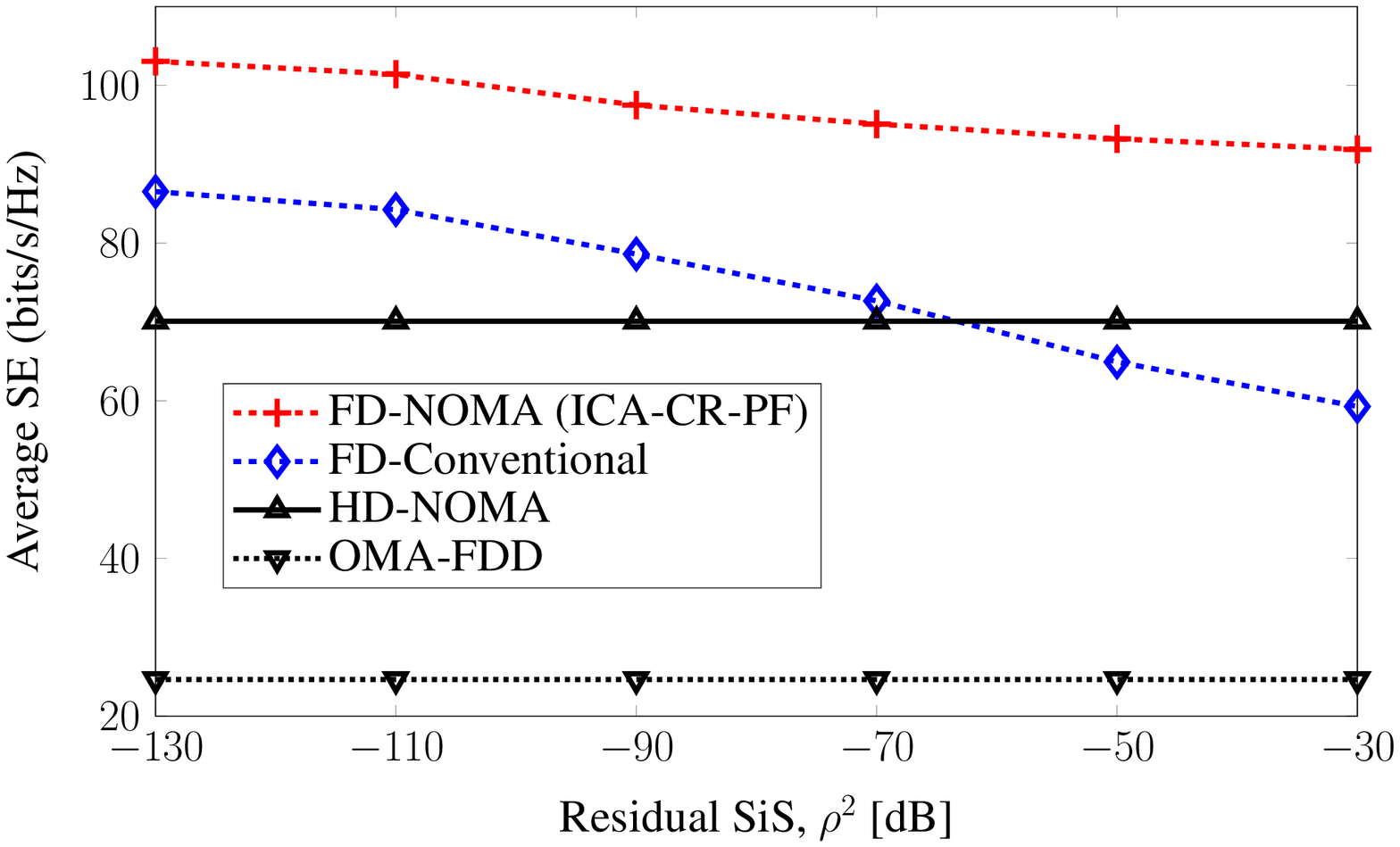}
			\label{fig: SumRate vs rho b}
		}
	\end{subfigure}
	\vspace{-10pt}
	\caption{\hili The effects of SI on DURR and SE.}
	\label{fig: SumRate vs rho}
\end{figure}

In Fig. \ref{fig: SumRate vs rho}, we show the effects of SI on the system performance by varying $ \rho^2 $ from $ -130 $ dB to $ -30 $ dB. For this purpose, we also consider  the orthogonal multiple access with frequency division duplexing (OMA-FDD) scheme. In particular, the entire system bandwidth  is partitioned into $ (L+M) $ orthogonal sub-bandwidths, and each sub-bandwidth is allocated to at most one user to avoid interference. Fig. \ref{fig: SumRate vs rho}(a) plots the DL-to-UL rate ratio (DURR) with respect to $ \rho^2 $. As expected, DURRs of HD-NOMA and OMA-FDD are unchanged with varying $ \rho^2 $ since there is no SI on these schemes. We also observe that for FD schemes, DURRs first decrease and then increase when $ \rho^2 $ increases. This reveals an interesting result which can be explained as follows. For very small $ \rho^2 $,  BS pays more attention to DL users to maximize the total SE since the effect of SI is negligible. For moderate $ \rho^2 $, BS aims to balance the performance between DL and UL by scaling down its transmit power to mitigate the effect of SI. When $ \rho^2 $ becomes more stringent, BS sacrifices the performance of UL by scaling up  its transmit power to boost that of DL, as long as the total SE is maximized. Such phenomena also confirm the performance degradation of FD schemes with $ \rho^2 $ in Fig. \ref{fig: SumRate vs rho}(b). Interestingly, the SE of the proposed FD-NOMA is quite robust to the SI and always outperforms the HD-NOMA for a given range of $ \rho^2 $, which further confirms the benefit of optimizing UL users' decoding order. Moreover,  the proposed FD-NOMA provides much better SE than the OMA-FDD   at $ \rho^2=-30 $ dB due to its potential to improve both SE and edge throughput.

To evaluate the effect of CCI on system performance, we examine a simulation setup as illustrated in Fig. \ref{fig: CCI layout}, in which the centered-BS equipped with 4 antennas serves 4 DL users with fixed locations and 2 UL users with fixed distances to the BS simultaneously moving along with a rotation angle, $ \varphi_{\mathtt{r}} $. The group of DL users on the right-hand side (RHS group) includes $ \DLUi{11} $ and $ \DLUi{21} $, while the group of DL users on the left-hand side (LHS group) is formed by $ \DLUi{12} $ and $ \DLUi{22} $. To better illustrate the effect of CCI, we define the sum SE of DL users in the RHS and LHS groups as  $ R_{11+21}^{\dl} $ and $ R_{12+22}^{\dl} $, respectively. Fig. \ref{fig: CCI RHS ratio} depicts $ R_{11+21}^{\dl} $ and its ratio over $ R_{11+21}^{\dl}+R_{12+22}^{\dl} $. As expected, $ R_{11+21}^{\dl} $ increases when UL users move far away from DL users in the RHS group, i.e., $ \varphi_{\mathtt{r}}\in\{0,\pi/4,\pi/2\} $. When $ \varphi_{\mathtt{r}}\in\{3\pi/4,\pi\} $, $ R_{11+21}^{\dl} $ decreases due to the strong effect of CCI on DL users in the LHS group. The reason is that the BS needs to allocate more power to DL users in the LHS group to combat the strong CCI, leading to a lower power for DL users in the RHS group. The medium ratio of $ R_{11+21}^{\dl} $ to $ R_{11+21}^{\dl}+R_{12+22}^{\dl} $ w.r.t. $ \varphi_{\mathtt{r}} $  on the right y-axis implies that the proposed FD-NOMA offers uniform service to DL users. Accordingly, it is expected to outperform other FD schemes in maximizing the total SE, as seen from results on the left y-axis. Fig. \ref{fig: CCI RHS RateThreshold} further examines $ R_{11+21}^{\dl} $ versus $\varphi_{\mathtt{r}}$ with different values of $\bar{R} \in \{0,1,2,4\}$. Clearly, $ R_{11+21}^{\dl} $ deteriorates with $\bar{R}$, leading to the infeasibility of the proposed algorithm under the strong effect of CCI and higher required rate, i.e., $\varphi_{\mathtt{r}}\in\{0,\pi\}$ and $\bar{R}=4$ bits/s/Hz.

\begin{figure}[t]
	\centering
	\vspace{-20pt}
	\begin{subfigure}[Simulation setup considered in Figs. \ref{fig: CCI RHS ratio} and \ref{fig: CCI RHS RateThreshold}, with $ r_1  = 25$m, $ r_2 = 75 $m, $ \Delta r = 1 $m, $ N = 4 $, $ K=2 $ and $ L = 2 $.]
		{	
		\includegraphics[width=.28\columnwidth, trim={0cm, 0cm, 0cm, 0cm}]{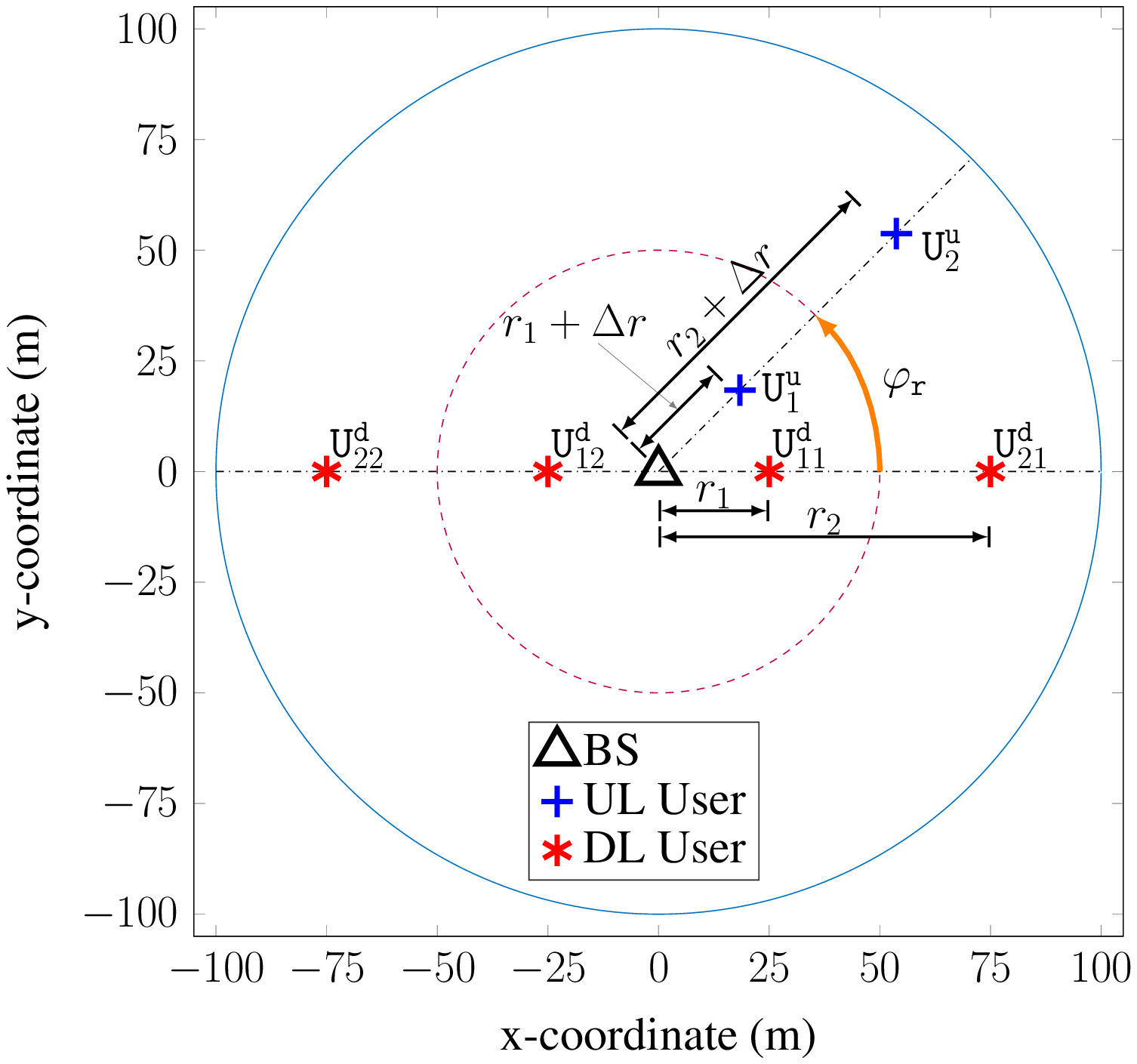}
			\vspace{-50pt}
			\label{fig: CCI layout}
		}
	\end{subfigure}
	\begin{subfigure}[$ R_{11+21}^{\dl} $ and its ratio over $ R_{11+21}^{\dl}+R_{12+22}^{\dl} $.]
		{
			\includegraphics[width=.3\columnwidth, trim={0cm, 0cm, 0cm, 0cm}]{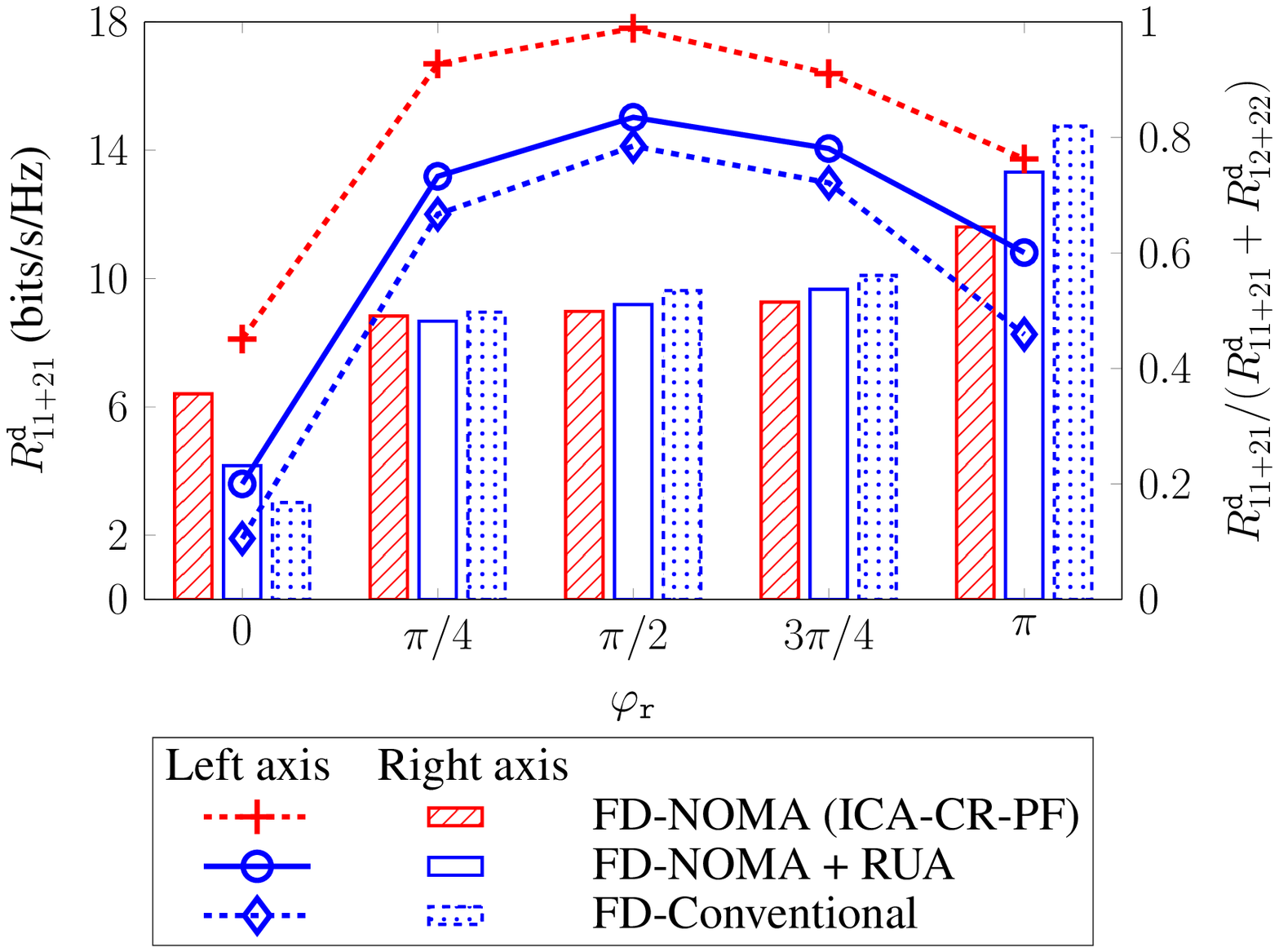}
			\vspace{-50pt}
			\label{fig: CCI RHS ratio}
		}
	\end{subfigure}
	\begin{subfigure}
		[$ R_{11+21}^{\dl} $ with different rate thresholds using FD-NOMA (ICA-CRP-PF), $\bar{R}$.]
		{
		\includegraphics[width=0.3\columnwidth, trim={0cm, 0cm, 0cm, 0cm}]{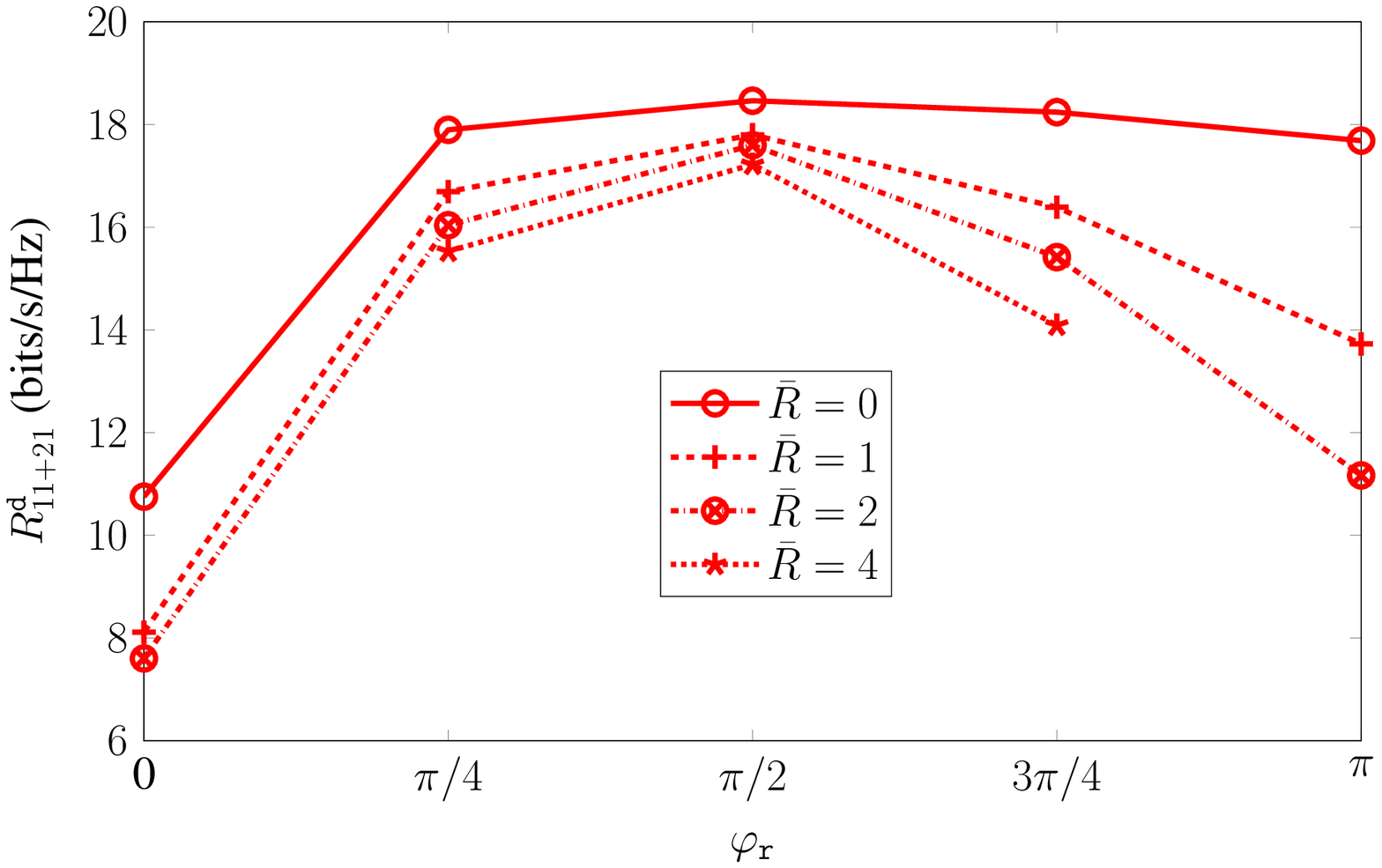}
			\vspace{-10pt}
			\label{fig: CCI RHS RateThreshold}
		}
	\end{subfigure}
	\vspace{-10pt} 
	\caption{\hili The effects of CCI and rate threshold, with the simulation setup in Fig. \ref{fig: CCI layout}.}
	\label{fig: }
\end{figure}
}

\subsection{Convergence Behavior}

\begin{figure}[t]
	\centering
	\vspace{-40pt}
	\centering
	\begin{subfigure}[Convergence speed of ICA-CR-PF with  different values of $ \varrho $.]
		{
			\includegraphics[width=.4\columnwidth, trim={0cm, 0cm, 0cm, 0cm}]{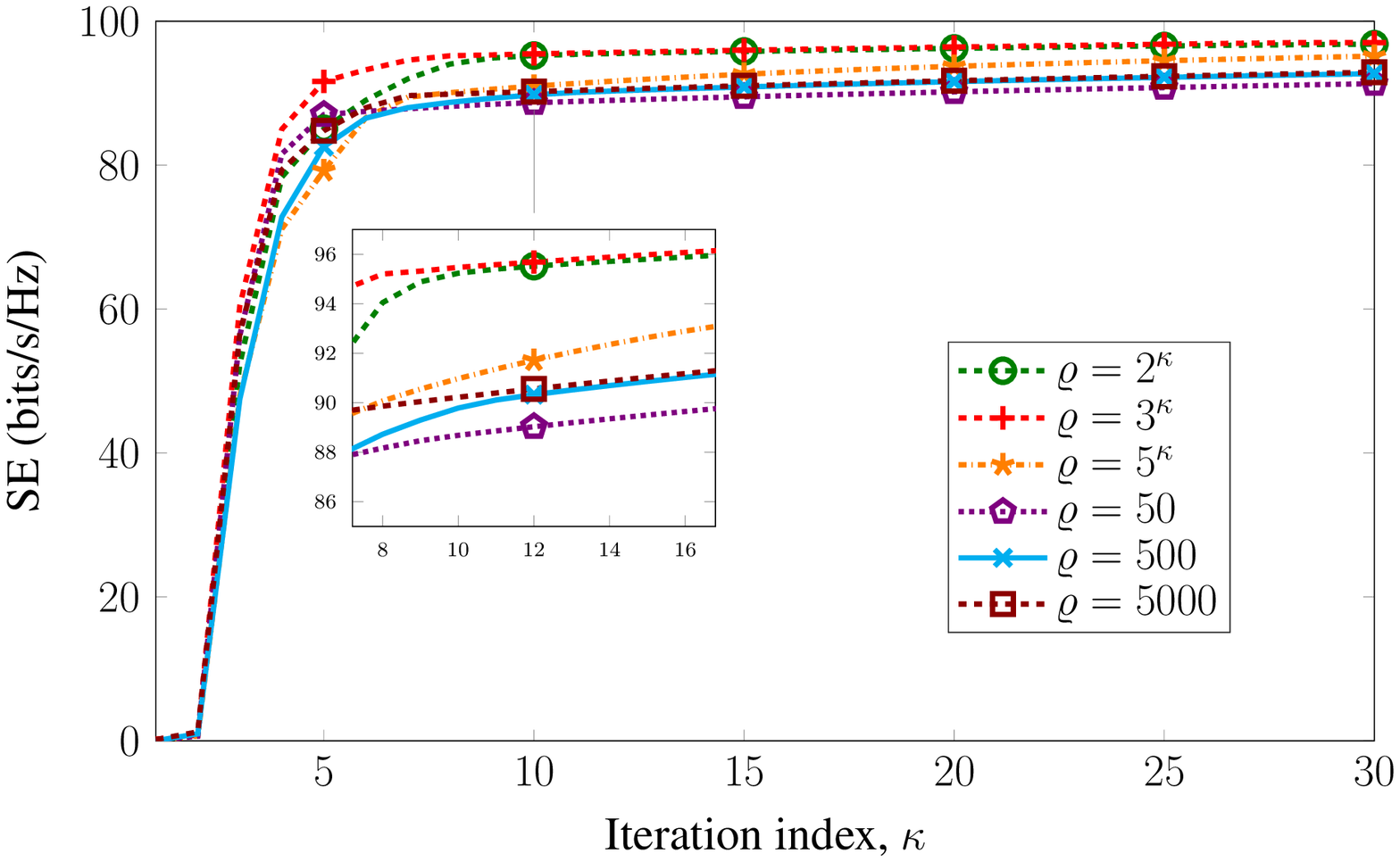}
			\label{fig: Convergence - penalt. para.}
		}
	\end{subfigure}
	\hfill
	\begin{subfigure}[Convergence speed for different schemes.]
		{
			\includegraphics[width=.4\columnwidth, trim={0cm, 0cm, 0cm, 0cm}]{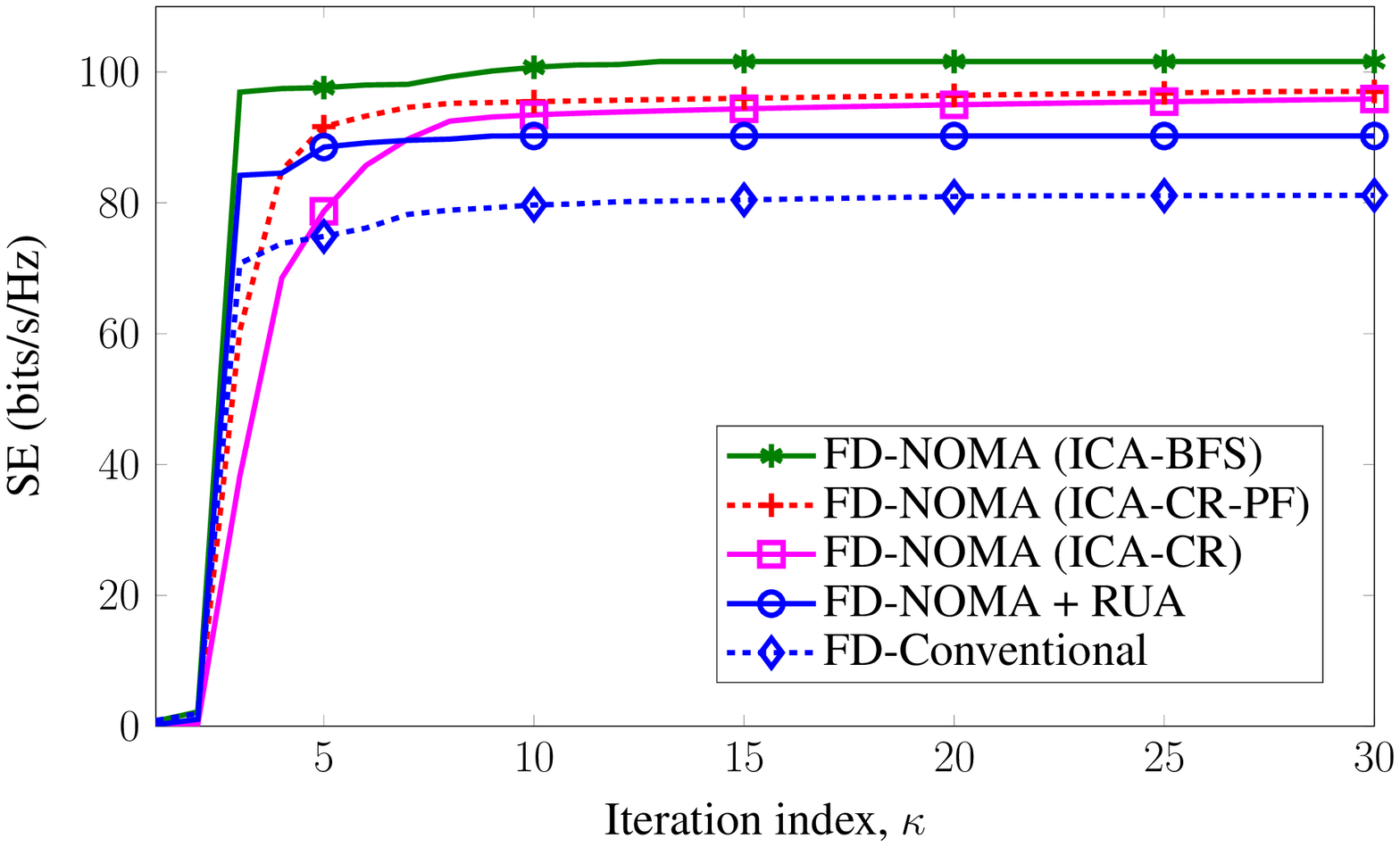}
			\label{fig: Convergence - 5 schemes}
		}
	\end{subfigure}
	\vspace{-10pt} 
	\caption{Typical convergence behavior of the proposed algorithms with one random channel realization.}
	\vspace{-10pt}
	\label{fig:conver}
\end{figure}

In Fig. \ref{fig:conver}, we explore the convergent properties of the proposed algorithms. First, the effect of the penalty parameter $\varrho $ (i.e., step 4 of Alg. \ref{alg: Continuous relaxation problem - pen}) on the convergence behavior and performance of the ICA-CR-PF algorithm is  investigated in Fig. \ref{fig: Convergence - penalt. para.}. {\hilidra Herein, $\varrho $ is numerically examined in two cases: given values as $\varrho\in\{50,500,5000\} $  and  adaptive values per iteration as $\varrho=a^{\kappa}$, with $ a\in\{2,3,5\} $. In the first case, the given values of $\varrho$, which are larger than $ \mfrac{1}{(L+2K)\epsilon(1-\epsilon)}\ln\bigl(\mfrac{P_{\mathtt{bs}}^{\max}}{B.\sigma^2}\bigr) $ (i.e., eq. \eqref{eq: weight of pen. func.}) are sufficiently estimated according to $ \epsilon=\{10^{-1},10^{-2},10^{-3}\} $, respectively. It is seen that for $\varrho=\{500,5000\} $, the SEs at the 10-th iteration reach more than 90\% of that at the 50-th iteration. Although the proposed  ICA-CR-PF algorithm with $\varrho=50 $ provides the same performance ratio after the 5-th iteration, the achievable SE  is worse than that for $\varrho=\{500,5000\} $. In the second case, an increase in $\varrho $ per iteration makes the algorithm converge faster. The results clearly show that the ICA-CR-PF algorithm with $ \varrho=3^{\kappa} $ outperforms the others. Fig. \ref{fig: Convergence - 5 schemes} illustrates the convergence rates of five FD schemes, in which the ICA-CR-PF algorithm uses $ \varrho=3^{\kappa} $. We  exclude the convergence behavior of the HD-NOMA scheme, since it separates the optimization for DL and UL transmissions. As seen, the proposed algorithms provide better performance compared to the others. Fig. \ref{fig: Convergence - 5 schemes} also shows that, compared to the ICA-CR based algorithm, the ICA-CR-PF based algorithm  converges much faster. This can be attributed to the fact that the absence of the PF in the ICA-CR based algorithm may take more iterations to stabilize.}

{\hilidra Finally, we provide further insight of selection of  $\varrho $ in the ICA-CR-PF based algorithm, as illustrated in Fig. \ref{fig: Norm alpha beta - pen par}. As earlier shown, $ \varrho=a^{\kappa} $ provides faster convergence. However, the best value of $ a $ mainly depends on the specific setting. Therefore, the ICA-CR-PF based algorithm combined with the binary search is used to find $ a $ just once. To evaluate the effectiveness of $ a $, we define the convergence measurements as $ \mathbf{u}\triangleq\mathrm{vec}\bigl(\bigl[[\alpha_{kj}^2-\alpha_{kj}]_{k,j\in\mathcal{K}}\;[\beta_{\ell m}^2-\beta_{\ell m}]_{\ell,m\in\mathcal{L}}\bigr]\bigr) $ and $ \mathbf{\hat{f}}_p\triangleq0.1\mathrm{vec}\bigl(\bigl[[f_p(\alpha_{kj})]_{k,j\in\mathcal{K}}\;[f_p(\beta_{\ell m})]_{\ell,m\in\mathcal{L}}\bigr]\bigr) $, where $ \mathrm{vec}(\mathbf{X}) $ represents the vectorization of the matrix $ \mathbf{X} $. \textbf{Remark \ref{rem: pen. par. selection}} indicates that $ a $ is selected such that $ \varrho $ should not be too large, to satisfy the condition in \eqref{eq: weight of pen. func.}. Numerically, the convergence rate of $ \|\mathbf{u}\|_{\infty} $ and $ \|\mathbf{\hat{f}}_p\|_{\infty} $ is the best when $ a $ is large enough, and however, it becomes worse when $ a $ further increases. In implementation, the binary search is used for finding $ a \in [2,5] $, and for each value of $ a $, the ICA-CR-PF based algorithm investigates the values of $ \|\mathbf{u}\|_{\infty} $ and $ \|\mathbf{\hat{f}}_p\|_{\infty} $ so that the smallest value of $ a $ providing the lowest convergence rates of $ \|\mathbf{u}\|_{\infty} $ and $ \|\mathbf{\hat{f}}_p\|_{\infty} $ is selected. For example, Fig. \ref{fig: Norm alpha beta - pen par} depicts the values of $ \|\mathbf{u}\|_{\infty} $ and $ \|\mathbf{\hat{f}}_p\|_{\infty} $ (left y-axis) and penalty parameter $ \varrho $ (right y-axis) versus the iteration index $ \kappa $ (common x-axis) in the cases of $ a=\{2,3\} $. For the above setting,  $ a=3 $ in Fig. \ref{fig: Norm alpha beta - pen par 3} is a better choice as it provides the lower convergence rates of $ \|\mathbf{u}\|_{\infty} $ and $ \|\mathbf{\hat{f}}_p\|_{\infty} $. We have also numerically observed that when $ a > 3 $, the quick increase in $ \varrho $ per iteration contradicts \textbf{Remark \ref{rem: pen. par. selection}}. Therefore, the smallest value of $ a $ needs to be found within $ [2,3] $. Whenever the sufficient value of $ a $ is found, the ICA-CR-PF based algorithm can operate under different channel conditions. Remarkably, the convergence behaviors of the UA variables are almost same  at the beginning of setting. It can be seen that the values of $ \boldsymbol{\alpha} $ and $ \boldsymbol{\beta} $ are close to binary at the 9-th iteration. Thus, without loss of optimality, the binary variables can be fixed as in \eqref{eq: rounding binary var.} when $ \|\mathbf{u}\|_{\infty} $ satisfies a given error tolerance, and then, the ICA-CR-PF based algorithm continuously solves problem \eqref{eq: prob. general form - relax. pen. 1} for power control as equivalent to a subproblem  \eqref{eq: prob. BB form convex prog.}.}

\begin{figure}[t] 
	\centering
	\vspace{-20pt}
	\begin{subfigure}[$ \varrho=2^{\kappa} $]
		{
			\centering
			\includegraphics[width=.36\columnwidth, trim={0cm, 0cm, 0cm, 0cm}]{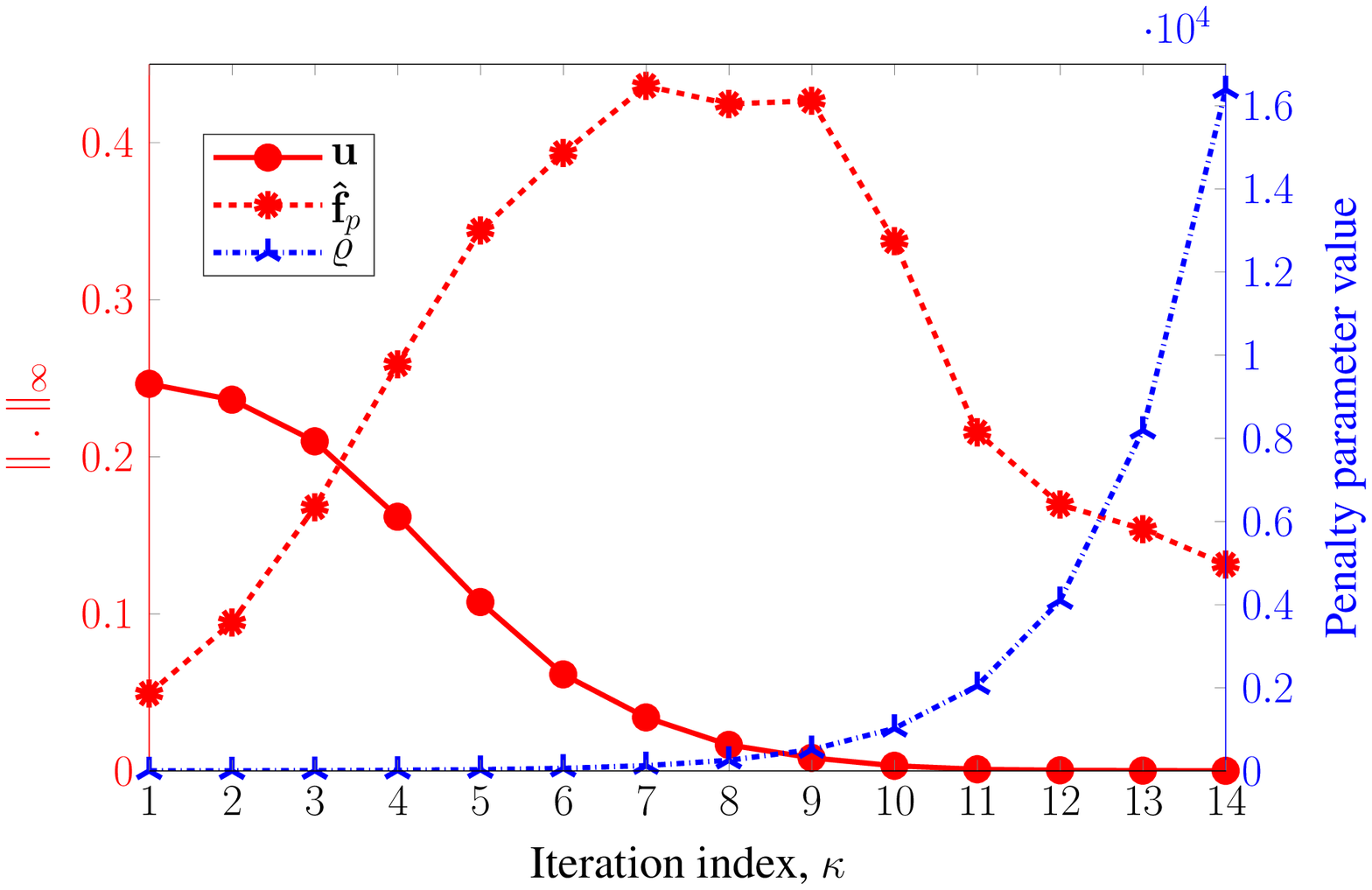}
			\vspace{-40pt}
			\label{fig: Norm alpha beta - pen par 2}
		}
	\end{subfigure}
	\hfill
	\begin{subfigure}[$ \varrho=3^{\kappa} $]
		{
			\centering
			\includegraphics[width=.36\columnwidth, trim={0cm, 0cm, 0cm, 0cm}]{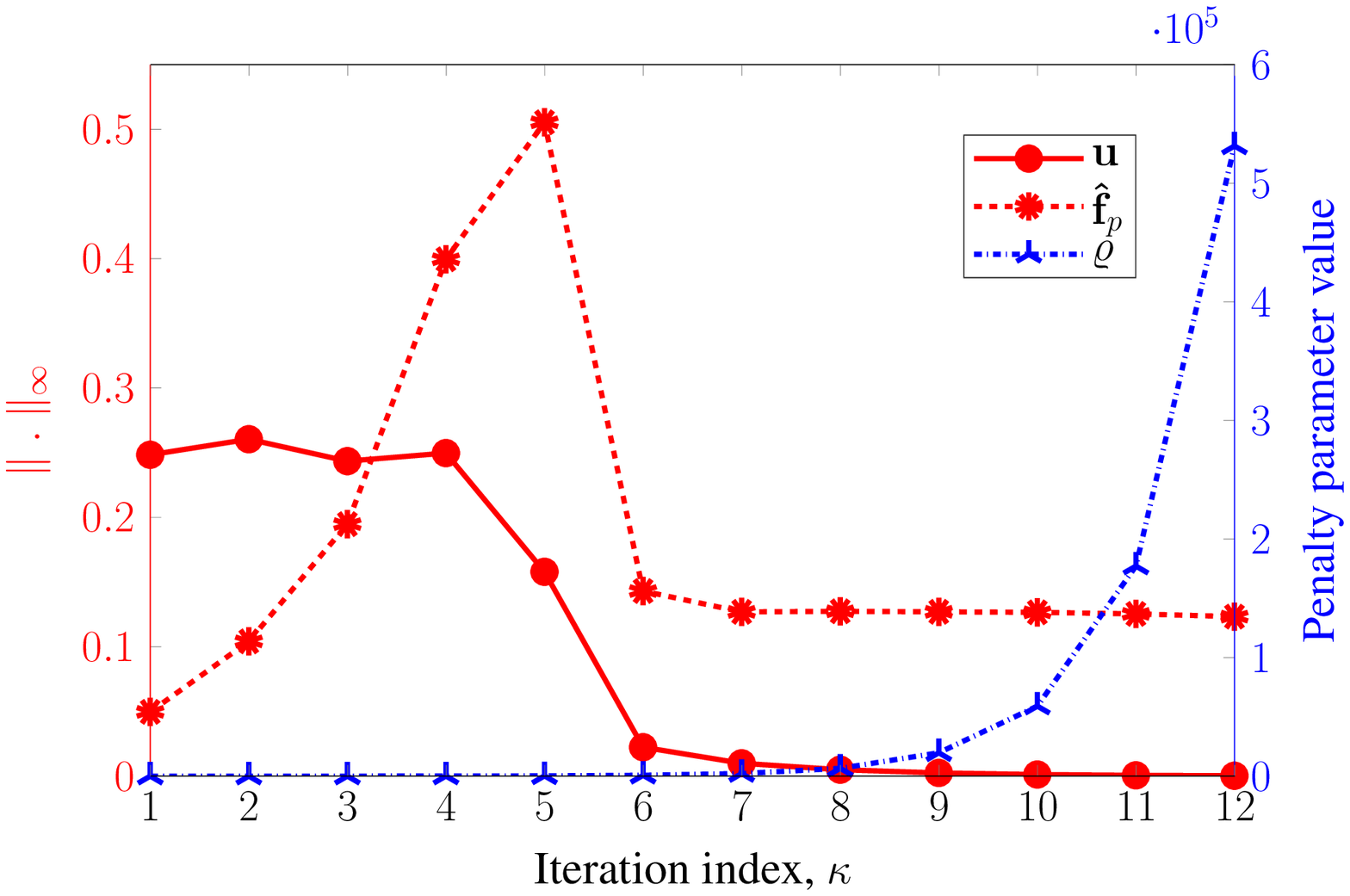}
			\vspace{-40pt}
			\label{fig: Norm alpha beta - pen par 3}
		}
	\end{subfigure}
	\vspace{-10pt} 
	\caption{Convergence rate of UA variables and PF values with  $ \varrho=a^{\kappa} $, for $ a=\{2,3\} $.}
	\label{fig: Norm alpha beta - pen par}
	\vspace{-10pt}
\end{figure}


\section{Conclusion}\label{Conclusion}
In this paper, joint power control and user association problem has been proposed to maximize the total SE of a cellular FD-NOMA system. We have employed a tensor model for DL users and a permutation matrix for UL users to formulate the problem of user association, which significantly reduce  the number
of association variables. By presenting novel methods to approximate the formulated non-convex problem, we have developed two iterative algorithms with low computational complexity. In the first method, the binary  variables are relaxed to be continuous and an iterative algorithm based on the ICA framework has been proposed to solve the resulting non-convex CR problem. In the second method, the uncertainties of binary  variables are further penalized without causing additional complexity as it aims at finding a  high-performance UA solution. Our extensive numerical results suggest that the second approach is more effective in terms of the achievable SE and convergence speed.
In addition, by the brute-force search algorithm, we have transformed the original problem into  subproblems under a given UA based on which the ICA framework has been customized to find an optimal solution. Our proposed iterative algorithms improve achievable SE at each iteration and converge eventually, and are also superior to other known algorithms.  We have also concluded that an appropriate number of zones/clusters with more distinct channel conditions is of important to achieve remarkable gains in  FD-NOMA systems. {\hili The robustness of the proposed method against the significant effects of SI and CCI is also revealed.}

\appendices
\renewcommand{\thesectiondis}[2]{\Alph{section}:}

\section{Proof of Theorem \ref{thm: user association matrix}} \label{app: user association matrix}
\renewcommand{\theequation}{\ref{app: user association matrix}.\arabic{equation}}\setcounter{equation}{0}\vspace{-5pt}
Let $ \mathbf{C}_0 $ be the unitary matrix representing the clustering indices. It is obvious that $ \mathbf{C}_{i},\;i\in\mathcal{Z}, $ is the change-of-basis matrix of zone $ i $, with respect to the basis $ \mathbf{C}_0 $. Therefore, $ \bigl[\mathbf{C}_{i}\bigr]_{kj},\;k,j\in\mathcal{K} $, with $ [\mathbf{X}]_{a,b} $ denoting the element at the $ a $-th row and the $ b $-th column of matrix $ \mathbf{X} $, indicates whether the $ j $-th DL user in zone $ i $ belongs to the $ k $-th cluster. From \textbf{Definition \ref{def: index clusters}}, the matrix of DL user associations between  users in zone $ 1 $ and  clusters is equivalent to the change-of-basis matrix $ \mathbf{C}_1 $, i.e., $ \mathbf{C}_1=\mathbf{C}_0=\mathbf{I}_K $. As a result, a user association matrix $ \mathbf{T}^{1i}$ is equivalent to the change-of-basis matrix $ \mathbf{C}_{i} $ w.r.t. the basis $ \mathbf{C}_0 $. From the transformation law of tensor, the user association matrix $ \mathbf{T}^{iz},\;(i,z)\in\{\mathcal{Z}\times\mathcal{Z}\} $ is calculated by
\begin{align}\label{eq:A1}
\mathbf{T}^{iz} = \bigl(\mathbf{T}^{1i}\bigr)^{-1}\mathbf{C}_{1}\mathbf{C}_{z} =\mathbf{C}_{i}^{-1}\mathbf{C}_{z}.
\end{align}
Based on \textbf{Definition \ref{def: tensor for assignment}}, it is realized that each DL user in a certain zone is assigned to exactly one cluster, and two arbitrary DL users in each cluster come from two different zones. Therefore, the matrix $ \mathbf{C}_i,\;i\in\mathcal{Z} $, satisfies the following conditions:
\begin{align}
\sum\nolimits_{k\in\mathcal{K}}\bigl[\mathbf{C}_i\bigr]_{kj}=1\ \text{and}\ \sum\nolimits_{j\in\mathcal{K}}\bigl[\mathbf{C}_i\bigr]_{kj}=1.
\end{align}
Accordingly, $ \mathbf{C}_{i} $ characterized as a permutation matrix, satisfies the property that $ \mathbf{C}_{i}^{-1}=\mathbf{C}_{i}^T $. Equation \eqref{eq: user association matrix} is then obtained by substituting  $ \mathbf{C}_{i}^{-1}=\mathbf{C}_{i}^T $ into \eqref{eq:A1}.

\section{Proof of Lemma \ref{lem: UB xy2}} \label{app: UB xy2}
\renewcommand{\theequation}{\ref{app: UB xy2}.\arabic{equation}}\setcounter{equation}{0}\vspace{-7pt}
By imposing an SOC constraint $ y^2\leq z,\; z>0 $, the function $ h(x,y) $ is upper bounded by
\begin{align} \label{eq: hxy function}
h(x,y) \leq xz := \tilde{h}(x,z).
\end{align}
Due to the concavity of function $ \sqrt{uv} $,
we make use of the following inequality \cite{Beck:JGO:10, Dinh:JSAC:18}:
\begin{align} \label{eq: concave func. sqrt uv}
\sqrt{uv}\leq\mfrac{\sqrt{v^{(\kappa)}}}{2\sqrt{u^{(\kappa)}}}u+\mfrac{\sqrt{u^{(\kappa)}}}{2\sqrt{v^{(\kappa)}}}v,
\end{align}
for any $ u, v, u^{(\kappa)}, v^{(\kappa)} >0 $, where $ u^{(\kappa)} $ and $ v^{(\kappa)} $ are the known neighborhoods of $ u $ and $ v $, respectively. By letting $ x=\sqrt{u},\;z=\sqrt{v} $, and substituting \eqref{eq: concave func. sqrt uv} into \eqref{eq: hxy function}, we arrive at \eqref{eq: UB xy2}.

\section{Proof of Theorem \ref{thm: relax. prob with pen.}} \label{app: relax. prob with pen.}
\renewcommand{\theequation}{\ref{app: relax. prob with pen.}.\arabic{equation}}\setcounter{equation}{0}\vspace{-7pt}
Firstly, we can treat constraint \eqref{eq: prob. mixed-integer theorem c} as $ \alpha_n(\alpha_n-1)=0,\;\forall n\in\mathcal{N}$, which is equivalent to
\begin{subequations} \label{eq: alpha relax. form}
	\begin{gather}
	\alpha_n(\alpha_n-1) \geq 0, \label{eq: alpha relax. form a}\\
	\alpha_n(\alpha_n-1) \leq 0 \label{eq: alpha relax. form b}.
	\end{gather}
\end{subequations}
For the non-convex constraint \eqref{eq: alpha relax. form a}, we apply the relaxation approach to release \eqref{eq: alpha relax. form a}, where  $ \alpha_n $ is constrained by a box $ [0,1] $, i.e., $ \alpha_n^2\leq\alpha_n $ for $ 0\leq\alpha_n\leq 1 $. Inspired from \cite{Sun:TCOMM:Mar2017}, we then introduce an additional PF, denoted by $ f_p(\alpha_n) $, satisfying:
\begin{align} \label{eq: pen. func. cond.}
f_p(\alpha_n) = \left\{ 
\begin{aligned}
-\infty, & \quad \text{ if } 0<\alpha_n<1, \\
0, & \quad \text{ if } \alpha_n\in\{0,1\}.
\end{aligned}
\right.
\end{align}
This indicates that the objective value for the maximization problem becomes $ -\infty $ when constraint \eqref{eq: prob. mixed-integer theorem c} is violated. Numerically, the objective value is corrupted by a large penalty parameter, such that it becomes smaller than an estimated optimal value $ \hat{f}_0(\mathbf{x}, \boldsymbol{\alpha}) $ when $ \alpha_n $ is not binary. By exploiting the convexity of \eqref{eq: alpha relax. form b}, the PF can be constructed as $ f_p(\alpha_n)\triangleq\varrho_n(\alpha_n^2-\alpha_n), \; \forall n\in\mathcal{N} $, to meet the condition in \eqref{eq: pen. func. cond.}. Herein, $ \alpha_n^2-\alpha_n\leq0 $ provides a well-defined set mapping into the codomain of $ f_p $ in \eqref{eq: pen. func. cond.} for $ \alpha_n\in[0,1] $, while $ \varrho_n>0 $ is selected to be large enough such that the optimality of \eqref{eq: prob. mixed-integer theorem} holds. By this way,  the relaxation problem  \eqref{eq: prob. relax theorem} is provided.

Secondly, we need to demonstrate how $ \varrho_n>0 $ is found to ensure the optimal solution of \eqref{eq: prob. mixed-integer theorem}  by solving the relaxation problem with PF in \eqref{eq: prob. relax theorem}. Let $ \mathcal{F}=\mathcal{X}\times\{0,1\}^{N_{\alpha}\times1} $ and $ \mathcal{F}_R=\mathcal{X}\times[0,1]^{N_{\alpha}\times1} $ be the feasible regions of \eqref{eq: prob. mixed-integer theorem} and \eqref{eq: prob. relax theorem}, respectively, and suppose that $ \mathbf{y}^*\in\mathcal{F} $ and $ \mathbf{y}_{r}^*\in\mathcal{F}_R $ are their optimal solutions. If $ \mathbf{y}^*\in\mathcal{F}\subset\mathcal{F}_R $, $ \mathbf{y}^*\in\mathcal{F}_R $ provides a lower bound of $ f_0(\mathbf{y}_{r}^*) $, i.e., $ f_0(\mathbf{y}^*)\leq f_0(\mathbf{y}_{r}^*) $. To find $ \varrho_n $, we consider a part of feasible region of the relaxation problem as $ \bar{\mathcal{F}} \triangleq \{\mathbf{y}_r=(\mathbf{x},\boldsymbol{\alpha})\in\mathcal{F}_R\backslash\mathcal{F}\;|\epsilon\leq\alpha_n\leq1-\epsilon, \forall n\in\mathcal{N} \} $, with $ \epsilon\in(0,1) $ being an arbitrarily small number such that $ \mathcal{F} \cong \mathcal{F}_R\backslash \bar{\mathcal{F}} $. If any $ \mathbf{\hat{y}}^*\in\mathcal{F}_R\backslash \bar{\mathcal{F}} $, the optimal solution of \eqref{eq: prob. relax theorem} is close to that of \eqref{eq: prob. mixed-integer theorem}, i.e., $ \mathbf{\hat{y}}^*\rightarrow\mathbf{y}^* $. The penalty parameter $ \varrho_n>0 $ must be selected to satisfy
\begin{align} \label{eq: optimality cond. for relax. prob.}
\underset{\mathbf{\bar{y}}^*\in\bar{\mathcal{F}}}{\sup}\bigl(f_0(\mathbf{\bar{y}}^*)+\sum\nolimits_{n\in\mathcal{N}}f_p(\alpha_n)\bigr)=\sum\nolimits_{n\in\mathcal{N}}f_p(\alpha_n)+\underset{\mathbf{\bar{y}}^*\in\bar{\mathcal{F}}}{\sup}\;f_0(\mathbf{\bar{y}}^*)\leq \underset{\mathbf{y}^*\in\mathcal{F}}{\inf}\; f_0(\mathbf{y}^*).
\end{align}
On the other hand, it is clear that $ \alpha_n(\alpha_n-1)\leq\epsilon(\epsilon-1)$ due to $\alpha_n\in[\epsilon,1-\epsilon] $ and  $ \bar{\mathcal{F}}\subseteq\mathcal{F}_R $. Therefore, $ \sum_{n\in\mathcal{N}}f_p(\alpha_n) $ and $ \underset{\mathbf{\bar{y}}^*\in\bar{\mathcal{F}}}{\sup}\;f_0(\mathbf{\bar{y}}^*) $ are respectively upper bounded as 
\begingroup
\allowdisplaybreaks
\begin{subequations} \label{eq: pen. func. upper bound}
	\begin{align} 
	\sum\nolimits_{n\in\mathcal{N}}f_p(\alpha_n) & \leq\sum\nolimits_{n\in\mathcal{N}}\varrho\epsilon(\epsilon-1)=N_{\alpha}\varrho\epsilon(\epsilon-1)<0, \\
	\underset{\mathbf{\bar{y}}^*\in\bar{\mathcal{F}}}{\sup}\;f_0(\mathbf{\bar{y}}^*) & \leq \underset{\mathbf{y}_{r}^*\in\mathcal{F}_R}{\sup}\;f_0(\mathbf{y}_{r}^*),
	\end{align}
\end{subequations}\endgroup
where $ \varrho\triangleq\max\{\varrho_n\}_{n\in\mathcal{N}} $. The inequality \eqref{eq: optimality cond. for relax. prob.} strongly holds with \eqref{eq: pen. func. upper bound}, and then we obtain
\begin{align} \label{eq: weight of pen. func.}
\varrho \geq\mfrac{1}{N_{\alpha}\epsilon(\epsilon-1)} \Bigl(\underset{\mathbf{y}^*\in\mathcal{F}}{\inf}\;f_0(\mathbf{y}^*)-\underset{\mathbf{y}_{r}^*\in\mathcal{F}_R}{\sup}\; f_0(\mathbf{y}_r^*)\Bigr).
\end{align}
Since the function $ f_0 $ is closed and bounded on $ \mathcal{F} $, it is also closed and bounded on $ \mathcal{F}_R $ for $ \boldsymbol{\alpha} \in [0,1] $. Moreover, $ \mathcal{F}\subset\mathcal{F}_R $ indicates that $ \underset{\mathbf{y}_{r}^*\in\mathcal{F}_R}{\inf}\; f_0(\mathbf{y}_r^*)\leq\underset{\mathbf{y}^*\in\mathcal{F}}{\inf}\;f_0(\mathbf{y}^*)<\underset{\mathbf{y}^*\in\mathcal{F}}{\sup}\;f_0(\mathbf{y}^*)\leq\underset{\mathbf{y}_{r}^*\in\mathcal{F}_R}{\sup}\; f_0(\mathbf{y}_r^*) $, resulting in the positive value for the RHS of \eqref{eq: weight of pen. func.}. The inequality \eqref{eq: weight of pen. func.} also means that there exists $ \varrho>0 $ so  that the penalty parameter satisfies the condition \eqref{eq: optimality cond. for relax. prob.}. Note that the smaller value of $ \epsilon $ in \eqref{eq: weight of pen. func.} results in a larger value of its RHS. Therefore, $ \varrho $ needs to be selected such that it is adapted to $ \epsilon $.

Finally, we consider the property of the feasible set. Assuming $ \mathcal{X} $ is a compact convex set. In  problem \eqref{eq: prob. relax theorem}, we have $ \alpha_n\in \mathcal{I}\triangleq[0,1],\;\forall n\in\mathcal{N} $.  Since $ \mathcal{I} $ is a closed interval, $ \mathcal{I} $ is compact \cite{Munkres:Topo}. Considering $ \delta=\theta\alpha+(1-\theta)\beta $ for $ \forall \alpha, \beta \in  \mathcal{I},\; 0\leq\theta\leq 1 $, we have $ 0\leq\theta\alpha\leq \theta $ and $ 0\leq(1-\theta)\beta\leq(1-\theta) $. This means that $ 0\leq\delta\leq1 $, or $ \delta\in\mathcal{I} $, leading to the fact that $ \mathcal{I} $ is convex. Therefore, the feasible set of \eqref{eq: prob. relax theorem} is $ \mathcal{X}\times\mathcal{I} $, which is a compact convex set.  Thus, the proof is completed.

\begingroup
\setstretch{0.88}
\bibliographystyle{IEEEtran}
\bibliography{IEEEfull}
\endgroup

\end{document}